%% file: main.tex
\documentclass[11pt, a4paper]{article}
\usepackage[utf8]{inputenc}
\usepackage{amsfonts}
\usepackage[margin=1in]{geometry}
\usepackage{hyperref,amsthm,enumerate,
            imakeidx, xparse, mathtools,
            xcolor, color, colortbl, 
            amssymb, tensor, cite,
            soul, graphicx ,titlesec,  appendix, tikz,
            amsmath,scalerel, comment, float, multirow, multicol, caption, subcaption} 
\usepackage[most]{tcolorbox}
\usepackage{stackengine}
\usetikzlibrary{arrows,decorations.markings}
\usetikzlibrary{shapes.misc}
\usetikzlibrary{arrows.meta}
\usetikzlibrary{angles,quotes}
\usepackage{tkz-euclide}
\usetikzlibrary{intersections}
\usetikzlibrary{calc}
\usetikzlibrary{backgrounds}
\numberwithin{equation}{subsection}
\hypersetup{
pdfstartview = {FitH},
}
\hypersetup{	colorlinks=true,         	linkcolor=darkblue,        	citecolor=blue,       	urlcolor=blue            }

\newcommand{\be}{\begin{equation}}
\newcommand{\ee}{\end{equation}}

\newtheorem{proposition}{Proposition}

\newcommand{\cA}{{\cal A}}
\newcommand{\cC}{{\cal C}}
\newcommand{\cD}{{\cal D}}

\newcommand{\cL}{{\cal L}}
\newcommand{\cH}{{\cal H}}

\newcommand{\cV}{{\cal V}}
\newcommand{\cS}{{\cal S}}

\newcommand{\cQ}{{\cal Q}}
\newcommand{\cU}{{\cal U}}

\newcommand{\cK}{{\cal K}}

\def\so{\mathfrak{so}}
\def\sl{\mathfrak{sl}}

\def\su{\mathfrak{su}}

\def\bbR{\mathbb{R}}
\def\bbC{\mathbb{C}}

\def\bbZ{\mathbb{Z}}
\def\bbN{\mathbb{N}}

\def\SL{\mathrm{SL}}
\def\SO{\mathrm{SO}}

\def\dd{\mathrm{d}}
\def\s{\mathrm{s}}

\def\i{\mathrm{i}}
\def\M{\mathrm{M}}

\def\dh{\mathrm{d} h}
\def\dg{\mathrm{d} g}
\def\dx{\mathrm{d} x}

\def\tn{\tilde n}

\def\on{\bar n}

\def\la{\langle}
\def\ra{\rangle}
\def\ot{\otimes}

\def\mone{^{-1}}

\def\adag{a^{\dagger}}

\definecolor{darkblue}{rgb}{0.0, 0.3, 0.65}
\definecolor{darkred}{rgb}{0.65, 0.0, 0.0}
\definecolor{darkgreen}{rgb}{0.0, 0.65, 0.0}

\definecolor{Darkblue}{rgb}{0.0, 0.0, 0.3}
\definecolor{Darkred}{rgb}{0.3, 0.0, 0.0}
\definecolor{Darkgreen}{rgb}{0.0, 0.3, 0.0}

\definecolor{lightred}{rgb}{1.0,0.65,0.65}

\definecolor{deepblue}{rgb}{0.0,0.8,0.8}
\definecolor{sky}{rgb}{0.0,0.75,1}
\definecolor{cargreen}{rgb}{0.0, 0.8, 0.6}
\definecolor{aquamarine}{rgb}{0.5,1.0,0.83}
\definecolor{darklavander}{rgb}{0.8,0.8,1.0}
\definecolor{lavander}{rgb}{0.9,0.9,0.98}
\definecolor{lightlavander}{rgb}{0.95,0.95,0.99}
\definecolor{mauve}{rgb}{0.86,0.82,1.0}
\definecolor{platinum}{rgb}{0.9,0.89,0.89}

\definecolor{eigengrau}{RGB}{22,22,29}

\usepackage{authblk}

\makeatletter
\renewenvironment{abstract}{%
    \if@twocolumn
      \section*{\abstractname}%
    \else 
      \begin{center}%
        {\bfseries\sffamily\abstractname\vspace{\z@}}
      \end{center}%
      \quotation
    \fi}
    {\if@twocolumn\else\endquotation\fi}
\makeatother

\title{A state sum for four-dimensional Lorentzian quantum geometry in terms of edge vectors}

\begin{document}

\author[1]{\sffamily Roukaya Dekhil\thanks{roukaya.dekhil@unifi.it}}
\author[2,3]{\sffamily Matteo Laudonio\thanks{matteo.laudonio.6@gmail.com}\thanks{matteo.laudonio@u-bordeaux.fr}}
\author[4,5]{\sffamily Daniele Oriti\thanks{doriti@ucm.es}}

\affil[1]{\small Universit\'a degli Studi di Firenze,\\ Piazza di San Marco, 4, 50121 Firenze FI, Italy, EU}
\affil[2]{\small Univ. Bordeaux, LABRI, 351 Cours de la Libération, 33400 Talence, France}
\affil[3]{\small Department of Applied Mathematics, University of Waterloo, 200 University Avenue West, Waterloo, Ontario, Canada, N2L 3G1}
\affil[4]{\small Depto. de 
Facultad de Ciencias 
Universidad Complutense de Madrid, Plaza de las Ciencias 1, 28040 Madrid, Spain, EU}
\affil[5]{\small  Munich Center for Quantum Science and Technology (MCQST), Schellingstrasse 4, 80799 Munich, Germany}

\date{}
\maketitle
\begin{abstract}
We present the construction of a new state sum model for $4d$ Lorentzian quantum gravity based on the description of quantum simplicial geometry in terms of edge vectors. Quantum states and amplitudes for simplicial geometry are built from irreducible representations of the translation group, then related to the representations of the Lorentz group via expansors, leading to interesting (and intricate) non-commutative structures. We also show how the new model connects to the Lorentzian Barrett-Crane spin foam model, formulated in terms of quantized triangle bivectors. \end{abstract}
\newpage
\tableofcontents
\newpage
\section{Introduction}
\label{intro}
Several strategies for the construction of a non-perturbative and background independent theory of quantum gravity, including canonical loop quantum gravity (LQG) \cite{Rovelli:2004LQGBook, Thiemann:2007LQGBook, Pereze_qunatumgracity_notes}, topological quantum field theories (TQFT) \cite{Witten:19883dGravityTopological,Barrett:1995QuantumGravityTopological, Crane1993,CraneKauffmanYetter:1994StateSum4d}, lattice approaches \cite{Barrett:2018ybl, Dittrich:2008ar} and group field theory (GFT)\cite{Oriti:2014GFTLQG,Oriti:2011GFT, FreidelGFT, Krajewski:2011GFT}, converge to the formalism of state sum or spin foam models for quantum geometry
\cite{Baez_1998,baezSF, Perez:2013SpinFoam, Engle_2014SF, Oriti_2001SF,roveliSF}. 
The formalism is based on the use of lattice discretizations of spacetime and geometry, which are then quantized, often following the guideline of continuum gravitational actions, but framing quantum geometry in purely algebraic (and combinatorial) terms, encoding it in the representation theory of appropriately chosen Lie groups.  
The most popular spin foam models are constructed as the discrete (and algebraic) quantum counterpart of the formulations of continuum gravity as a constrained BF theory, i.e. the Plebanski and Plebanski-Holst actions, in both Riemannian and Lorentzian contexts. The first led to the Barrett-Crane model \cite{Barrett1998,BaratinOriti:2011GFTBarrettCrane,Jercher:2022mky}, while the second led to the EPRL models \cite{EngleLivinePereiraRovelli:2007EPRL} and the Baratin-Oriti one \cite{BaratinOritimodel} (in the Riemannian case). 
In these models, the geometricity (or \lq\lq simplicity\rq\rq) constraints reducing topological BF theory to gravity are imposed on the discrete counterpart of the B field of BF theory, i.e. bivectors associated with the triangles of a simplicial complex, at the quantum level, upon identification of the same bivectors with Lie algebra elements of the Lorentz group. 
At the simplicial geometry level, these constraints are expected to allow to invert the set of bivectors to a set of edge vectors, representing the discrete counterpart of the tetrad field of continuum gravity, thus ensuring the complete and proper encoding of (quantum) simplicial geometry. This proper encoding should then also be reflected, remaining at the discrete level, in a semiclassical approximation reproducing the Regge calculus dynamics for piecewise-flat geometries, i.e. the direct discrete counterpart of continuum GR.\\
Much of the work on these models in the last decades has been devoted to the analysis of the quantum amplitudes of these spin foam models, aiming at proving the proper encoding of simplicial geometry. This is a non-trivial task, by all means, since the imposition of the geometricity constraints at the quantum level is a complex and inevitably ambiguous procedure, and the resulting quantum amplitudes are highly involved from the analytic point of view.\\ 
The results of this endeavor have been many and important (see the cited reviews and references therein), but not fully conclusive, in several aspects, especially for generic lattices. The connection with quantum metric geometries, thus beyond the semiclassical regime, while still at the discrete level, is less under control, and thus the precise nature of quantum corrections. And the above issues, while crucial for the geometric understanding of spin foam amplitudes in the discrete, have consequences for the continuum approximation of the quantum gravity dynamics, which are obviously key aspects of interest, and on which we need to gain control.  
A conceptually straightforward way to solve the issues faced by spin foam models, in this respect, would be of course to depart from the conventional strategy based on constrained BF formulations and quantum bivectors, and thus on the representation theory of the Lorentz group, and construct directly a state sum model based on the quantization of edge vectors. From the point of view of continuum gravity, this would be the counterpart of a direct quantization of Palatini gravity. Such departure would need to be based on appropriate new analytic tools and a different algebraic framework, which would then have to be connected to the usual one.
A proposal in this direction has been made early on, in fact, by Crane and Yetter \cite{CraneYetter}. They proposed a model based on infinite dimensional irreducible representations (based on \lq expansors\rq) of the Lorentz group acting on vectors in Minkowski space (first introduced by Dirac \cite{Dirac}), with the quantization of Minkowski space realized in terms of representations of the translation group, to be then combined to define the quantization of bivector.  
In this paper, we give an explicit realization of these early suggestions, and put forward a new state sum model of Lorentzian $4d$ quantum geometry based on edge vectors, and on the combined representation theory of the translation and Lorentz groups.\\

The paper is organized as follows. In section \ref{Sec_InfiniteDimReps} we lay down the main elements from the theory of infinite dimensional unitary representations of the Lorentz group that will eventually be associated with the edge vectors as well as the bivectors of the 2-simplex of the new model. We identify the displacement vector on the edge vectors as functions of the translation group on the Minkowski space, and hence we sketch the mathematics behind it, where we relate the Lorentzian harmonic oscillators via their coherent states to the representations of the translation group. We are then equipped with the necessary tools to identify the quantum description of a triangle based on quantum edge vectors as well as the skew-symmetric combination thereof (a bivector). We then move up to the definition of a quantum tetrahedron, in section \ref{Sec_QuantumTetrahedron}, where we also discuss the differences between our description of the quantum tetrahedron with the standard one, which uses quantum bivectors. In section \ref{Sec_SF-EdgeVectors} we employ the obtained quantum states for the quantum tetrahedron as the seeds to write down the new spin foam amplitude. We then show how to relate the new quantum amplitudes to the Barrett-Crane quantization of bivectors, and the resulting 4-simplex amplitude, from quantum edge vectors. Finally, in section \ref{SubSec_GFTEdgeVectors} we complete the definition of the new model of $4d$ Lorentzian quantum geometry, embedding its quantum amplitudes in a GFT formulation, based on elements (and representations) of the translation group.

\section{Relevant representation theory of translation and Lorentz groups}

\label{Sec_InfiniteDimReps}
The group theoretic quantization of edge vectors and their combinations to encode the (piecewise-flat) geometry of a simplicial complex is a central point of our construction. The possibility to quantize simplicial geometry in group theoretic terms, thus obtaining a purely algebraic state sum model associated with simplicial lattices is, more generally, the key fact at the core of the spin foam formalism for quantum gravity \cite{baezSF,roveliSF,Perez:2013SpinFoam} and of the group field theory formalism as well \cite{Oriti:2011GFT,Oriti:2014GFTLQG, BaratinOritimodel,Reisen_Rovelli_Feynman}. 

This group-theoretic quantization of edge vectors is based on the translation group. The precise relation between the representation theory of the translation group and that of the Lorentz group, on the other hand, is the basis for relating our description of quantum geometry, as well as the amplitudes of the new model, to that of the known spin foam models. The latter is based on a description of simplicial geometry in terms of bivectors associated with triangles and their quantization in terms of the representation theory of the Lorentz group. 

The tools from the representation theory of the translation group and the way they relate to the Lorentz group are interesting in themselves, from a mathematical point of view. Therefore, we illustrate them in some detail in this section, before proceeding to the construction of the new model.
Specifically, the relevant representations are the unitary infinite dimensional representations of both groups. We start from those of the Lorentz group \cite{werner,gelfand_representation, Dirac}, the mathematical basis of known spin foam models for Lorentzian $4d$ gravity.

We first present the representations of the principal series \cite{Barrett1999, DaoNguyen} for the $\SL(2,\bbC)$ group (the covering group of the Lorentz group) then we recall how Dirac initially introduced the infinite dimensional representations for a non-compact group. 
Dirac's construction in terms of expansors \cite{Dirac,Chandra} is then presented; it is particularly relevant to our discussion since it leads to the infinite dimensional representations of the translation group, which we discuss in the second part of this section. The underlying relation between the representations in the principle series for the Lorentz group and those of the translation group becomes evident once we exploit properties of the harmonic oscillator basis, which we deal with in the last part of this section. We remark that one can work in the spinor basis as well, in order to realize the irreducible representations of the Lorentz group, as studied in \cite{Chandra}. The reader who is familiar with the representation theory of the translation and Lorentz group, including the formalism of expansors, and who wants to see directly how these mathematical elements enter the description of quantum simplicial geometry and the construction of the new model, can skip this section and move directly to the following one.

\subsection{Infinite dimensional unitary representations of the Lorentz group}
\label{SubSec_InfiniteDimRepsLorentz}

 Let us first recall the realization of the infinite dimensional unitary representations of the Lorentz group and its algebra in the space of homogeneous functions, and then explore the Plancherel decomposition of the various classes of such representations. We then present the alternative representation of the Lorentz group in terms of expansors, introduced by Dirac in \cite{Dirac, Chandra}.
Finite dimensional representations of the $\SL(2,\bbC)$ group, i.e. tensor and spinor representations, are the ones mostly used in theoretical physics. However, these are not unitary and the problem of identifying finite unitary representations remains unsolved \cite{werner,pierre}. 
On the other hand, infinite dimensional representations of the Lorentz group have been extensively studied in \cite{gelfand_representation, Chandra,DaoNguyen}, and unitary irreducible representations were first derived in \cite{Dirac}. An element $g\in \SL(2,\bbC)$ can be described by the matrix: 
\be
    g =
    \begin{pmatrix}
        \alpha & \beta \\
        \gamma & \delta 
    \end{pmatrix} \,,
\ee
with $\alpha,\beta,\gamma,\delta \in \bbC$ satisfying the relation $\alpha\delta - \beta\gamma = 1$. A representation of the group, denoted $R_\chi(g)$, can then be given by its action on the set of homogeneous polynomial functions $\phi$ of two complex variables $z_1,z_2 \in \bbC$ of order $n_1-1$ in $z_1$ and $n_2-1$ in $z_2$ where $n_i\in \bbN$. We denote by $\chi$  the pair of numbers $\chi=\left(n_1,n_2\right)$. Such action takes the form:
\be
\label{Eq1:rep_theory}
    R_\chi(g) \, \phi(z_1,z_2) = \phi(\alpha z_1 + \gamma z_2, \beta z_1 + \delta z_2) \,.
\ee

Among the $\SL(2,\bbC)$ infinite dimensional representations, we can distinguish the \textit{unitary} ones, which form the so-called \textit{principal series}. Such representations are denoted by $R_{j \, \mu}(g)$, labelled by the half integer $j$ and the real number $\mu \in \bbR$, and again the $\SL(2,\bbC)$ transformations are specified by the action of $R_{i \, \bar\mu}$ on the polynomials of degree $(\frac{1}{2}(\mu+j),\frac{1}{2}(\mu-j))$.

The only isomorphism between such representations is given by $R_{j \, \mu} = R_{-j \, -\mu}$. The representations in the principal series are further classified according to the value of $j$ into \textit{bosonic} integer $j$ and \textit{fermionic} for $j$ half-integer. 

Throughout this paper, when discussing the unitary infinite dimensional representations of the Lorentz group, we will thus refer to the bosonic representations in the principal series of $\SL(2,\bbC)$. 
\smallskip \\
We also note that the irreducible representations of the principal series $R_{j \, \mu}$ correspond to the canonical $D$-matrices $D^{j \, \mu}_{\ell,m;\ell',m'}$. These can be realized as the expectation value of a general unitary operator $U(g)$ associated with the transformation $g \in \SO(1,3)$, in the (canonical) orthogonal basis $|\, j , \mu ; \ell' , m' \ra$ of the Lorentz group 
\be
    D^{j \, \mu}_{\ell,m;\ell',m'}(g) =
    \la j , \mu ; \ell , m \, | U(g) |\, j , \mu ; \ell' , m' \ra \,.
    \label{CanonicalD-Matrices}
\ee
where $(j\, , \mu)$ are the representation labels, whereas $(l\, ,m)$ are the angular momentum- and magnetic quantum numbers. \\
In the following, the so-called \textit{balanced} representations among the principal series of the Lorentz representations will play a central role. As we will see in section \ref{Sec_QuantumBiVector}, we are interested in triangles whose geometry is encoded by a \textit{simple} bivector \cite{Reisenberger:1998SpinNetworkRelativistic, Barrett1999}. On the level of the associated representation, it is obtained by setting the representation labels $j=0$ or $\mu=0$. According to \cite{Barrett1999,gelfand_representation,PerezRovelli}, a natural expression for such representations is given by the Gelfand-Graev transformations on the hyperboloids in Minkowski space. Let $x_{\mu} \in \bbR^4$ be the embedding coordinates for any such hyperboloid, with scalar product $x \cdot x = x_0^2 - \textbf{x}^2$, with $\textbf{x}^2 = x_1^2 + x_2^2 + x_3^2$. Let us consider three hyperboloids: $Q_1$ given by $x \cdot x = 1$ and $x_0 > 0$, the null positive cone $Q_0$ given by $x \cdot x = 0$ and $x_0 > 0$, and the de-Sitter space $Q_{-1}$ given by $x \cdot x =-1$. The Fourier decomposition for the square-integrable functions on the three hyperboloids is given by the Plancherel theorem:
\be
    \begin{aligned}
        L^2[Q_1] & =
        \bigoplus_{\mu} \, R_{0 \, \mu} \, \dd \mu \, \mu^2 \,,\\
        L^2[Q_0] & =
        2 \, \bigoplus_{\mu} \, R_{0 \, \mu} \, \dd \mu \, \mu^2 \,,\\
        L^2[Q_{-1}] & =
        \bigg(2 \,\bigoplus_{\mu} \, R_{0 \, \mu} \, \dd \mu \, \mu^2\bigg) \bigoplus
        \bigg(\bigoplus_j \, R_{j \, 0}\bigg) \,,
    \end{aligned}
    \label{PlancharelDecomposition}
\ee
where $\dd \mu \, \mu^2$ is the Plancherel measure for $j = 0$. The Plancherel theorem thus provides a definition of the Hilbert spaces of (square-integrable) functions on the three hyperboloids based on a combination of the Lorentz balanced (simple) representations.\\
For completeness, looking at the Lorentz algebra, we can notice that it is isomorphic to $\sl(2,\bbC)$ or $\so(3,\bbC)$, both considered as real algebras. Hence, the Lorentz algebra generators can be written as a combination of a rotation $J$ and a boost $N$. In components, this translates to writing a general element of the algebra as $L_{ab} = J_{ab} + \i N_{ab}$. Furthermore, there exist two invariant inner scalar products on this Lie algebra given by:
\be
    \langle L \,,\, L \rangle = J^2 - N^2 
    \,\,,\quad
    \langle L \,,\, \ast L \rangle = 2 \, J \cdot N \,,
    \label{LorentzScalarProd}
\ee
where the symbol $\ast$ denotes the Hodge dual: $\ast L_{ab} = \frac{1}{2} {\varepsilon_{ab}}^{cd} L_{cd}$ and ${\varepsilon_{ab}}^{cd}$ is the $4d$ Levi-Civita tensor.
They lead to the Casimir elements with eigenvalues:
\be
    C_1 = j^2 - \mu^2 -1
    \,\,,\quad
    C_2 = 2 j \, \mu \,,
    \label{LorentzCasimirEigenvalue}
\ee
which can be recast as the real and imaginary parts of the complex quantity $\omega^2 - 1 = C_1 + \i C_2$, with $\omega = j + \i \mu$.

\noindent The known constructions of spin foam and group field theory models for Lorentzian quantum gravity are based on the above realization of Lorentz representations (indeed, quantum states, as well as the GFT field, are given by functions on (several copies of) the hyperboloids).

\noindent Let us now introduce another realization of the same representations of the Lorentz group.
Historically, these were introduced by Dirac in \cite{Dirac} and further studied in \cite{Chandra} in terms of so-called \textit{expansors}. They are tensors-like objects in a manifold with an infinite enumerable number of components. They also possess an invariant positive definite quadratic form for their squared length. Such representations are realized on the space of homogeneous polynomials on Minkowski space $\M^4$. 

\noindent The notions introduced in what follows will help to establish the relation between the representations of the Lorentz group and that of the translation group (geometrically, between the quantum bivectors and the quantum edge vectors).

\noindent We focus on homogeneous polynomials (of degree $n$) of the real vector $\xi_{\mu}$ built from monomials $\xi_x^i \xi_y^j \xi_x^k \xi_t^{-1-h}$, with $i,j,k,h \in \bbN$. Notice that, due to the negative power of the time coordinate $\xi_t$, the combination of such monomials is infinite dimensional. The homogeneous polynomial can be written in the form of a power series as
\be
    P(\xi_{\mu}) = \sum_{ijkh} \, A_{i \, j \, k \, h} \, \xi_x^i \xi_y^j \xi_z^k \xi_t^{-1-h} \,,
    \label{ExpansorDef}
\ee
where the coefficient $A$ is called expansor. These coefficients are regarded as the coordinates of vectors in an infinite dimensional space. 
On this space of homogeneous functions, one can define the unitary representations of the Lorentz group, with the unitarity condition enforced by the scalar product:
\be
    P_1 \cdot P_2 = \sum_{ijkh} \, A_{i \, j \, k \, h} B_{i \, j \, k \, h} \,,
    \label{ScalarProduct}
\ee
that also provides a notion of norm given by the square length of the polynomial $|P|^2 = \sum_{ijkh} \, A^2_{i \, j \, k \, h}$.

\noindent In order to derive the representation of the Lorentz group on such expansors, one first considers the infinitesimal Lorentz transformations on the basis coordinates $\xi$. Demanding that the square length $\xi_{0}^{2}-\xi_{1}^{2}-\xi_{2}^{2}-\xi_{3}^{2}$ remains invariant under these transformations, one obtains the corresponding transformation for the expansors. It is precisely these induced linear transformations on the expansors that identify them as unitary representations of the Lorentz group.

\subsection{The Harmonic oscillator representation of the translation group}
\label{SubSec_InfiniteDimRepTranslations}

Dirac emphasized that the expansors can be interpreted as a tensor product of four harmonic oscillators. This is also crucial in the case of $n$-value integrals, where one can make a variable transformation that is familiar in quantum mechanics, relating the expansors in \eqref{ExpansorDef} with harmonic functions \cite{Dirac}. This can be seen explicitly considering the transformation map:
\be
    \begin{aligned}
        x_a & = 
        \frac{1}{\sqrt{2}}\big(\xi_a + \frac{\partial}{\partial \xi_a}\big) \,,\\
        x_t & =
        \frac{1}{\sqrt{2}}\big(\xi_t - \frac{\partial}{\partial \xi_t}\big) \,,
    \end{aligned}
    \qquad
    \begin{aligned}
        \frac{\partial}{\partial x_a} & = 
        \frac{1}{\sqrt{2}}\big(\frac{\partial}{\partial \xi_a} - \xi_a\big) \,,\\
        \frac{\partial}{\partial x_t} & =
        \frac{1}{\sqrt{2}}\big(\frac{\partial}{\partial \xi_t} + \xi_t\big) \,,
    \end{aligned}
    \label{InverseMap_OperatorCoordinates}
\ee
for $a = x,y,z$. This map implies the correct commutation relations between all the above-introduced operators. In the new $x-$variables, the homogeneous polynomial on Minkowski space \eqref{ExpansorDef} can be represented as a combination of polynomials in \eqref{InverseMap_OperatorCoordinates} which are basically a general composition of four Hermite functions:
\be
    P(x_{\mu}) = \sum_{ijkh} \, A_{i \, j \, k \, h} \, \Psi_{i \, j \, k \, h}(x_{\mu}) \,,
    \label{HomogeneousPolynomials_Harmonic}
\ee
where the expansors are the coefficient of such combination:
\be
    \begin{aligned}
        \Psi_{i \, j \, k \, h}(t,x,y,z) & =
        \frac{1}{\pi n! \sqrt{2^{i+j+k+h}}} \, 
        \big(x_i - \partial_{x_i}\big)^i \,
        \big(x_j - \partial_{x_j}\big)^j \,
        \big(x_k - \partial_{x_k}\big)^k \,
        \big(x_h - \partial_{x_h}\big)^h \,
        e^{- \frac{1}{2} (x_t^2 + x_ax^a)} \\
        & =
        \psi_{h}(t) \, \psi_{i}(x) \, \psi_{j}(y) \, \psi_{k}(z) \,.
    \end{aligned}
    \label{FourDimHarmonicOscillator}
\ee
In these terms, the scalar product \eqref{ScalarProduct} simplifies thanks to the orthogonality of the Hermite functions and it takes the integral form:
\be
    P_1 \cdot P_2 = \int \dx^4 \, 
    \psi_{h_1}(t) \, \psi_{i_1}(x) \, \psi_{j_1}(y) \, \psi_{k_1}(z) \,\, 
    \psi_{h_2}(t) \psi_{i_2}(x) \psi_{j_2}(y)  \psi_{k_2}(z)
    =
    \frac{i_1!j_1!k_1!}{h_1!} \delta_{i_1i_2} \delta_{j_1j_2}\delta_{k_1k_2} \, \delta_{h_1h_2} \,.
\ee
The alternative representation of the $\xi$ variables that Dirac introduced is related to the theory of the four-dimensional harmonic oscillator. 
Indeed, the four $x$-parameters can be treated as the coordinates of a four-dimensional harmonic oscillator, whereas the respective four operators $\partial_{x_\mu}$ are the associated conjugate momenta $p_{x_\mu}$. Furthermore, 
a state of the oscillator with components $0, 1, 2, 3$ occupying the $i$th, $j$th,
$k$th, $h$th quantum states respectively, is represented by $\Psi_{ijkh}$ given in \eqref{FourDimHarmonicOscillator}. Following the map (\ref{InverseMap_OperatorCoordinates}) one can get back the $\xi$-representation and the function $\Psi_{i \, j \, k \, h}(x_{\mu})$ goes over to $\ \xi_x^i \xi_y^j \xi_z^k \xi_t^{-1-h}$. Note that the degree of the expansor corresponds to the energy of the state of the harmonic oscillator, not including the zero-point energy.\\
Recalling the expressions of the ladder operators associated with a four-dimensional harmonic oscillator (see appendix (\ref{App_4dHarmonicOscillator})), we recognize that they can be written as the inverse relation of \eqref{InverseMap_OperatorCoordinates}:
\be
    \begin{aligned}
        &
        \adag_i = \xi_i = \frac{1}{\sqrt{2}}(x_i-\partial_i) \,,\\
        &
        a_i = \partial_{\xi_i} = \frac{1}{\sqrt{2}}(x_i+\partial_i) \,,
    \end{aligned}
    \qquad
    \begin{aligned}
        &
        \adag_0 = -\partial_{\xi_t} = \frac{1}{\sqrt{2}}(t-\partial_t) \,,\\
        &
        a_0 = \xi_t = \frac{1}{\sqrt{2}}(t+\partial_t) \,,
    \end{aligned}
    \label{Map_OperatorCoordinates}
\ee
where $i=1,2,3$ are the spacelike indices and the $0$ index stands for the timelike one. Notice how the Lorentzian signature is reflected in the relation between the creation and annihilation operators and the coordinates on Minkowski space. For completeness, the ladder operators always satisfy the canonical commutators
\be
\label{canonical_commutation_relation}
    [a_{\mu} \,,\, \adag_{\nu}] = \delta_{\mu\nu} \,,
\ee
where the spacelike creation operators $\adag_i$ are represented by the space coordinates of Minkowski space $\xi_i$, but the timelike creation operator $\adag_0$ is represented by the time component of the corresponding momentum vector (up to a sign), and vice-versa for the annihilation operators. 

\

\noindent Based on the above, the momenta $p_{\mu}$ can be seen as the generators of the \textit{group of translations} on Minkowski space.
A basis of infinite dimensional representations for the translation group is thus given by eigenvectors of the harmonic oscillator ladder operators. In quantum mechanics, it is well known that such a set of eigenvectors are those that mostly resemble the classical behavior of the oscillator, i.e. \textit{coherent states} \cite{Genest2013}. In particular, there exists a coherent ket for the annihilation operator and a coherent bra for the creation operator, given by:
\be
    a \,|\, \alpha \ra = \alpha \,|\, \alpha \ra
    \,\,,\quad
    \la \alpha \,|\, \adag = \la \alpha \,|\, \alpha^* 
    \,,\qquad \text{ with } \,\,\, \alpha \in \bbC \,.
    \label{CoherentStates_Eigenbasis}
\ee
The action of the creation operator on the coherent ket or that of the annihilation operator on the coherent bra can be derived by expanding the coherent state as a combination of the energy (number operator) eigenstates of the harmonic oscillator:
\be
    \,|\, \alpha \ra = e^{- \frac{1}{2} |\alpha|^2} \,
    \sum_{n=0}^{\infty} \, \frac{\alpha^n}{\sqrt{n!}} \,|\, n\ra \,.
    \label{CoherentStates_CombinationEigenstates}
\ee

The edge vectors we want to base our quantum simplicial geometry on are elements of $\M^4$ and thus, upon quantization, can be regarded as a tensor product of coherent states $|\, \alpha_{\mu}\ra = |\, \alpha_t,\alpha_x,\alpha_y,\alpha_x\ra$. This is, in turn, an eigenvector simultaneously of the three dimensional translation group generators $p_i = -i \partial_{\xi_i} = -i a_i$, \textit{and} the timelike position operator $\xi_t = a_0$.
This fact is the key element for our quantum geometric construction.
We refer to \cite{Genest2013} for more details on the coherent states of the harmonic oscillator. 
However, such representations are not unitary, as the coherent states are not orthogonal
\be
    \la \alpha \,|\, \beta \ra = e^{-\frac{1}{2}(|\alpha|^2 + |\beta|^2 - 2\alpha^* \beta)} \,.
\ee
On the other hand, since generic states of the harmonic oscillator can be decomposed as a combination of coherent states, they form an overcomplete basis and give a resolution of the identity:
\be
    \int \dd \alpha \, |\,\alpha \ra\la \alpha \,| = \pi \,.
\ee
The homogeneous polynomials in  \eqref{ExpansorDef} or  \eqref{HomogeneousPolynomials_Harmonic} are thus functions on Minkowski space or equivalently functions on the translation group.
In practice, we consider a vector $e = (\lambda_t,\lambda_x,\lambda_y,\lambda_z)\in \M^4$ in Minkowski space parametrized by the coordinates $\lambda_{\nu}$. 
The Hilbert space associated with such a vector is the space of square-integrable functions on four copies of $\bbR$, $L^2[e] := L^2[\lambda]$. In this picture, the coordinates $\lambda_{\nu} \in F(\bbR^4)$ are identified with the generators of the functions on Minkowski space, i.e. of the Hilbert space $L^2[\lambda]$. We thus refer to the \textit{vector wave function} (or alternatively to the quantization of the vector $e$), as the general function:
\be
    f(\lambda) \,\, \in L^2[\lambda] \cong F(\M^4) \,.
    \label{Vector_WaveFunction}
\ee
A basis for the four dimensional harmonic oscillator is realized as a tensor product of independent oscillators. Therefore, we define the four dimensional ladder operators as a linear combination of creation and annihilation operators such that
\be
    \begin{aligned}
        \adag(\lambda) & = a_0 \lambda_t + \adag_1 \lambda_x  + \adag_2 \lambda_y  + \adag_3 \lambda_z \,,\\ 
        a(\lambda) & = \adag_0 \lambda_t + a_1 \lambda_x  + a_2 \lambda_y  + a_3 \lambda_z \,.
    \end{aligned}
    \label{LadderOperator_Tensor}
\ee
The first operator is associated with the position on Minkowski space, identified by the vector $e$, which has the coherent bra $\la \alpha_{\nu} \,|$ as eigenvector, with eigenvalue $\alpha^*_{\nu}$; the second operator is associated with the momentum (or equivalently to the translation generator) on Minkowski space identified by the vector $e$, that has the coherent ket $|\, \alpha_{\nu} \ra$ as eigenvector, with eigenvalue $\alpha_{\nu}$. \\
The wave function \eqref{Vector_WaveFunction} of the quantum vector can thus be expanded in Fourier modes of the translation group (\lq momentum space\rq) or likewise in the harmonic decomposition of the position group (\lq configuration space\rq):
\be
    \begin{aligned}
        \textit{Translation:} \quad
        & f(\lambda) :=
        \int   \dd\alpha \dd\alpha' \, 
        \la \alpha_{\nu} \,| a(\lambda) |\, \alpha_{\nu}' \ra \, f_{\alpha_{\nu}\,,\alpha_{\nu}'} \,,\\
        \textit{Position:} \quad
        & f'(\lambda) := 
        \int \dd \alpha  
        \dd\alpha' \, 
        \la \alpha_{\nu} \,| \adag(\lambda) |\, \alpha_{\nu}' \ra \, f'_{\alpha_{\nu}\,,\alpha_{\nu}'} \,,
    \end{aligned}
    \label{VecWaveFunc_FourierDecomposition}
\ee
where the terms $\la \alpha_{\nu} \,| a(\lambda) |\, \alpha_{\nu}' \ra$ and $\la \alpha_{\nu} \,| \adag(\lambda) |\, \alpha_{\nu}' \ra$ are matrix elements that play the role of plane waves in the operation relating the Fourier transform of functions belonging to the Hilbert space $L^2[\lambda]$ to the spaces of Fourier modes on the translation (or position) group. The two decompositions are related to each other since the two plane waves $\la \alpha_{\nu} \,| a(\lambda) |\, \alpha_{\nu}' \ra$ and $\la \alpha_{\nu} \,| \adag(\lambda) |\, \alpha_{\nu}' \ra$ are the complex conjugate of each other\footnote{By making further assumptions on the wave functions, one can derive the relation between the Fourier modes $f_{\alpha_{\nu}\,,\alpha_{\nu}'} ,\, f_{\alpha_{\nu}\,,\alpha_{\nu}'}'$. For instance, requiring the wave function to be real, $\bar{f}(\lambda) = f(\lambda)$, one would obtain the condition $\bar{f}_{\alpha_{\nu}\,,\alpha_{\nu}'} = f_{\alpha_{\nu}\,,\alpha_{\nu}'}'$.}. This is, in the end, a reformulation of standard Fourier analysis on Minkowski space in group-theoretic terms, that will turn out to be useful in the following. \\
As we will exploit in section \ref{Sec_QuantumBiVector}, an explicit realization of the plane waves $\la \alpha_{\nu} \,| a(\lambda) |\, \alpha_{\nu}' \ra$ and $\la \alpha_{\nu} \,| \adag(\lambda) |\, \alpha_{\nu}' \ra$ can be provided in terms of solutions of the harmonic oscillator.

\paragraph{The four dimensional Lorentzian harmonic oscillator}
As shown in \cite{Genest2013,hyperbolic_laplacian, laplacian_hyperbolic_solution}, an infinite dimensional realization for the balanced representations $R_{0 \, \mu}$ of the Lorentz group can be obtained as the space of solutions of the Laplace equation on the hyperboloid. The Laplacian operator coincides with the Casimir $C_1$ of the Lorentz algebra, with eigenvalue $-(1+\mu^2)$. This relation is particularly interesting for us because it provides the quantization of a simple (triangle) bivector and its underlying connection to the quantum (edge) vector.\\
By a simple change of variables (see appendix (\ref{App_4dHarmonicOscillator})), the Casimir $C_1$  in \eqref{LorentzCasimirEigenvalue} for balanced representations\footnote{The only one not identically vanishing - the vanishing of the other Casimir corresponds, in the spin foam context, to the quantization of the diagonal simplicity constraints} can be expressed in terms of hyperbolic parameters as:
\be
    C_1 = J^2 - N^2 = 
    \frac{1}{\sinh^2\eta}\partial_{\eta}(\sinh^2 \eta \partial_{\eta}) +
    \frac{1}{\sinh^2\eta} \Big(
    \frac{1}{\sin\theta}\partial_{\theta}(\sin \theta \partial_{\theta}) + \frac{1}{\sin^2\theta} \partial^2_{\varphi}\Big) \,,
    \label{CasimirC1_Laplace}
\ee
which is the Laplacian on the hyperboloid $Q_1$ introduced in section \ref{SubSec_InfiniteDimRepsLorentz} for the Plancherel decomposition of the Lorentz representations. The parameters above $\eta \in \bbR \,, \theta \in [0,\pi]\,, \phi \in [0,2\pi]$ and $r$ are the hyperbolic coordinates. In section \ref{SubSec_InfiniteDimRepTranslations} we explained that the ladder operators of the harmonic oscillator can be used as generators of the translation (and position) group.
Therefore, we can use the coordinates on Minkowski space $\xi_{\mu}$ and their momenta, recast as the ladder operators of the harmonic oscillator coordinates \eqref{Map_OperatorCoordinates}, to realize the Lorentz generators (rotations and boosts)
\be
    J_a := -i \varepsilon_{abc} \, \xi_b \partial_{\xi_c} =
    -i \varepsilon_{abc} \, \adag_b a_c
    \,\,,\quad
    N_a := -i (\xi_t \partial_{\xi_a} + \xi_a \partial_{\xi_t}) =
    -i (a_0 a_a - \adag_0 \adag_a) \,.
    \label{Lorentz-HarmonicOscillator}
\ee
One can check that the above operators satisfy the usual $\so(1,3)$ commutation relations
\be
    [J_a\,,\,J_b] = i\varepsilon_{abc} \, J_c 
    \,\,,\quad
    [J_a\,,\,N_b] = i\varepsilon_{abc} \, N_c 
    \,\,,\quad
    [N_a\,,\,N_b] = -i\varepsilon_{abc} \, J_c \,.
    \label{LorentzCommutators}
\ee
Since we would like, in the following, to express the balanced representations of the Lorentz group in terms of representations of the translation group (and vice-versa), we need to derive the eigenstates of \eqref{CasimirC1_Laplace} (representations of the Lorentz group) as a combination of the eigenstates of the harmonic oscillator (representations of the translation group).
In particular, we show in this subsection that the eigenfunctions of \eqref{CasimirC1_Laplace} are the non-radial contribution of the eigenfunctions of the four dimensional harmonic oscillators.
\smallskip \\
To proceed, we recall that the target space used by Dirac \cite{Dirac} (the space of homogeneous polynomials) to construct the infinite dimensional representations of the Lorentz group, can be re-expressed as the Hilbert space of states for the four dimensional Lorentzian harmonic oscillator. 
Here the polynomials \eqref{ExpansorDef} take the form \eqref{HomogeneousPolynomials_Harmonic}. Therefore, the set of homogeneous polynomials on Minkowski space can be derived as the general solution of the Schr\"odinger equation $\cH \, \Psi = E \, \Psi$, where the Hamiltonian operator yields
\be
    \cH = - \frac{1}{2} \Delta + \frac{1}{2} (t^2 - x^2 - y^2 - z^2) \,,
\ee
is the Hamiltonian of the four dimensional Lorentzian harmonic oscillator and $\Delta$ is the four dimensional flat d'Alambertian operator $\Delta = \partial_t^2 - \partial_x^2 - \partial_y^2 - \partial_z^2$. In the appendix (\ref{App_4dHarmonicOscillator}) we show how to solve such Schr\"odinger equation in different coordinate basis, inspired by \cite{Genest2013}, and derive the eigenbasis in terms of the variables of the \textit{hyperbolic basis}:
\be
    \Psi_{n_r,\mu,\ell,m}(r,\eta,\theta,\phi) =
    (-1)^{n_r} \sqrt{\frac{n_r!}{\Gamma(n_r + \mu + 1/2)}} \,
    r^{\mu-1} e^{-\frac{1}{2}r^2} \, L^{(i\mu)}_{n_r}(r^2) \,
    \frac{1}{\sinh\eta} \, Q_{\ell}^{i\mu}(\coth\eta) \,
    Y^m_{\ell}(\theta,\phi) \,,
    \label{HyperbolicBasis2}
\ee
associated with  the energy $E = 2n_r+i\mu+1$, where $L_{n}^{(\alpha)}$ are the Laguerre polynomials, $Q^{\alpha}_{\lambda}$ are the Legendre functions of the second kind and $Y_{m}^{\ell}$ are the spherical harmonics.
Furthermore, we also show how this wave function is related to the homogeneous polynomials proposed by Dirac, expressed in terms of the harmonic functions:
\be
    \Psi_{n_r,\mu,\ell,m}(r,\eta,\theta,\phi) =
    \sum_{n_t,n_x,n_y,n_z} \, 
    \cC^{n_t,\, n_x,\, n_y,\, n_z}_{n_r,\, \mu,\, \ell,\, m} \, \Psi_{n_t,n_x,n_y,n_z}(t,x,y,z) \,,
    \label{Lorentz-Translations}
\ee
where the coefficients are given by 
 \begin{align}
    \cC^{n_t,\, n_x,\, n_y,\, n_z}_{n_r,\, \mu,\, \ell,\, m} 
         & =
  \la n_t,n_x,n_y,n_z | n_r,\mu;\ell,m \ra 
    \nonumber \\
     & =
     \sum_{n_R, n_{\rho}} \,
     \frac{i^{m + |m|} (-1)^{\tn_x+n_{\xi}} (\sigma_m i)^{n_y}}{2^{(1-\delta_{m,0})/2}} \,
    e^{i\varphi} \,\, 
     \cC^{\frac{1+|m|}{2}, \frac{1}{4} + \frac{q_z}{2}, \frac{\ell}{2} + \frac{3}{4}}_{n_{\rho}, \tn_z, n_R} \nonumber\\
     &\times
   \cC^{\frac{1}{4} + \frac{q_x}{2}, \frac{1}{4} + \frac{q_y}{2}, \frac{1+|m|}{2}}_{\tn_x, \tn_y, n_{\rho}} \,
   \cC^{1/4+q_t/2,\, \ell/2+3/4,\, (1+i\mu)/2}_{\tn_t+\on_t/2,\, n_R,\, n_r} \,,
   \label{Lorentz-Translations_Coefficients}
 \end{align}
which are a combination of the $\su(1,1)$ Clebsh-Gordan coefficients \eqref{ClebshGordan_su(1,1)}. 
As we have anticipated, the non-radial part of the solution, yielding the equation \eqref{HyperbolicBasis}, reads
\be
    \Psi_{\mu,\ell,m}(\eta,\theta,\phi) =
    \frac{1}{\sinh\eta} \, Q_{\ell}^{i\mu}(\coth\eta) \,
    Y^m_{\ell}(\theta,\phi) \,,
\ee
and it is the solution of the Laplace equation with Laplacian \eqref{CasimirC1_Laplace} and eigenvalue $-(1+\mu^2)$. Therefore, it is the infinite dimensional representation of the Lorentz group associated with a timelike bivector (normal to a spacelike triangle). These solutions can indeed be used to construct the (timelike) balanced $D$-matrices \eqref{CanonicalD-Matrices}:
\be
\label{eq:D_matrix_timelike_hyperbolik}
    D^{0 , \mu}_{\ell,m;\ell',m'}(g) = 
    \la 0,\mu;\ell,m \,| U(g) |\, 0,\mu;\ell',m' \ra = 
    \int \dd r \dd \eta \dd \Omega r^3 \sin^2(\eta) \, 
    \overline{\Psi}_{\mu,\ell,m}(\eta,\theta,\phi) \Psi_{\mu,\ell',m'}(\eta',\theta',\phi') \,,
\ee
where the primed coordinates $\{\eta',\theta',\phi'\}$ encode the action of the Lorentz transformation $g$.
We further note that, the wave function $\Psi_{n_t,n_x,n_y,n_z}(t,x,y,z)$ is an infinite dimensional realization of the harmonic oscillator eigenstate $|\, n_t,n_x,n_y,n_z \ra$. 
We recall that the coherent states of the harmonic oscillator, which are a linear combination of the eigenstates $|\, n_t,n_x,n_y,n_z \ra$, were identified as the eigenbasis of the translation group \eqref{CoherentStates_Eigenbasis}. We can thus extend the relation \eqref{CoherentStates_CombinationEigenstates} to the wave functions
\be
    \Psi_{\alpha_t,\alpha_x,\alpha_y,\alpha_z}(t,x,y,z) = 
    e^{-\frac{1}{2}\sum_{\nu}|\alpha_{\nu}|^2} \, \sum_{\{n_{\nu}\}} \,
    \frac{\alpha_t^{n_t} \alpha_x^{n_x} \alpha_y^{n_y} \alpha_z^{n_z}}{\sqrt{n_t! n_x! n_y! n_z!}} \, 
    \Psi_{n_t,n_x,n_y,n_z}(t,x,y,z) \,,
    \label{HarmonicOscillator_CoherentMinkowskian}
\ee
which, together with \eqref{Lorentz-Translations}, provides the relation between the wave functions associated with  the infinite dimensional representations of the Lorentz group and the wave functions associated with  the infinite dimensional representations of the translation group, given by
\begin{align}
    \Psi_{n_r,\mu,\ell,m}(r,\eta,\theta,\phi)& =
    \sum_{n_t,n_x,n_y,n_z} \, \int \dd^4 \alpha_{\nu} \, e^{-\frac{1}{2} \sum_{\nu} |\alpha_{\nu}|^2} 
    \cC^{n_t,\, n_x,\, n_y,\, n_z}_{n_r,\, \mu,\, \ell,\, m} \,
    \frac{\alpha_t^{* \, n_t} \alpha_x^{* \, n_x} \alpha_y^{* \, n_y} \alpha_z^{* \, n_z}}{\pi^4 \, \sqrt{n_t! n_x! n_y! n_z!}}\nonumber\\
&\times\Psi_{\alpha_t,\alpha_x,\alpha_y,\alpha_z}(t,x,y,z) \,.
    \label{HarmonicOscillator_HyperbolicCoherent}
\end{align}

These relations will be crucial for relating the new spin foam model based on edge vectors, quantized in terms of representations of the translation group, and the Barrett-Crane spin foam model based on a description of simplicial geometry using bivectors, quantized in terms of representations of the Lorentz group. 

\noindent Let us summarize the results of the above construction. 

We have shown that: a) expansors can be used as a representation for translation group; b) since edge vectors can be interpreted as translations on Minkowski space, expansors can be used as a representation for edge vectors; c) expansors can be expressed as a combination of $4d$ Lorentzian harmonic oscillators, and specifically as a combination of both creation and annihilation operators; d) hence, edge vectors can be represented analogously, i.e. in terms of the harmonic oscillator phase space variables.

\noindent We emphasize here that, while functions of individual edge vectors form a commutative algebra, their expression in terms of ladder operators of the $4d$ Lorentzian harmonic oscillator phase spaces (see \eqref{LadderOperator_Tensor}), via expansors, will play a crucial role in the later construction. It is at the root of the non-commutativity of the Lorentz algebra constructed from them \eqref{canonical_commutation_relation}, and thus of the algebra of functions of the quantum bivectors\footnote{This structure also forms the foundation of the bivector/flux representation \cite{FinocchiaroOriti:2018SFDufloMap}}, but it holds for the quantum triangle too. 
Indeed, this non-commutative structure manifests explicitly in the quantum triangle states, expressed as functions of pairs of edge vectors, thus of the tensor product of harmonic oscillators phase spaces, and features crucially in the subsequent definition of quantum amplitudes for the new model.

\noindent The precise form of non-commutativity found at the level of the algebra of functions of triangle (and bivector) data depends on the specific quantization map one chooses for the classical phase space, at the level of triangle/bivector data, or at the very initial level of edge vectors, thus harmonic oscillators. We will discuss this further below.
This choice affects then the final state sum amplitudes for quantum geometry, as we will discuss.

\section{Quantum triangle}
\label{Sec_QuantumTriangle}
We now start our new construction for quantum simplicial geometry with the analysis of a single triangle embedded in Minkowski space, first at the classical, then at the quantum level.
We recall how to combine geometric edge vector data on a triangle and their classical constraints. We then build the corresponding quantum version of the constraints, using the group-theoretic tools of section \ref{Sec_InfiniteDimReps} that define a quantum edge vector. Since both edge vectors and bivectors can be used to describe a geometric triangle, we illustrate both constructions and point out their differences at the quantum level. 
\subsection{Classical triangle}
\paragraph{Classical triangle.} The intrinsic geometry of a classical triangle in Minkowski space is completely determined (up to embedding information) by two among its three edge vectors $e_1,e_2,e_3 \in \M^4$, as illustrated in Fig.\ref{Fig:ClassicalTriangle}, with the latter set satisfying the \textit{closure condition}, namely the relation: $e_1 + e_2 +e_3=0$. Working with the larger, constrained set of variables ensures that the description is not affected by any specific choice of a pair of edge vectors. To spell out, for later use, these obvious facts, consider three edge vectors in Minkowski space given by:
\be
    e_1 = (\zeta_t,\zeta_x,\zeta_y,\zeta_z)
    \,\,,\quad
    e_2 = (\lambda_t,\lambda_x,\lambda_y,\lambda_z)
    \,\,,\quad
    e_3 = (\omega_t,\omega_x,\omega_y,\omega_z) \,,
    \label{Vectors_Triplet}
\ee
parametrized by the coordinates $\zeta_{\mu}, \lambda_{\mu}, \omega_{\mu} \in F(\M^4)$. All the geometric properties of the corresponding triangle are determined by any pair of these edge vectors, plus the closure condition:
\be
    e_1 + e_2 + e_3 =0
    \quad\Rightarrow\quad
    \zeta_{\mu} + \lambda_{\mu}+\omega_{\mu}=0 \,.
    \label{Traingle_ClosureCondition}
\ee

\paragraph{Classical bivector.}
A classical bivector, associated with the same triangle in Minkowski space, is obtained by taking the wedge product of two edge vectors. A bivector contains less information than the one required to specify the geometry of the associated triangle. In fact, if we consider again three edge vectors $e_1,e_2,e_3 \in \M^4$, the bivector geometry is determined by the two conditions:
\begin{enumerate}
    \item \textit{Closure relation}: every edge vector of the three vectors of the triangle is given by the sum of the other two $e_1 + e_2 +e_3=0$;
    \item \textit{Skew-symmetry}: the normal to the plane spanned by the constrained three vectors, in which the triangle lies, is given by the restriction to the wedge (or external) product of any two edge vectors (up to a change of orientation).
\end{enumerate}
The first condition alone would reduce the information contained in the three (edge) vectors to that characterizing a classical triangle, while the second restricts it further (again, in a way that remains independent of any specific choice of two edge vectors). Given the three edge vectors \eqref{Vectors_Triplet} (and the closure condition), the bivector represented in Fig.\ref{Fig:ClassicalTriangle} is thus constructed as:
\be
    b := e_1 \wedge e_2 = e_1 \wedge e_3 = e_3 \wedge e_2 \,.
    \label{Classicalbivector}
\ee
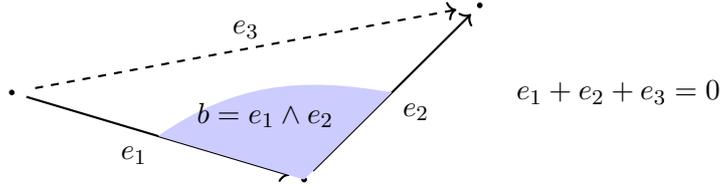
\begin{figure}
    \centering
    \input{QuantumTriangle}
    \caption{Triangle with edge vectors $e_1,e_2,e_3 \in \M^4$ and closure relation \eqref{Traingle_ClosureCondition}. In blue is the bivector part.}
    \label{Fig:ClassicalTriangle}
\end{figure}
\subsection{Quantum triangle}
\label{Sec_QuantumBiVector}
We now proceed to the quantization of a triangle geometry and its bivector, following the suggestions in \cite{CraneYetter}, and relying on the relations we obtained in section \ref{SubSec_InfiniteDimRepTranslations}.  

In the construction, it is essential to keep in mind the relation between the $4d$ harmonic oscillator and the representations of both the Lorentz and translation groups. Also, we see the emergence of the aforementioned non-commutative structure in quantum states—expressed as functions of edge vectors, and thus the role of the choice of quantization map. Specifically, the algebra of a quantum triangle, constructed from the tensor product of algebras associated with a pair of edge vectors (identified as the algebra of a $4d$ harmonic oscillator), is given a universal enveloping algebra structure. The quantization map for this defines then the star product and the non-commutative plane waves for the corresponding algebra of functions. In the following, we keep the discussion on the star product as general as possible, indicating where a choice has to be made, and providing an example of such choice.
\medskip

\paragraph{Quantum triangle from quantum edge vectors.} We now proceed with deriving the wave function of a quantum triangle using a pair of quantum edge vectors. 
\medskip

Recall that, we identified in \eqref{Vector_WaveFunction} the quantum states of an edge vector in Minkowski space with the (square-integrable) functions of the translation group. This was then realized as the Hilbert space of the four dimensional harmonic oscillator. Let us then consider a triangle $t$ with edge vectors that satisfy the closure condition \eqref{Traingle_ClosureCondition}. 
To start with the most symmetric realization, the space of states associated with the quantum triangle is simply the tensor product of the Hilbert spaces associated with its three edge vectors, subject to the closure condition:
\be
    L^2[e_1,e_2, e_3] := L^2[e_1] \ot L^2[e_2] \ot L^2[e_3]
    \,,\quad e_1+e_2+e_3=0 \,.
    \label{Triangle_HilbertSpace}
\ee
The closure condition can also be seen as imposing invariance of states under the diagonal action of the translation group on the three arguments of the wave function seen as points in $\M^4$, thus identifying the space of states as spanned by invariant tensors in the three representation spaces of the translation group. This is conceptually correct, but the non-compactness of the translation group makes a definition of such invariant tensors not trivial. 

\noindent Obviously, the closure condition makes one of the three edge vectors immaterial, fixing it as a function of the other two, so we have a reduced, but equivalent quantization in terms of two quantum edge vectors only. In other words, the Hilbert space of a quantum triangle is in fact the tensor product of the Hilbert spaces associated with  any two of its edge vectors, for instance; $e_1,e_2 \subset t$, subject to the closure condition:
\be
    L^2[e_1,e_2] := L^2[e_1] \ot L^2[e_2] 
    \,,\quad e_1+e_2+e_3=0 \,.
    \label{Triangle_HilbertSpace_2}
\ee

To ensure proper geometric interpretation, our construction should not be affected by the choice of the pair of vectors $e_1,e_2$ (this is a leftover of the closure condition in the reduced description). 
To ensure this, we introduce the operator that maps the quantum state space of a triangle in terms of the edge vectors $e_1,e_2$ into the quantum state space of the exact same triangle in terms of the edge vectors $e_1,e_3$.
Let such \textit{switching operator} be defined as the map:
\be
    \sigma : L^2[e_1,e_2] \to L^2[e_1,e_3] \,.
    \label{SwitchingOperator}
\ee
Thus, the proper quantum state space of a triangle is the space $L^2[e_1,e_2]$ equipped with the closure condition and invariant under the action of the switching operator $\sigma$.

\paragraph{Quantum triangle from the quantum bivector.}
As we have seen, one can restrict the classical configuration space of a triangle to its anti-symmetric part, to obtain the geometry of the associated bivector. At the quantum level, this procedure is realized by taking the skew-symmetric part of the tensor product \eqref{Triangle_HilbertSpace}, which is invariant under the $\sigma$ operator. This subspace, denoted by $T^A[e_1,e_2]$, is therefore naturally identified as the Hilbert space of a bivector.
Moreover, as discussed in section \ref{Sec_InfiniteDimReps}, by constructing bivectors (at both classical and quantum levels) directly from wedging edge vectors, we restrict our attention only to the simple bivectors\footnote{Geometrically, this simplicity condition ensures that the two edge vectors and the normal vector are not planar, and is imposed  purely as a condition on bivectors:
\be
    \langle b \,,\, \ast b \rangle = 0 \,. 
\ee These are indeed the well-known simplicity constraints on which the usual spin foam construction in terms of bivectors is based \cite{BaratinOriti:2011GFTBarrettCrane, BarrettCrane:1999BarrettCrane1}.}.
At the quantum level, this can be equivalently translated into associating the bivectors with duals of Lie algebra elements $b = \ast L \in \so(3,1)^*$ and quantizing them by replacing the (dual) Lie algebra with a sum over their representation category. The simplicity constraint amounts to restricting to the balanced representations $R_{0 \, \rho}$ or $R_{j \, 0}$ of the principle series, whose Plancherel decomposition is given in the set of equations \eqref{PlancharelDecomposition}; the choice of one or the other of the two classes of balanced representations corresponds to considering spacelike or timelike triangles, respectively (see \cite{Barrett1999,Jercher:2022mky}).
We can then establish that, once we associate an element of the dual Lorentz Lie algebra $\so^*(1,3)$ to each wedge product of a pair of translations on Minkowski space $\M^4 \wedge \M^4$, we identify an isomorphism between the Hilbert space of wedged vectors (a simple bivector) with the space of functions on such Lie algebra. This isomorphism reads:
\be
    T^A[e_1,e_2] := L^2[e_1] \wedge L^2[e_2] \cong F(\so^*(1,3)) \,.
\ee

Using the Hilbert space of a quantum bivector, we can construct explicitly the operator corresponding to it, along the lines suggested in \cite{CraneYetter}, and provide different representations for the same Hilbert space. This is relevant for deriving the Fourier expansion of bivector wave functions. \\
Let us recall that in section \ref{SubSec_InfiniteDimRepTranslations} we identified the set of representations of the translation group on Minkowski space as the space of solutions of the $4d$ Lorentzian harmonic oscillator. 
In this formulation,   \eqref{Map_OperatorCoordinates}, expressed as a combination of the ladder operators, provides an explicit infinite dimensional realization of the isomorphism $\M^4 \wedge \M^4 \cong \so^*(1,3)$. The elements of the dual Lorentz algebra (bivectors) are then given by the wedge product of creation and annihilation operators (translations). One can extend such a scheme to a general simple bivector. To this end, consider a pair of edge vectors:
\be
    e_1 = (\zeta_t,\zeta_1,\zeta_2,\zeta_3)
    \,\,,\quad
    e_2 = (\lambda_t,\lambda_1,\lambda_2,\lambda_3) \,.
    \label{Vectors}
\ee
One can associate to such edges the position and momentum (or translation) operators; these are realized as functional operators on the Hilbert space $L^2[\zeta]$ and $L^2[\lambda]$ of the edge vectors. Similarly to \eqref{LadderOperator_Tensor} when represented in the harmonic basis, elements living on these Hilbert spaces are expressed as a combination of the creation and annihilation operators of the harmonic oscillator as in \eqref{eq:edgevectors_creation_annihilationOP}.\\
According to the map \eqref{Map_OperatorCoordinates}, the annihilation operators $a$ are then associated with the generators of translations on Minkowski space, while the creation operators $\adag$ can be seen as the quantization of their dual momenta (position operators $\xi$ on Minkowski space).
In a similar manner, the wedge product of the two edge vectors $e_1 \wedge e_2$ can be assigned an operator that acts on the Hilbert space $L^2[\zeta,\lambda]$ that reads
\be
    b_{e_1 \wedge e_2} := 
    -i \adag_1 \wedge a_2 = 
    -i (\adag_1 a_2  - \adag_2a_1) \,.
    \label{QuantumBiVector}
\ee
The operator $b_{e_1 \wedge e_2}$ is understood as the quantization of the simple bivector $e_1 \wedge e_2$.
Note that for $\zeta_{\mu} = \delta_{\mu,0}$ and $\lambda_{\mu} = \delta_{\mu,i}$ we have $b_{e_1 \wedge e_2} = N_i$, while for $\zeta_{\mu} = \delta_{\mu,i}$ and $\lambda_{\mu} = \delta_{\mu,j}$ we instead obtain $b_{e_1 \wedge e_2} = \varepsilon_{ijk} \, L_k$.
Moreover, we can alternatively express the bivector operator \eqref{QuantumBiVector} in the $\so(1,3)$ basis, namely in terms of rotations and boosts as:
\be
    b_{e_1 \wedge e_2} = \sum_a \, \big(\alpha_a N_a + i\beta_a L_a\big) \,,
    \label{BiVector_Lorentz}
\ee
where
\be
    \alpha_a = (\zeta_t \lambda_a - \zeta_a \lambda_t)
    \,\,,\quad
    \beta_a = -i\varepsilon_{abc} \, \zeta_a \lambda_b \,.
    \label{BiVectorCoordinates_Lorentz-Translation}
\ee
\smallskip \\
According to \cite{CraneYetter}, the Hilbert space $L^2[\zeta,\lambda]$ associated with the quantum bivector (\ref{QuantumBiVector}) has to be invariant under the switching operator. We prove this fact for the bivector operator (\ref{BiVector_Lorentz}).

\begin{proposition}[Invariance under the switching operator]
    Consider a triplet of vectors $e_1,e_2,e_3$ on Minkowski space parametrized by the coordinates $\zeta,\lambda,\omega$, such that they form the boundary of a triangle: $\zeta + \lambda + \omega=0$.
    The Hilbert space of the bivector \eqref{QuantumBiVector} does not depend on which pair of edge vectors are used to construct the bivector \eqref{QuantumBiVector}. Therefore, it is invariant under the switching operator
    \be
        \sigma : L^2[\zeta,\lambda] \,\,\to\,\, L^2[\zeta,\omega]
        \qquad \text{ for } \,\,\,
        \zeta + \lambda + \omega=0\,,
    \ee
    up to a sign that reflects the orientation of the vector normal to the triangle.
\end{proposition}
\begin{proof}
     We show that the quantum bivector in \eqref{BiVector_Lorentz} is invariant under such switching operator and thus agrees with \cite{CraneYetter}. 
    According to the closure condition $\zeta + \lambda + \omega=0$, the third edge of the triangle with coordinates $\omega$ can be written as (up to a minus sign specifying the orientation)
    \be
        e_3 = (\omega_t\,,\omega_1\,,\omega_2\,,\omega_3) =
        (\zeta_t + \lambda_t,\zeta_1+\lambda_1,\zeta_2+\lambda_2,\zeta_3+\lambda_3) \,,
    \ee
    and a basis for the associated position and momentum operators is given by
    \be
        \begin{aligned}
            &
            \adag_3 = a_t \omega_t + \adag_x \omega_x + \adag_y \omega_y + \adag_z \omega_z \,,\\
            &
            a_3 = \adag_t \omega_t + a_x \omega_x + a_y \omega_y + a_z \omega_z \,,
        \end{aligned}
        \qquad\to \quad
        \adag_3 = \adag_1 + \adag_2 
        \,\,,\quad
        a_3 = a_1 + a_2 \,.
    \ee
    A direct computation shows that the quantization of the bivector $e_1 \wedge e_3$ is equivalent to that of $e_1 \wedge e_2$, since they obey the condition $e_1 + e_2 +e_3=0$. Thus, the associated Hilbert spaces are isomorphic (invariant under the switching operator).
    \be
        b_{e_1 \wedge e_3} :=
        -i (\adag_1 a_3 - \adag_3 a_1) =
        -i \big(\adag_1 (a_1 + a_2) - (\adag_1 + \adag_2) a_1\big) =
        -i (\adag_1 a_2 - \adag_2 a_1) :=
        b_{e_1 \wedge e_2} \,.
    \ee
\end{proof}

\noindent

Notice that such a realization provides the quantization of a triangle on Minkowski space with edges parametrized by the coordinates $\zeta,\lambda,\omega$. It also confirms that, once we restrict the Hilbert space of the quantum triangle to the skew-symmetric part, we establish the proper quantization of a bivector, which is given as the wedge product of any two of the triangle edge vectors.

\paragraph{Wave functions of quantum triangles and bivectors.}
Following the above general construction of Hilbert spaces and operators for a quantum triangle and a quantum bivector, we can now derive explicit expressions and some properties of their quantum states.
Let us then consider a pair of vectors $e_1,e_2$, parametrized respectively by the coordinates $\zeta,\lambda \in F(\M^4)$.
We refer to the \textit{triangle wave function} (or the quantization of the triangle $t$), as the function:
\be
    f(\zeta,\lambda) \,\,\in L^2[\zeta,\lambda] \,.
    \label{Triangle_WaveFunction}
\ee
  To ensure the appropriate encoding of the geometric data of a triangle, we recall that the two coordinates of the edge vectors must satisfy the closure relation $\zeta + \lambda +\omega=0$ for a given edge $e_3$ parametrized by $\omega$, such that the resulting wave function from the quantization procedure \eqref{Triangle_WaveFunction} is invariant under the switching operator \eqref{SwitchingOperator}. 

  These functions form a non-commutative algebra. Indeed, the proper quantization is realized upon specifying the quantization map that defines the star product between such functions of triangles, and that can be imposed already at the level of constituting ladder operators. An example is illustrated below \eqref{eq:quantizationmap_linear_ordering}.

\medskip
\textit{Example: normal ordering.} The algebra of a triangle, let it be $\mathrm{A}$, is given by the tensor product of the algebras of a pair of edge vectors and that the algebra of a single $4d$ edge vector can in turn be identified as the algebra of the $4d$ harmonic oscillator, as in \eqref{LadderOperator_Tensor}. Specifically, a basis for the algebra $\mathrm{A}$ of a triangle, endowed with the Hilbert space as in \eqref{Triangle_HilbertSpace}, is given by the set of tensor products between generators of harmonic oscillators, i.e. ${\adag_1 \adag_2 ,\, \adag_1 a_2 ,\, a_1 \adag_2 ,\, a_1 a_2}$ satisfying the closure constraint 
\be
    \cC : \qquad a_1 + a_2 + a_3 = 0 \,,\quad \adag_1 + \adag_2 + \adag_3 = 0 \,,
\ee
where we used the notation
\be
\label{eq:notation_creationannihilation_edgeVec}
    a_1 := a(\zeta) \,,\, \adag_1 := \adag(\zeta) 
    \,\,,\quad
    a_2 := a(\lambda) \,,\, \adag_1 := \adag(\lambda) \,,
\ee
and more precisely 

\be
\label{eq:edgevectors_creation_annihilationOP}
    \begin{aligned}
        & 
        \adag_1 := 
        a_t \zeta_t + \adag_x \zeta_x  + \adag_y \zeta_y  + \adag_z \zeta_z \,,\\
        & 
        a_1 := 
        \adag_t \zeta_t + a_x \zeta_x  + a_y \zeta_y  + a_z \zeta_z \,,
    \end{aligned}
    \qquad
    \begin{aligned}
        & 
        \adag_2 := 
        a_t \lambda_t + \adag_x \lambda_x  + \adag_y \lambda_y  + \adag_z \lambda_z \,,\\
        & 
        a_2 := 
        \adag_t \lambda_t + a_x \lambda_x  + a_y \lambda_y  + a_z \lambda_z \,.
    \end{aligned}
\ee
 Following the standard construction, a basis for the universal enveloping algebra $\cU(\mathrm{A})$ is given by the monomials on the product of the generators $\{\adag_1 \adag_2 ,\, \adag_1 a_2 ,\, a_1 \adag_2 ,\, a_1 a_2\}$. A choice of quantum representation for the universal enveloping algebra is then given by a choice of ordering for (products of) such generators. Hence, let $t(e_1,e_2)$ be a classical function on a triangle with the edge vectors $e_i$ being associated with the harmonic oscillator operators $a_i, \adag_i$ upon quantization, and satisfying the closure condition $e_1 + e_2 + e_3 = 0$. The Hilbert space $L^2[e_1,e_2]$ can thus be identified as the set of functions on the universal enveloping algebra, in turn expressed as linear combinations of the generators (ladder operators upon quantization).
The quantization map $\cQ(t)$ corresponding to normal ordering takes any such function ( with the set of coefficients $t_{h,l,m,n}$) into the operator:
\be
\label{eq:quantizationmap_linear_ordering}
    \cQ(t) \equiv \cQ(t(e_1,e_2)) = \sum_{h,l,m,n} \, 
    t_{h,l,m,n} \, : \, (\adag_1 \adag_2)^h \, (\adag_1 a_2)^l \, (a_1 \adag_2)^m \, (a_1 a_2)^n \, : \,.
\ee
where we used the notation $: \, :$ for the normal ordering. We also recall that these functions are subject to the closure constraint $e_1 + e_2 + e_3 = 0$, therefore the quantization map $\cQ$ will have to fulfill the condition $\cQ(t(e_1,e_2)) = \cQ(t(e_1,e_3))$. This request can be recast as a condition on the coordinate functions $t_{h,l,m,n}$, by substituting the $a_2 = - a_3 - a_1$ and $\adag_2 = - \adag_3 - \adag_1$, in the previous formula and imposing that the resulting combination of coefficients equals the original ones, thus having 
\be
    \cQ(t(e_1,e_3)) = \sum_{h,l,m,n} \, 
    t_{h,l,m,n} \, : \, (\adag_1 \adag_3)^h \, (\adag_1 a_3)^l \, (a_1 \adag_3)^m \, (a_1 a_3)^n \, : \,.
\ee

\noindent The quantization of a product of two classical triangle functions $t_1$, $t_2$ with the same quantization map is obtained by expanding those functions in as in \eqref{eq:notation_creationannihilation_edgeVec}, and then combining the two bases (given by two different monomials on the universal enveloping algebra) with the normal ordering:
\begin{align}
\label{eq:starproduct_triangles}
    \cQ(t_1 t_2) &
    =
    \sum_{h_1,l_1,m_1,n_1} \sum_{h_2,l_2,m_2,n_2} \, 
    t^1_{h_1,l_1,m_1,n_1} t^2_{h_2,l_2,m_2,n_2} \, 
    \\
    & : \, (\adag_1 \adag_2)^{h_1} \, (\adag_1 a_2)^{l_1} \, (a_1 \adag_2)^{m_1} \, (a_1 a_2)^{n_1} \,
    (\adag_1 \adag_2)^{h_2} \, (\adag_1 a_2)^{l_2} \, (a_1 \adag_2)^{m_2} \, (a_1 a_2)^{n_2} : \,.
\end{align}
which differs, in general, from the operator product of the operators obtained by separately quantizing the two classical functions, i.e. $\cQ(t_1) \cQ(t_2)$.

\noindent This difference, which affects how the classical algebraic structure of the space of triangle functions is preserved (or not) at the quantum level, is one key aspect of different choices of quantization maps. It is one of the reasons, for example, why the Duflo map is a preferred choice for quantizing classical functions over a Lie algebra, since it preserves the algebraic relations among invariant functions (Casimirs) of the Lie algebra. We will show a possible role for the Duflo map in this context below.

\noindent Once the quantization map $\cQ(t)$ is chosen (whatever it is), one can derive the star product on the quantum triangle, i.e. the non-commutative product of the Hilbert space $L^2[e_1,e_2]$ represented as a non-commutative space of triangle functions\footnote{In order to avoid introducing further symbols, we use the same ones we used for classical functions also for the wavefunctions. The context should avoid confusion.}. For generic wavefunctions $t_1$ and $t_2$, such product is {\it defined} as:
\be
    t_1 \star t_2 =
    \cQ\mone (\, \cQ(t_1) \, \cQ (t_2) \,) \,.
\ee
In the end, the star product of the Hilbert space of a triangle represented in terms of functions of the edge vectors is obtained from the combination of ordered products between ladder operators of the harmonic oscillators corresponding to the same edge vectors, and is determined by the chosen ordering.

\noindent Clearly, the explicit formula for the star product between generic functions can in principle be derived from this definition, but it involves rather tedious (and not particularly illuminating) calculations, even for rather simple functions.  

\noindent Other examples of quantization maps for the edge vectors (i.e. the 4d Lorentzian harmonic oscillators) can be chosen, and result in different non-commutative structures and star products at the level of functions of triangle data. 
A similar construction applies to quantum bivectors. 
We refer to the \textit{bivector wave function} (or the quantization of a bivector $b$), as the function:
\be
    g(\zeta,\lambda) \,\,\in T^A[\zeta,\lambda] = F(\M^4 \wedge \M^4) \,,
    \label{BiVector_WaveFunction}
\ee
that is the sub-class of the functions\footnote{We denote here the function of the bivector with $g(\zeta,\lambda)$. This should not be confused with the standard notation for a group element $g$ which will be clarified whenever it appears.} given in \eqref{Triangle_WaveFunction} corresponding to elements of the skew-symmetric part of the Hilbert space $L^2[\zeta,\lambda]$. 

\noindent By the same token, another formulation of the Hilbert space of the bivector can be characterized in terms of the (Lorentzian) Euler angles \eqref{BiVector_Lorentz}, rather than coordinates on the two vectors. The Hilbert space in this case takes the form $L^2[b] := L^2[\alpha,\beta] = F(\so^*(1,3))$, where the generators (the coordinates functions) are now the Euler angles $\{\alpha_a,\beta_a\} \in F(\so^*(1,3)) \cong F(\bbR^6_{\star})$ (here, the subscript $\star$ stands for a non-trivial (star) product for the functions on $\bbR^6$, given by the non-trivial product of the algebra $\so(1,3)$).
Obviously, the Hilbert space is not affected by the choice of coordinates: the two formulations are isomorphic, as $\M^4 \wedge \M^4 \cong \so^*(1,3)$ and related by the map \eqref{BiVectorCoordinates_Lorentz-Translation}. 

\smallskip

This is also a non-commutative space of functions. For defining the relevant star product one can proceed analogously to what we have seen for triangle wavefunctions. However, one could also take advantage of the equivalent formulation in terms of Lorentz algebra,  i.e. the isomorphism $\M^4 \wedge \M^4 \cong \so^*(1,3)$, to choose a quantization map motivated by this algebraic structure, and derive from it the associated star product among wavefunctions of geometric bivectors. The resulting star product would implicitly correspond to a certain quantization map at the level of products of edge vectors (thus, of triangles)\footnote{Obviously, a star product for (functions of) bivectors would not uniquely determine a star product for more general combinations (functions of) edge vectors.}. See \cite{FinocchiaroOriti:2018SFDufloMap, OritiRosati:2018GFTNonCommFourier} for the construction at this Lie algebra level for the 4-rotation and Lorentz algebras, and \cite{Guedes:2013vi} for $SU(2)$.
\noindent

\noindent A convenient choice of quantization map for (semi-simple) Lie algebras is the Duflo map. To use it to induce a star product for functions of bivectors constructed out of edge vectors (thus not generic bivectors\footnote{The bivectors constructed out of edge vectors correspond to the ones obtained after imposition of the geometricity (including simplicity) constraints on generic bivectors.}), let us call $\cL$ the operator that maps the classical functions on the algebra of triangles (specifically, the skew-symmetric elements, i.e. geometric bivectors) in the classical functions on the Lorentz algebra:
\be
    \cL : \mathrm{M}^4 \wedge \mathrm{M}^4 \,\to\, \so^*(1,3) \,.
\ee
In more concrete terms, the action of the operator $\cL$ is
\be
    \cL[t(e_1,e_2)] = t(b) \,, \quad \text{where} \quad b = - i e_1 \wedge e_2 \,.
\ee
Using such an operator together with the Duflo map $\cD$, the quantization map of a geometric bivector stemming from the Duflo map is simply given by the combination:
\be
    \cQ \equiv \cD \circ \cL \,.
\ee
Here we first use the operator $\cL$ to express the function of a geometric bivector as a function of the Lorentz algebra, then we quantize the function on the Lorentz algebra using the Duflo map $\cD$. The inverse map $\cD\mone\cL\mone$ applied to generic operators in the algebra of quantum bivectors would give us then a non-commutative function of bivectors constructed from quantum edge vectors, i.e. a function on the Hilbert space $L^2[b] \cong L^2[e_1 \wedge e_2]$. \\
The corresponding star product is induced using a combination of the above maps. Let $f(e_1,e_2)$ and $g(e_1',e_2')$ be two functions on triangles satisfying $e_1 + e_2 + e_3 = 0$ and $e_1' + e_2' + e_3' = 0$. Let also $b = b_{e_1 \wedge e_2}$ and $b' = b_{e_1' \wedge e_2'}$ be the associated bivectors. The star product between the two functions can be defined as:
\be
    f \star g =
    \cL\mone \cD\mone(\cD(\cL(f)) \cD(\cL(g))) \,.
\ee

\smallskip
\noindent The identification of the bivectors as non-commutative geometric entities,  allows us to make contact with the non-commutative flux representation of quantum states and dynamical amplitudes in loop quantum gravity, spin foam models, and group field theories. However, as emphasized in section \ref{Sec_InfiniteDimReps}, this is obtained starting from commutative functions of commutative edge vectors (which as we have seen are elements of the translation group), which are taken as primary or more fundamental variables. 
\smallskip 
\noindent We can use the operator associated with the bivector constructed in \eqref{QuantumBiVector} to expand the wave function in Fourier modes. This expansion can be realized both in terms of the edge vectors or in the Euler angle parametrization:
\be
    \begin{aligned}
        \textit{Vectors:} \quad
        & f(\lambda_1,\lambda_2) := -i
        \int \dd \alpha \dd\alpha' \,
        \la \alpha_{\nu} \,| \adag(\lambda_1) a(\lambda_2) - \adag(\lambda_2) a(\lambda_1) |\, \alpha_{\nu}' \ra \,  f_{\alpha_{\nu}\,,\alpha_{\nu}'} \,,\\
        \textit{Euler:} \quad
        & f(\alpha,\beta) := 
        \sum_{n_r=0}^{\infty} \int \dd \mu \mu^2 \, 
        \sum_{\ell,\ell'} \sum_{m,m'} \,
        \la n_r, \mu ; \ell, m \,| \sum_a \, \big(\alpha_a \, N_a + i\beta_a \, L_a \big) |\, n_r, \mu ; \ell', m ' \ra \, 
        f^{n_r, \mu ; \ell, m}_{n_r, \mu ; \ell', m'} \,,
    \end{aligned}
    \label{BiVecWaveFunc_FourierDecomposition}
\ee

with $\ell,\ell' \in [0,\infty)$, $\ell \leq m \leq \ell$ and $\ell' \leq m' \leq \ell'$. These are general expressions for the decomposition of triangle wavefunctions; however, for timelike bivectors thus spacelike triangles only special group elements are used (those restricted to the hyperboloid) and thus only some of the matrix elements are non-vanishing; specifically, one has only balanced representations with $n=0$ contributing.  The terms $\la \alpha_{\nu} \,| \adag(\lambda_1) a(\lambda_2) - \adag(\lambda_2) a(\lambda_1) |\, \alpha_{\nu}' \ra$ and $\la n_r, \mu ; \ell, m \,| \sum_a \, \big(\alpha_a \, N_a + \beta_a \, L_a \big) |\, n_r, \mu ; \ell', m ' \ra$ 
appearing in \eqref{BiVecWaveFunc_FourierDecomposition} play the role of non-commutative plane waves in the Fourier transform. 
An explicit expression for the bivector wave function \eqref{BiVector_WaveFunction} is obtained by the infinite dimensional representations of such plane waves. Introducing the compact notations $b(\lambda_1,\lambda_2) = -i(\adag(\lambda_1) a(\lambda_2) - \adag(\lambda_2) a(\lambda_1))$ and $b(\alpha,\beta) = \sum_a \, \big(\alpha_a \, N_a + i\beta_a \, L_a \big)$, we end up with the following expression for the bivector operator:
\be
    \begin{aligned}
        \la \alpha_{\nu} \,| b(\lambda_1,\lambda_2) |\, \alpha_{\nu}' \ra & =
        \int \dx^4 \, \overline{\Psi}_{\alpha_{\nu}}(t,x,y,z) \, \Psi_{\alpha_{\nu}'}(t',x',y',z')
        \,,\\
        \la n_r, \mu ; \ell, m \,| b(\alpha,\beta) |\, n_r, \mu ; \ell', m ' \ra & =
        \int \dd r \dd \eta \dd \Omega \, r^3 \sinh^2(\eta) \,
        \overline{\Psi}_{n_r, \mu ; \ell, m}(r,\eta,\theta,\phi) \, \Psi_{n_r, \mu ; \ell', m'}(r',\eta',\theta',\phi') \,,
    \end{aligned}
    \label{BiVector_Decomposition_PlaneWave}
\ee
where the primed coordinates encode the action of the operator $b$, in the basis \eqref{QuantumBiVector} or \eqref{BiVector_Lorentz} in terms of \eqref{HarmonicOscillator_CoherentMinkowskian} and \eqref{HarmonicOscillator_HyperbolicCoherent}. 
In section \ref{sec:Plancherel} we will relate the above expansions to the standard Plancherel decomposition for functions on the Lorentz group, in order to provide an explicit relation between the amplitudes of the new model based on edge vectors to those of usual spin foam models based on bivectors and Lorentz group elements.
We further emphasize that the plane waves appearing in the above Fourier decomposition can simultaneously be recast in terms of the eigenfunctions of the harmonic oscillator expressed in the Minkowski basis \eqref{MinkowskiBasis} (as well as in the hyperbolic basis \eqref{HyperbolicBasis}), and likewise in terms of the eigenfunctions of the coherent states. This can be made evident using the relations \eqref{HarmonicOscillator_CoherentMinkowskian}, \eqref{HarmonicOscillator_HyperbolicCoherent} or their inverses. \\
For what will come later, it is important to stress that the bivector wave functions, as elements of the algebra $\bbR^6 \cong \so^*(1,3)$, are equipped with a sum:
\be
    f(b+b') = \,\,
    \left\{\,\,
    \begin{aligned}
        & f(\lambda_1+\lambda_1',\lambda_2+\lambda_2') \,,\\
        & f(\alpha_1+\alpha_1',\beta_2+\beta_2') \,.
    \end{aligned}
    \right.
\ee
It is also interesting to recall that the Lie algebra $\so^*(1,3)$ is endowed with a nontrivial Poisson structure; indeed, as a quantum group, it is isomorphic to $\bbR^6_{\star}$, which is endowed with a $\star$ product. 
In the Fourier decomposition, this product appears as a non-trivial combination of plane waves \eqref{BiVector_Decomposition_PlaneWave}.

\section{Quantum tetrahedron}
\label{Sec_QuantumTetrahedron}

Now we are equipped with the essential material to provide a quantum description of a tetrahedron, both in terms of edge vectors and of bivectors. This is the building block of quantum geometry in spin foam models and group field theories, as well as in the simplicial sector of canonical LQG (where it corresponds to open 4-valent spin network vertices). It is also the starting point for the quantum geometric construction of spin foam amplitudes. In the following we first recall the classical description; then, we provide its quantum counterpart and the definition of quantum states.  

\subsection{Classical tetrahedron} 
\paragraph{Classical tetrahedron via edge vectors.} Let us consider six edge vectors $e_a \in \M^4$, for $a=1,\dots,6$, that satisfy the four closure conditions:
\be
    e_1 + e_2 + e_3 = 0
    \,\,,\quad
    e_4 + e_5 + e_3=0
    \,\,,\quad
    e_1 + e_5 +e_6=0
    \,\,,\quad
    e_2 + e_4 + e_6 = 0
    \,.
    \label{FourTriangles}
\ee
They specify the full geometry of the tetrahedron.
The four closure conditions are not independent of each other: any one of them can be written as a linear combination of the other three relations. 
Thus, the geometry of a tetrahedron is actually completely determined by three of the six edge vectors, attached to a common vertex, such that, taken in pairs, they enter three of the four closure conditions \eqref{FourTriangles}  \cite{Baez1999, Barbieri_1998SN_Simpli}. 
For instance, using Fig.(\ref{Fig:Tetrahedron}), the geometry of a tetrahedron $\tau_e$ in Minkowski space is encoded in the vector triplet $e_1,e_3,e_5 \subset \tau_e$ with the set of three conditions (as in Fig.(\ref{Fig:QuantumTetrahedron})):
\be
\label{threeconditions}
    e_1 + e_2 + e_3 = 0
    \,\,,\quad
    e_4 + e_5 + e_3=0
    \,\,,\quad
    e_1 + e_5 + e_6 =0\,,
\ee 
\begin{figure}
    \centering
    \begin{subfigure}[b]{0.475\textwidth}
        \input{Tetrahedron}
        \caption{Combination of four triangles into a tetrahedron. The four triangles are labeled by $t_i$ and the six edges
(shared pairwise by triangles) by $e_j$.}
        \label{Fig:Tetrahedron}
    \end{subfigure}%
    \begin{subfigure}[b]{0.475\textwidth}
        \centering
        \input{QuantumTetrahedron}
        \caption{Tetrahedron from three edge vectors, $e_1,e_3,e_5$.}
        \label{Fig:QuantumTetrahedron}
    \end{subfigure}
    \caption{Combinatorics of $\tau$ and its closure constraints from edge vectors.}
\end{figure}
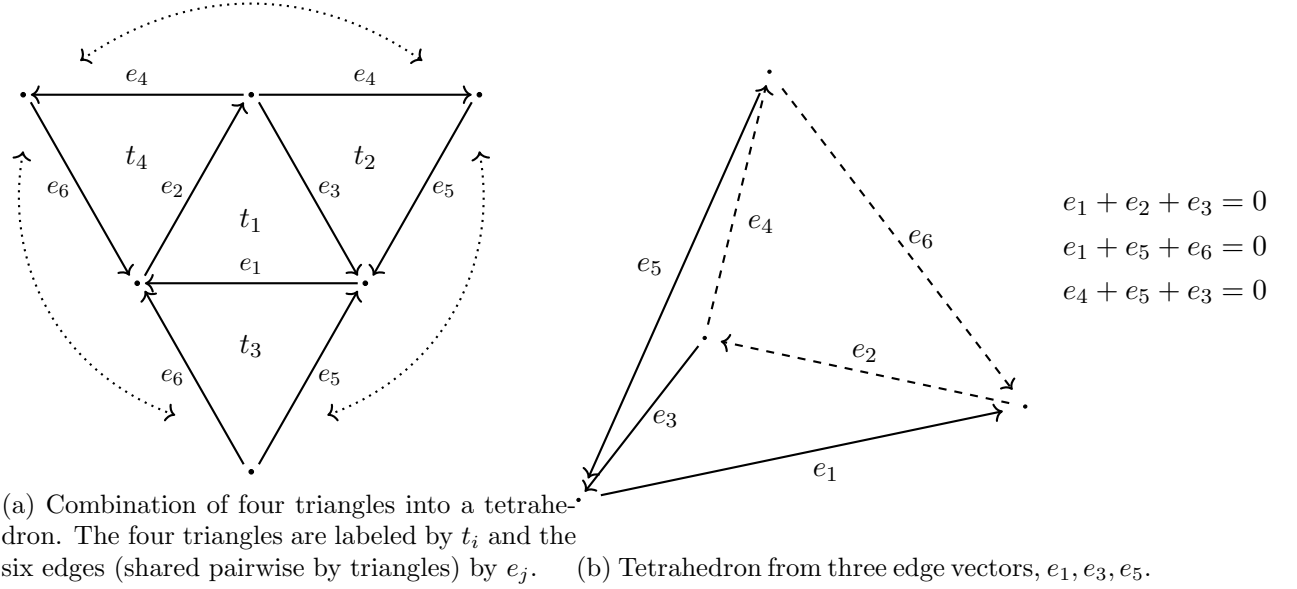
Note that the \textit{normal vector} to the tetrahedron in its $4d$ embedding can be obtained easily in this edge-based formulation. Indeed, we can express it as the Hodge dual of the wedge product of any triplet of independent edge vectors, and this yields:
\be
    n_{\tau_e} = \ast(e_i \wedge e_j \wedge e_k)
    \,\,,\quad \text{ for } \quad
    e_i + e_j + e_k \neq 0 \,.
    \label{TetrahedronNormal_Vectors}
\ee
The existence of the normal $n_{\tau_e}$ is ensured by the fact that, in three dimensions, there exists always a triplet of independent edge vectors, thus the wedge product \eqref{TetrahedronNormal_Vectors} does not vanish.

\paragraph{Classical tetrahedron: bivectors.}
Alternatively to the above description based on the edge vectors, one can characterize a classical tetrahedron $\tau_b$ in Minkowski space in terms of the four bivectors $b$. This description is obtained from the skew-symmetric geometry of the four triangles specified by the closure conditions in \eqref{FourTriangles}. In order to properly encode the geometric data in the tetrahedron $\tau_b$, the four bivectors have to satisfy the following two constraints:
\begin{enumerate}
    \item \textit{Dependence relation}: the wedge product of each pair of bivectors $b_i,b_j \subset \{b_1,b_2,b_3,b_4\}$ vanishes: $b_i \wedge b_j = 0$;
    \item \textit{Closure relation}: each of the four bivectors is given by the sum of the other three: $b_1 + b_2 + b_3 + b_4 = 0$. 
\end{enumerate}
The first condition ensures that each bivector shares one and only one vector with each of the others (thus it is indeed a simple bivector constructed from the edge vectors of a geometric tetrahedron), while the second property is a direct consequence of the fact that the tetrahedron geometry is encoded in the closure of the four triangles \eqref{FourTriangles}.
\\
We point out the crucial difference between the geometric construction exploiting the full data of the triangle, and that based solely on bivectors. One can observe that the bivector-based tetrahedron $\tau_b$ encodes less geometric information with respect to the edge-based one $\tau_e$. Specifically, one cannot reconstruct the normal vector from the bivectors alone, the knowledge of which on the other hand is relevant for understanding how the extrinsic geometric data transform under $4d$ Lorentz transformations in the embedding Minkowski space.
This becomes evident when we start with the edge vectors formulation and recover the condition on the bivectors. Let us see this explicitly. Given the four triangles specified by the closure conditions \eqref{FourTriangles}, one can define the four bivectors as:
\be
    b_1 = e_2 \wedge e_3 \,\,,\quad
    b_2 = e_3 \wedge e_5 \,\,,\quad
    b_3 = e_6 \wedge e_5 \,\,,\quad
    b_4 = e_2 \wedge e_6 \,\,
    \label{BiVectors_Tetrahedron}
\ee
The first relation (the dependence relation) for the construction of $\tau_b$ is easily established from edge vectors, once we rely on the fact that at most a triplet of the six edge vectors $e_a$ are linearly independent in four dimensions, namely:
\be
    b_i \wedge b_j =
    (e_a \wedge e_b) \wedge (e_b \wedge e_c) =
    e_a \wedge (e_b \wedge e_b) \wedge e_c = 0 \,,
\ee
for $a \neq b$, $a \neq c$ and $b \neq c$.
As for the closure relation, the closure of the four triangles \eqref{FourTriangles} implies the gluing of the boundary of the tetrahedron. This is translated into the fact that the four bivectors \eqref{BiVectors_Tetrahedron} sum up to zero:
\begin{align}
    b_1 + b_2 + b_3 + b_4 & =
    e_2 \wedge e_3 + e_3 \wedge e_5 + e_6 \wedge e_5 + e_2 \wedge e_6
    \nonumber \\
    & =
    (e_2 - e_5) \wedge e_3 + (e_2 - e_5) \wedge e_6 =
    (e_2 - e_5) \wedge (e_3 + e_6)
    \nonumber \\
    & =
    (e_2 - e_5) \wedge (e_2 - e_5) = 0 \,.
    \label{ClosureBiVectors}
\end{align}
We thus recover the bivector picture of the tetrahedron $\tau_b$. This is obtained by imposing further restrictions on the edge vectors in order to access the skew-symmetric sector of the geometry described by the bivectors. However, the converse clearly does not work. 
There is no notion of the normal vector \eqref{TetrahedronNormal_Vectors} in terms of such entities. In turn, this is needed to ensure $4d$ covariance of the bivector description (including the constraints ensuring that they come from edge vectors) and for controlling the full four-dimensional geometry (thus, the gluing of 4-simplices). A similar point was raised in \cite{BaratinOriti:2011GFTBarrettCrane, Jercher:2022mky}, where an extended formulation of the Barrett-Crane model was realized to include explicitly the normal vector.

\subsection{Quantum tetrahedron}
We begin by outlining the general procedure for defining the Hilbert space of a tetrahedron using quantum edge vectors and bivectors.  
Furthermore, we define the fundamental properties and characteristics of the tetrahedron wave function.
\paragraph{Quantum tetrahedron from edge vectors.}
 As we discussed, the geometry of a tetrahedron is completely determined by three of its edge vectors meeting at one of the four vertices, where the closure conditions \eqref{threeconditions} must hold. Identifying the quantum space of a single edge vector with the space of (square-integrable) functions on the translations group, we naturally define the quantum space of a tetrahedron $\tau_e$ in Minkowski space as the tensor product of three of them, equipped with the proper closure conditions:
\be
    L^2[e_1,e_3,e_5] := L^2[e_1] \ot L^2[e_3] \ot L^2[e_5] 
    \,\,,\quad
    \left\{\,\,
    \begin{aligned}
        & e_1+e_2+e_3=0 \,,\\
        & e_4+e_5+e_3=0 \,,\\
        & e_1+e_5+e_6=0 \,.
    \end{aligned}
    \right.
    \label{QuantumTetrahedron_Hilbert}
\ee
Here the tensor product of each pair of spaces $L^2[e_i,e_j]$ for $i\neq j=1,3,5$ is the quantum space of a triangle (see section \ref{Sec_QuantumTriangle}), and it is automatically invariant under the switching operator \eqref{SwitchingOperator}. 
As was the case for the individual triangles, one can also obtain the same Hilbert space by starting from the larger symmetric Hilbert space of six edge vectors and \lq \lq project down\rq\rq by imposing the closure constraints. The latter enforces indeed invariance under a diagonal action of the translation group. However, as noted earlier, the non-compact nature of the translation group makes this definition viable but rather formal, since the constraints imposition does not correspond to well-defined projection operators and the resulting amplitude would require regularization. The situation is not much different, though, from the one concerning the imposition of gauge invariance under the Lorentz group, and corresponding closure condition, in usual spin foam amplitudes, where the resulting divergences can easily be regularized. Similarly, and subject to the same cautionary remarks, one can obtain the Hilbert space of the quantum tetrahedron by tensoring the Hilbert spaces of its four quantum triangles, each defined in terms of edge vectors, and then projecting down again with respect to the relevant closure conditions.

\medskip

\paragraph{Quantum tetrahedron from bivectors} One can reconstruct $\tau_b$ in terms of the three bivectors by reducing the degrees of freedom of each triangle $t$ to its skew-symmetric part. At the level of Hilbert spaces, this translates into 
\be
    T^A[e_1,e_2,e_3] := T^A[e_1,e_3] \ot T^A[e_1,e_5] \ot T^A[e_3,e_5] \cong F(\so^*(1,3))^{\times 3} \,.
    \label{QuantumTetrahedron_Hilbert_SkewSymm}
\ee
However, as we already observed at the classical level,  some geometric information, namely the data of the normal vector to  $\tau_b$, is not encoded in the skew-symmetric part of the Hilbert space given in \eqref{QuantumTetrahedron_Hilbert_SkewSymm}. It is instead included in the sub-space of \eqref{QuantumTetrahedron_Hilbert} obtained by the triple anti-symmetric product:
\be
    L^2[n_\tau] = L^2[e_1] \wedge L^2[e_3] \wedge L^2[e_5] 
    \subset L^2[e_1,e_3,e_5] 
    \,\,,\quad
    L^2[n_\tau] \nsubseteq T^A[e_1,e_2,e_3] \,,
\ee

i.e. it can be reconstructed from quantum edge vector data.

\paragraph{Tetrahedron wave function.}
We can now present the construction of the wave functions associated with the tetrahedron $\tau_e$ (as well as $\tau_b$)\footnote{Provided we specified beforehand the quantization map as the example mentioned in section \ref{Sec_QuantumTriangle}, we can write down the wave function of the tetrahedron and products thereof, which will enter the new spin foam amplitude.}. 

These functions are the elements of the Hilbert space \eqref{QuantumTetrahedron_Hilbert}, and therefore are defined as functions on three copies of the translation group on Minkowski space:
\be
    f(\lambda_a,\lambda_b,\lambda_c) \,\,\in F(\M^4)^{\times 3} \,,
    \label{TetrahedronWaveFunction_Edges}
\ee
where the coordinates of three edges of the tetrahedron meeting at one of its four vertices satisfy the closure conditions for three different triangles, e.g. $a=1$, $b=3$ and $c=5$ in the notation of Fig.(\ref{Fig:Tetrahedron}) and Fig.(\ref{Fig:QuantumTetrahedron}). 
\\
As mentioned, a more symmetric way to construct these functions is to extend the Hilbert space \eqref{QuantumTetrahedron_Hilbert} and consider as a first step a set of six edges $e_i$, for $i=1,\dots,6$. This set of vectors should close to form four triangles according to the constraints \eqref{FourTriangles}.
Therefore, the elements of this extended Hilbert space would be functions of the form $f(\lambda_1,\lambda_2,\lambda_3,\lambda_4,\lambda_5,\lambda_6) \in F(\M^4)^{\times 6}$ subject to the quantum closure constraint:
\be
    \hat{\cC}_t(\lambda_1,\, \dots,\, \lambda_6) =
    \delta_\star(\lambda_1 + \lambda_2 + \lambda_3) \, 
    \delta_\star(-\lambda_3 + \lambda_4 + \lambda_5) \,
    \delta_\star(-\lambda_5 + \lambda_6 - \lambda_1) \,,
    \label{ClosureConstraint}
\ee
where the subscript $t$ stands for the closure of three of the four triangles $t$ of $\tau_e$ and this way, the closure of the fourth triangle is automatically obtained.

Note that the delta function is the non-commutative delta function, acting as  the standard delta distribution inside integration when star-multiplied to any function, with the star product following from the choice of the quantization map, obtained earlier at the level of quantization of triangles. 

\smallskip 
\noindent As in the classical prescription, by reducing the Hilbert space \eqref{QuantumTetrahedron_Hilbert} to its skew-symmetric part \eqref{QuantumTetrahedron_Hilbert_SkewSymm}, we obtain the standard formulation of a quantum tetrahedron in terms of bivectors, which are however automatically simple, i.e. functions of the (quantum) edge vectors.  
The same remarks about the underlying non-commutative structure at the level of bivectors (thus of a flux representation, as developed for canonical loop quantum gravity,  spin foam models and group field theories \cite{BaratinOritimodel,BaratinDittrichOritiTambornino:2010LQG,BaratinOriti:2011GFTBarrettCrane,BaratinOriti:2010GFTNonCommutativeMetric}) apply here. \\ 
One can consider a set of three bivectors, writing down the wave function for $\tau_b$ as an element of the quantum space \eqref{QuantumTetrahedron_Hilbert_SkewSymm}:
\be
    f(x_1,x_2,x_3) \,\,\in F(\so^*(1,3))^{\times 3} \,,
    \label{TetrahedronWaveFunction_BiVectors}
\ee
where the three bivectors coordinates obey the closure condition $x_1 + x_2 + x_3 + x_4 = 0$ for some $x_4 \in F(\so^*(1,3))$.
The coordinates (the variables $x$) of the bivectors can be identified, for instance using \eqref{BiVector_Lorentz}, as the Euler angles of the Lorentz group:
\be
    x_i := \{\alpha^a_i,\,\beta^a_i\} \,\, \in F(\bbR^6_{\star}) \cong F(\so^*(1,3)) \,.
\ee
Once again, the same expression for the tetrahedron wave function can be achieved by extending the Hilbert space and considering a constrained sub-space implementing the skew-symmetry condition. In this scenario, we work on four copies of the space of functions in the dual Lorentz algebra, whose elements are the functions $f(x_1,x_2,x_3,x_4) \in F(\so^*(1,3))^{\times 4}$ constrained by the condition $(\hat{\cC}_{\tau_b} \star f) $. This imposes the following closure constraint on the bivector coordinates\footnote{Recall that the subscript $\tau_b$ indicates that this constraint enforces the closure of the boundary of the tetrahedron $\tau_b$.}:
\be
    \hat{\cC}_{\tau_b}(x_1,x_2,x_3,x_4) = \delta_\star(x_1+x_2+x_3+x_4) \,.
    \label{ClosureConstraint_Triangles}
\ee
 As we already indicated in the discussion of the classical tetrahedron, the constraint on the bivector  \eqref{ClosureConstraint_Triangles} is inherited by the closure of the triangles in \eqref{ClosureConstraint}. Similarly, we can identify the sub-space of $F(\so^*(1,3))^{\times 4}$, constrained by the condition in  \eqref{ClosureConstraint_Triangles}, as a sub-space of the extended tetrahedron Hilbert space  $L^2[e_1,e_2,e_3,e_4,e_5,e_6]$ obeying the condition \eqref{ClosureConstraint}, i.e. of what we defined as the tetrahedron Hilbert space:
\be
    (\hat{\cC}_t \star f)(\lambda_1,\, \dots,\, \lambda_6) 
    \,\,\,\Rightarrow\,\,\,
    (\hat{\cC}_{\tau_b} \star f)(x_1,\, x_2,\, x_3,\, x_4) \,.
    \label{Tetrahedron_WaveFunction_BiVectors}
\ee
As already pointed out, the above relation is not an equivalence, since the set of edge vectors contains {\it more} information about the geometry of the tetrahedron and its embedding in four dimensions than the set of simple bivectors.\\ This fact may require some further discussion. The quantum states of tetrahedra depending on the full set of (independent) edge vectors contain in fact more information than that corresponding to the intrinsic quantum geometry of a tetrahedron as the possibility of reconstructing the $4d$ normal, thus the embedding information in $4d$ Minkowski space, indicates. This extra information is needed, we argue, to ensure the proper gluing of 4-simplices across shared tetrahedra to define proper four-dimensional quantum geometries, and it may also exclude degenerate bivector geometries (this was the original motivation for the proposal in \cite{CraneYetter}). At the level of quantum states, this redundancy from the point of view of intrinsic geometric data is reflected in the fact that our quantum states functions of edge vectors form an overcomplete basis of the Hilbert space of a quantum triangle, akin to coherent states. It is also reflected in the fact that they define non-commutative representations of the Hilbert space, rather than one in terms of proper eigenstates of (a maximally commuting set of) geometric operators. This will play an explicit role in the definition of the amplitudes for the whole simplicial complex and for relating the same to the usual spin foam amplitudes in terms of irreps of the Lorentz group.\\
One can expand the wave function for the tetrahedron $\tau_b$  \eqref{TetrahedronWaveFunction_BiVectors} in the Fourier decomposition as we did for the bivector wave function in \eqref{BiVecWaveFunc_FourierDecomposition}. Let us introduce the compact notation $T_{\alpha_{\nu}}^{\alpha_{\nu}'}(\lambda_1,\lambda_2)$ for the translation group plane waves $\la \alpha_{\nu} \,| \adag(\lambda_1) a(\lambda_2) - \adag(\lambda_2) a(\lambda_1) |\, \alpha_{\nu}' \ra$. The decomposition of \eqref{BiVecWaveFunc_FourierDecomposition} then reads:
\be
    f(\lambda_1,\lambda_2) = -i \int \dd \alpha \dd \alpha' \, 
    T_{\alpha_{\nu}}^{\alpha_{\nu}'}(\lambda_1,\lambda_2) \, f^{\alpha_{\nu}}_{\alpha_{\nu}'} \,,
\ee
Hence, the tetrahedron $\tau_e$ wave function \eqref{TetrahedronWaveFunction_BiVectors} decomposes in the Fourier modes as:
\begin{align}
    f(\lambda_1,\, \dots,\, \lambda_6) &=
    \int \prod^6\dd \alpha\prod^6\dd {\alpha'}^6 \,
    T_{\alpha_{1;\nu}}^{\alpha_{1;\nu}'}(\lambda_2,\lambda_3) \,
    T_{\alpha_{2;\nu}}^{\alpha_{2;\nu}'}(\lambda_3,\lambda_5) \\ \nonumber
    & T_{\alpha_{3;\nu}}^{\alpha_{3;\nu}'}(\lambda_6,\lambda_5) \,
    T_{\alpha_{4;\nu}}^{\alpha_{4;\nu}'}(\lambda_2,\lambda_6) \,  f^{\alpha_{1;\nu},\,\alpha_{2;\nu},\,\alpha_{3;\nu},\,\alpha_{4;\nu}}_{\alpha_{1;\nu}'\,\alpha_{2;\nu}'\,\alpha_{3;\nu}'\,\alpha_{4;\nu}'} \,.
    \label{TetraWaveFunct_FourierDecomposition}
\end{align}
where the different labels appearing in $T_{\alpha_{\nu}}^{\alpha_{\nu}'}(\lambda_i,\lambda_j)$ are those that decorate the quantum vectors obeying a certain combinatorial relation (this is dictated by the constraint \eqref{ClosureConstraint}). This is of course encoded in the combinations of the ladder operators and can be exactly extracted by evaluating the matrix elements appearing in \eqref{BiVector_Decomposition_PlaneWave} (relying as well on the material in the appendix \eqref{App_4dHarmonicOscillator}). In section \ref{Sec_SF-EdgeVectors} we will see how this is realized for the example of a timelike bivector. In section \ref{sec:Plancherel}
we will also relate these expressions with those involving Lorentz group representations in a Plancherel decomposition, as the key step for a more traditional spin foam expression of the quantum amplitudes of the new model.

\section{Quantum amplitudes for simplicial 4-geometries}
\label{Sec_SF-EdgeVectors}
Now we can finally provide the explicit description of the (spin foam) amplitudes for four dimensional Lorentzian (simplicial) geometries based on quantum edge vectors. To this end, we work with the quantum space of the tetrahedron $\tau_e$ derived in section \ref{Sec_QuantumTetrahedron} expressed in terms of the quantum edge vectors, as the basic building block.  
Starting from the explicit expression of the spin foam amplitudes, we then show how they can be also expressed as a combination of the Barrett-Crane amplitudes. This is done only in the case of a spacelike tetrahedron, although the more general construction can be obtained straightforwardly using the tools developed in section \ref{Sec_InfiniteDimReps}. We finally give the definition of the full spin foam model (including a sum over simplicial complexes encoding the continuum limit) in the GFT formalism in section \ref{SubSec_GFTEdgeVectors} (the field theory is defined over several copies of the translation group with the proper closure conditions imposed on the triangles and tetrahedra). 
Note that, in many of the following formulae we have chosen to indicate explicitly only a subset of star products as such, while leaving some products of non-commutative functions indicated with a standard product symbol, in order to avoid cluttering the notation. The nature of the functions should make clear what is the nature of their products, when they appear.  
\subsection{The quantum amplitude}
\label{subsectio:chapter_newspinfoam}
The quantum amplitude for a general four-dimensional triangulation $\Gamma$ is defined as the integral over all allowed configurations of its edge vectors, with 4-simplex contributions encoding the conditions for an allowed edge vector geometry of each of them, and tetrahedral contributions encoding the needed identification of edge vectors across 4-simplices sharing them:
\be
    \cA_{\Gamma} = \int [\dd \lambda] \,\prod_{\tau} \, \cA_{\tau}[\{\lambda\}_{\tau},\{\lambda'\}_\tau]\,  \prod_{s} \, \cA_{\s}[\{\lambda\}_s]\,.
    \label{NewSFAMplitude}
\ee

The amplitude of a single 4-simplex is:
\begin{align}
    \cA_{\s}[\{\lambda\}_s] 
  &  =  
    \prod_{\alpha=1}^{3} \prod_{a=1}^{4-\alpha} \hat{\cC}_{t_{\alpha;a}}\left(\{\lambda_{\alpha;i}\})\right)\\
    &\star
    \big(\delta_\star(\lambda_{1;1}-\lambda_{2;4} )\star \delta_\star(\lambda_{1;1}-\lambda_{4;5})\big) \star\big(\delta_\star(\lambda_{1;2}-\lambda_{2;6}) \star \delta_\star(\lambda_{1;2}-\lambda_{5;3})\big)\nonumber\\
&\star\big(\delta_\star(\lambda_{1;3}-\lambda_{2;2}) \star \delta_\star(\lambda_{1;3}-\lambda_{3;6})\big) \star\big(\delta_\star(\lambda_{1;4}-\lambda_{3;5}) \star \delta_\star(\lambda_{1;4}-\lambda_{5;1})\big) \nonumber\\ 
&\star\big(\delta_\star(\lambda_{1;5}-\lambda_{3;1}) \star \delta_\star(\lambda_{1;5}-\lambda_{4;4})\big) \star\big(\delta_\star(\lambda_{1;6}-\lambda_{4;3} )\star \delta_\star(\lambda_{1;6}-\lambda_{5;2})\big) \nonumber\\ 
&\star \big(\delta_\star(\lambda_{2;1}-\lambda_{3;4}) \star \delta(\lambda_{2;1}-\lambda_{5;5})\big)\star \big(\delta_\star(\lambda_{2;3}-\lambda_{3;2} )\star \delta_\star(\lambda_{2;3}-(\lambda_{4;6})\big)\nonumber\\
&\star\big(\delta_\star(\lambda_{2;5}-\lambda_{4;1}) \star \delta_\star(\lambda_{2;5}-\lambda_{5;4})\big) \star\big(\delta_\star(\lambda_{3;3}-\lambda_{4;2}) \star \delta_\star(\lambda_{3;3}-\lambda_{5;6})\big)\,,    \label{4Simplex_Amplitude}
\end{align}
where $\lambda_{\alpha;i}$ is the  the $i^{th}$ displacement vector  associated with the edge $i$ of the tetrahedron $\alpha$, for $\alpha=1,\dots,5$ and $i=1,\dots,6$, and $t_{\alpha;a}$ stands for the triangle $a$ of the tetrahedron $\alpha$, for $a=1,2,3,4$. Furthermore, the $\star$ symbol refers to the non-commutative product of the functions on $\M^4$, once an appropriate quantization map is specified. 
It is important to stress that the above expression is written as a linear concatenation of (non-commutative) functions only for easiness of notation. A non-trivial ordering for the convolution operations, which is {\it not} a linear concatenation, has to be chosen and affects the result. Our choice of ordering is expressed graphically in Fig.(\ref{Fig:4-simplex}). 
The above simplex amplitude is simply a combination of the closure constraints for its five tetrahedra, in turn, given by closure constraints, encoded in non-commutative delta functions, for their triangle faces, and by the identification (gluing) conditions for the edge vectors associated with the same edge in two different tetrahedra.

\smallskip
\noindent Notice that, in defining the amplitude, we have not assigned different weights to different allowed configurations of edge vectors, and simply weighted equally all the allowed ones. This is the standard procedure in state sums constructions for quantum geometry (and true also for the most studied models of $4d$ quantum gravity, e.g. the Barrett-Crane and EPRL models, as well as for topological BF models). Any non-trivial action indicating the underlying discrete quantum dynamics appears only after re-expressing the amplitudes in terms of both relevant dynamical $1st$-order variables, i.e. bivector and discrete connection. This remains a choice, though, and not a logical necessity. The mathematical and physical consequences of different weight choices should be investigated.\\  
Note also that in the expression above we have decorations on all ten edges of the 4-simplex. Indeed, we constructed the 4-simplex amplitude \eqref{4Simplex_Amplitude} starting by the extended Hilbert spaces of the five tetrahedra, each of them being given as a constrained version of the functions on its six edge vectors. As we pointed out in section \ref{Sec_QuantumTetrahedron}, it is enough to impose \textit{three} triangle-closure conditions for each tetrahedron. Hence, in this formulation, the full geometry of the 4-simplex is recovered by providing the closure of only six of the ten triangles in it. We refer again to Fig.(\ref{Fig:4-simplex}) for the edges and triangles combinatorics of a 4-simplex. According to this figure, the expression of the amplitude appearing in \eqref{4Simplex_Amplitude} is achieved once we require the closure of the set of triangles $a=1,2,3$ of the tetrahedron $\alpha=1$, of the triangles $a=1,2$ of the tetrahedron $\alpha=2$ and the last triangle $a=1$ of tetrahedron $\alpha=3$. This is just an example of how one can fix the combinatorics. The price of a more symmetric definition, involving the closure conditions of all five tetrahedra and all their four triangles, which can of course be given in the same language, would be the presence of divergences due to the use of redundant delta functions.\\
\begin{figure}
    \centering
    \input{4Simplex}
    \caption{4-simplex boundary construction: five tetrahedra share five vertices. Each of the four faces of each tetrahedron is identified with one of the faces of the other four tetrahedra. We use the same color and a double dotted line for the identified faces.}
    \label{Fig:4-simplex}
\end{figure}
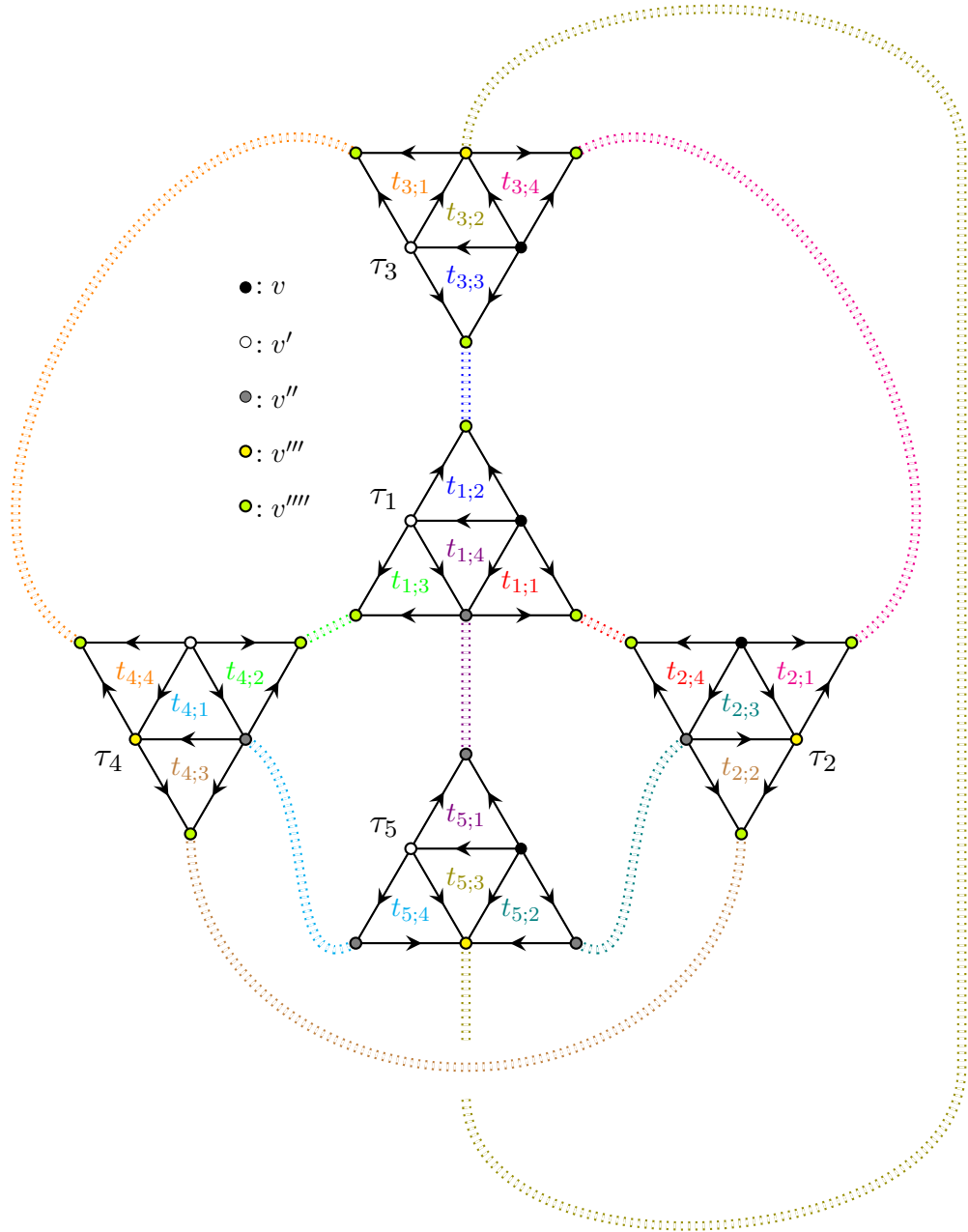
\smallskip \\
The 4-simplex amplitudes are then combined together, to give the expression of the amplitude for the whole simplicial complex, by identifying all the edge decorations across shared tetrahedra. To realize the gluing we use the amplitude $\cA_{\tau}$ given by;
\be
    \cA_{\tau} = \, \prod_{i=1}^6 \, \delta_\star(\lambda_{\alpha;i}-\lambda_{\beta;i})\,,    \label{TetrahedronAmplitude}
\ee
\smallskip \\
i.e. again a product of non-commutative delta functions imposing matching of geometric data between the neighboring tetrahedra $\alpha$ and $\beta$.\\
We had seen that the information about edge vectors, understood as elements of the translation group, can also be used to define quantum states of simple bivectors and then to define quantum states for tetrahedra as (non-commutative) functions of the same bivectors as \eqref{Tetrahedron_WaveFunction_BiVectors}. A (different, a priori) 4-simplex amplitude can therefore also be defined in terms of quantum bivectors $x$ associated with the triangles in the simplicial complex:
\be
    \cA_{\s}' 
    = 
    \Big(\prod_{\alpha=1}^{4} \, \hat{\cC}_{\tau_{\alpha}}(\{x_{\alpha;a}\})\Big) \star
    \prod_{\alpha=1}^4 \prod_{a=1}^4
    \delta_\star(x_{\alpha;a} - x_{\alpha+a;5-a}) \,,
    \label{4Simplex_Amplitude_BiVectors}
\ee
where this time $x_{\alpha;a}$ indicates the $a^{th}$ triangle of the tetrahedron $\alpha$, with $\alpha=1,\dots,5$ and $a=1,\dots,4$.
As for the amplitude \eqref{4Simplex_Amplitude}, the restriction to geometric configurations is obtained by the closure of four out of the five tetrahedra $\tau_b$. The four closures are encoded in the functions $\hat{\cC}_{\tau_b}$.
Similarly, a tetrahedron amplitude \eqref{TetrahedronAmplitude} simply provides the identification of the triangle bivectors associated with that tetrahedron in the two 4-simplices sharing it:
\be
    \cA_{\tau}' = \, \prod_{a=1}^4 \, \delta_\star(x_{\alpha;a} - x_{\beta;a}) \,.
\ee
The full quantum amplitude for the whole simplicial complex is then obtained by integration (with Lebesgue measure) of the complete set of bivector data entering the 4-simplex and tetrahedral amplitudes. We should recall, however, that the bivectors appearing in the above formula are in fact functions of edge vectors (they are simple bivectors, by construction, to encode the geometry of the simplicial complex). Therefore, the integration defining the full quantum amplitude should be performed over the set of edge vectors. This bivector construction gives, a priori, a different quantum amplitude than \eqref{NewSFAMplitude}. It is, on the other hand, much closer to the spin foam construction using the non-commutative flux representation put forward in \cite{BaratinOriti:2011GFTBarrettCrane}, giving a version of the Barrett-Crane model, and in \cite{BaratinOritimodel} for Holst-Plebanski gravity (including the Immirzi parameter). Indeed, the amplitudes expressed as functions of bivectors \eqref{4Simplex_Amplitude_BiVectors} can be obtained as part of the full quantum amplitudes in terms of edge vectors, as we are going to see in the next section, where the connection to the Barrett-Crane amplitudes is shown.

\subsection{Expansion in Lorentz group representations} \label{sec:Plancherel}
In the following, we show how the new spin foam amplitudes can be expressed in terms of irreps of the Lorentz group, i.e. in the more familiar spin foam language. This shows how they correspond to a nontrivial combination of Barrett-Crane amplitudes, constructed in such a way as to ensure the proper correspondence with edge vector geometries. We show this for the case of timelike bivectors, thus purely spacelike tetrahedra, using the tools in section \ref{Sec_InfiniteDimReps}.
\paragraph{The quantum minkowskian timelike bivector}
Let us recall that the quantization of bivectors in Minkowski space $\M^4$ is based on considering the elements of the dual Lorentz Lie algebra $\so^*(1,3)$. 
Our construction based on quantum edge vectors can be explicitly related to the usual Lorentz representations also introduced in section \ref{Sec_InfiniteDimReps}. For general Lorentz group elements:
\be
    g = \, \exp{\left(\sum_a \, (\alpha_a N_a + i\beta_a L_a)\right)} \,.
    \label{UnitaryRepLorentz}
\ee
one can define the Hilbert space of the square-integrable functions on the Lorentz group $L^2[\SO(1,3)] \cong F(\SO(1,3)) \ni f(g)$, and (similarly to the bivector (Fourier decomposed) wave functions we derived in  \eqref{BiVecWaveFunc_FourierDecomposition}),
we have the generalized Fourier decomposition:

\be
    f(g) :=  
    \sum_{n_r=0}^{\infty} \int \dd \mu \mu^2 \, 
    \sum_{\ell,\ell'} \sum_{m,m'} \,
    \la n_r, \mu ; \ell, m \,| g|\, n_r, \mu ; \ell', m ' \ra \, 
    f^{n_r, \mu ; \ell, m}_{n_r, \mu ; \ell', m'} \,.
    \label{LorentzWaveFunct_FourierDecomposition}
\ee
This is the standard Plancherel decomposition, with the generalized plane waves associated with unitary irreducible representations of the Lorentz group given by the standard $D$-matrices:
\be
    D^{0,\mu}_{\ell,m;\ell',m'}(g) :=
    \la \mu ; \ell, m \,| g |\, \mu ; \ell', m ' \ra \,.
    \label{time_likeD_matrix}
\ee
In our quantum geometric context, we are only interested in the case $j=0$, as we consider only the balanced \textit{timelike} representations. 
The balanced condition is automatically obtained when expanding functions that are invariant under (right-acting) rotations:
\be
    \int \dd u \, f(gu) = f(g) \,,
\ee
for all the rotations $u \in \SO(3)$. This condition restricts the Hilbert space of the Lorentz group to the space of functions on the upper hyperboloid $Q_1$, introduced with the Plancherel decomposition in section \ref{Sec_InfiniteDimReps}, and for which we had discussed the geometric interpretation. In this case, the Fourier decomposition \eqref{LorentzWaveFunct_FourierDecomposition} reduces to the expansion in terms of \lq timelike\rq   balanced infinite dimensional irreducible representations of the group\footnote{To avoid confusion, we recall that the label \lq timelike\rq  for the representations refers to the corresponding bivectors, which are then associated to -spacelike- triangles.}, setting the angular momentum quantum numbers $(l,m)$ equal to $0$:
\be
    f(g) :=  
    \sum_{n_r=0}^{\infty} \int \dd \mu \mu^2 \,
    \la n_r, \mu ; 0,0 \,| g |\, n_r, \mu ; 0,0 \ra \, 
    f^{n_r, \mu;0,0}_{n_r, \mu;0,0} \,,
    \label{LorentzWaveFunct_FourierDecomposition_TimeLike}
\ee
where we identify the non-radial part of the kernel as\footnote{The extension to the full harmonic oscillator basis $|\, n_r, \mu ; \ell', m ' \ra$ is trivial, as the plane wave given by $\la n_r, \mu ; \ell, m \,| U(g) |\, n_r, \mu ; \ell', m ' \ra$ is diagonal in the radial contribution.}:
\be
    D^{0,\mu}_{0,0;0,0}(g) :=
    \la n_r,\mu ; 0,0 \,| g |\, n_r,\mu ; 0,0 \ra \,.
    \label{TimeLike_Kernel}
\ee
Furthermore let us recall that the explicit expression for these plane waves in the hyperbolic basis given by \eqref{HyperbolicBasis} and once we restrict to $n_r=0$
we obtain the explicit expression for the $D$-matrix in the hyperbolic basis
\be
    D^{0,\mu}_{0,0;0,0}(g) =
    \la 0, \mu ; 0,0 \,| g |\, 0, \mu ; 0,0 \ra =
    \int \dd \eta \, Q_{0}^{1-\mu^2}(\coth\eta) \, Q_{0}^{1-\mu^2}(\coth\eta') \,,
    \label{CanonicalD_Expresion}
\ee
where $\eta'$ is obtained by the action of the transformation $g$ on the Legendre function $Q$ with the coordinates $\eta$.  This is the canonical $D$ matrix, which we will now express in the harmonic oscillator basis.
This can be achieved since we are able to expand the function on the Lorentz group in the Fourier modes of translations using the coherent states of the harmonic oscillator. This yields the harmonic decomposition
\be
    f(g) :=  
    \int \dd \alpha \dd \alpha' \, 
    \la \alpha_{\nu} \,| g |\, \alpha_{\nu}' \ra \, 
    f_{\alpha_{\nu},\, \alpha_{\nu}'} \,.
    \label{LorentzWaveFunct_FourierDecomposition_Translation}
\ee
The plane wave that appears in the decomposition
are related to the standard Lorentz representations obtained through the Plancherel decomposition \eqref{LorentzWaveFunct_FourierDecomposition}, by the change of basis \eqref{HarmonicOscillator_CoherentMinkowskian} and \eqref{HarmonicOscillator_HyperbolicCoherent}:
\begin{align}
    \la \alpha_{\nu} \,| g |\, \alpha_{\nu}' \ra &=
    e^{-\frac{1}{2}(|\alpha|^2 + |\alpha'|^2)} 
    \sum_{\{n_{\nu},n_{\nu}'\}} \sum_{n_r=0}^{\infty} \sum_{\ell,\ell'} \sum_{m,m'} \,
    \frac{\alpha_t^{n_t} \alpha_x^{n_x} \alpha_y^{n_y} \alpha_z^{n_z} {(\alpha_t')^*}^{n_t'} {(\alpha_x')^*}^{n_x'} {(\alpha_y')^*}^{n_y'} {(\alpha_z')^*}^{n_z'}}{\pi^8 \sqrt{n_t! n_x! n_y! n_z! n_t'! n_x'! n_y'! n_z'!}}\nonumber\\
   & \times \cC^{n_t,\, n_x,\, n_y,\, n_z}_{0,\, \mu,\, \ell,\, m} \cC^{n_t',\, n_x',\, n_y',\, n_z'}_{0,\, \mu',\, \ell',\, m'}\, D^{0,\mu}_{\ell,m;\ell',m'}(g)    \,.
\end{align}
Once again we restrict the Fourier expansion to the timelike sector on $Q_1$. In terms of edge vector coordinates, the upper hyperboloid is obtained by the intersection of a pair of co-linear light cones centered at the timelike points parametrized by $P_i = (0,x_i,y_i,z_i) \in \M^4$ of Minkowski space:
\be
    \left\{\,\,
    \begin{aligned}
        & t^2 = A \big((x-x_{1})^2 + (y-y_{1})^2 + (z-z_{1})^2 \big) \,,\\
        & t^2 = A \big((x-x_{2})^2 + (y-y_{2})^2 + (z-z_{2})^2 \big) \,,\\
    \end{aligned}
    \right.
    \quad\Rightarrow\quad
    t^2 - x^2 - y^2 - z^2 = R^2 \,,
\ee
with $R^2 = 2A \big((x_1-x_2)^2 + (y_1-y_2)^2 (z_1-z_2)^2\big)$ being the hyperbolic distance between two points.
We can therefore restrict the space of the translation group to the null edge vectors with unitary time coordinate 
\be
    e = (1,\lambda_x,\lambda_y,\lambda_z)
    \,\,,\quad 
    |e|^2 := 1- \lambda_x^2 + \lambda_y^2 + \lambda_z^2 = 0 \,.
    \label{TimeLike_Condition}
\ee
Note that, this can be regarded as the gauge condition\footnote{Furthermore, we note that the condition \eqref{TimeLike_Condition} can be written in terms of the Euler coordinates of a bivector \eqref{BiVector_Lorentz}, whose coordinates are explicitly given in \eqref{BiVectorCoordinates_Lorentz-Translation}.} on the edge vectors that allows recovering the timelike simple bivectors spanning the tangent space of the hyperboloid $Q_1$ (for more details on hyperbolic geometry we refer to \cite{Barrett2015}). Moreover, it is easy to check that this condition correctly implies that the bivector in such parametrization is timelike: $|b|^2 > 0$.\\
With this condition, the constrained Lorentz wave function in   \eqref{LorentzWaveFunct_FourierDecomposition_TimeLike} is automatically decomposed as a combination of edge vectors (or equivalently, generators of translations).\\
Now we arrive at the last map we need to implement in our construction. Note that \eqref{UnitaryRepLorentz} provides the relation between the bivector wave function \eqref{BiVector_WaveFunction} and the one on the Lorentz group. 
This type of relation between functions on the Lorentz group and functions on the dual Lorentz algebra is encoded in the non-commutative Fourier transform. This is basically an intertwining map
between the group and algebra representation (and vice-versa), ensuring their unitary equivalence. It is given by the integral transform:
\be \label{NCFourier}
    \hat{f}(x) = \int \dg \, e_{\star}(g \,,\, x) \, f(g) \,,
\ee
where $e_{\star}(g \,,\, x)$ is the star (non-commutative) exponential for the Lorentz group. We referred to it throughout this paper as the non-commutative plane wave. Its characterizing equations can be derived by requiring that the intertwined function spaces define a representation of the same underlying quantum algebra and applying the action of unitary operators on the various representations. Moreover, it is important to note that deriving an explicit expression of this non-commutative plane wave depends on the choice of a quantization map as was already discussed in section \ref{Sec_QuantumTriangle}. For more details see \cite{GuedesOritiRaasakka:2013NonCommFourier,FinocchiaroOriti:2018SFDufloMap,OritiRosati:2018GFTNonCommFourier}. \\
Since we are restricting the Lorentz wave functions to those that are invariant under rotations given in  \eqref{LorentzWaveFunct_FourierDecomposition_TimeLike}, the non-commutative Fourier transform allows to express the bivector functions in  \eqref{BiVector_WaveFunction} in terms of $D$-matrices  \eqref{TimeLike_Kernel} and it yields:
\be
    \hat{f}(x) = \int \dg \dd \mu \, \mu^2 \, e_{\star}(g \,,\, x) \,
    D^{0,\mu}_{0,0;0,0}(g) \, f^{0,\mu;0,0}_{0,\mu;0,0} \,.
\ee
More importantly, a similar expansion can be obtained for the bivector wave function in terms of the edge vector coordinates. Here, the group element $g \in \SO(1,3)$ is parametrized as in \eqref{UnitaryRepLorentz}, whereas the exponent is the bivector operator given by \eqref{QuantumBiVector} i.e. as a linear combination of the ladder operators, expressed in terms of the coordinates $(\lambda_1,\lambda_2) \in \M^4$:
\be
\label{BCexpansion_star_product}
    \hat{f}(\lambda_1,\lambda_2) = 
    \int \dg \dd \mu \, \mu^2 \, e_{\star}(g\,,\, x_{12}\left(\lambda_1,\lambda_2\right)) \,
    D^{0,\mu}_{0,0;0,0}(g) \, f^{0,\mu;0,0}_{0,\mu;0,0} \,,
\ee
with $\ell,\ell'=0$ for timelike bivectors $|b|^2$, that are obtained by the skew-symmetric product of the translation generators (the ladder operators) and must satisfy the condition \eqref{TimeLike_Condition} whereas the notation $x_{12}\left(\lambda_1,\lambda_2\right)$ denotes the bivector associated with the pair of edge vectors $\left(\lambda_1,\lambda_2\right)$. The non-commutative edge vector-based construction of a quantum triangle and quantum tetrahedron proceeds then following the steps in section \ref{Sec_QuantumTriangle}.-\ref{Sec_QuantumTetrahedron}.
\paragraph{Relation to the Barrett-Crane amplitudes.}
The expansion \eqref{BCexpansion_star_product} is the root of the relation between new spin foam amplitudes \eqref{NewSFAMplitude} and those of the Barrett-Crane model \cite{BarrettCrane:1999BarrettCrane1,Perez:2000ec} for the case of timelike bivector-based geometries.\\
Starting from the expression of a single simplex amplitude $\mathcal{A}_s$ expressed in terms of functions of edge vectors as depicted in \eqref{4Simplex_Amplitude}, we next reformulate it in terms of functions of bivectors by applying the relation derived in \eqref{BCexpansion_star_product}. Moreover, note that we can write down the integral expansion of the non-commutative delta function for the gluing operation of tetrahedra in terms of the non-commutative plane waves, namely
\be
\label{NCdelta_decomposition}
  \prod_{\alpha=1}^4\prod_{a=1}^4  \delta_{\star}\left(x_{\alpha ; a}-x_{\alpha+a ; 5-a}\right)=\prod_{\alpha=1}^4\prod_{a=1}^4\int \mathrm{d} h_{\alpha, a} \, e_{\star}\left(h_{\alpha, a}, x_{\alpha ; a}-x_{\alpha+a ; 5-a}\right)\,,
\ee\
where $h_{\alpha,a}$ is the group element associated with the bivector of the common triangle $a$ that glues a pair of tetrahedra along the shared face. The same can be conducted for the closure constraint at the level of each triangle of the tetrahedra where we obtain:

\be
    \delta_{\star}\left(\sum_{a=1}^4 x_{\alpha ; a}\right)=
    \int \mathrm{~d} h_a \prod_{a=1}^4 \prod_{i=1}^5 e_{\star}\left(h_a, x_{\alpha ; a}(\lambda_{\alpha;i},\lambda_{\alpha;i+1})\right) .
\ee
 where $h_a$ is the gluing element. Now, putting everything together in a single simplex amplitude and combining the non-commutative plane waves and the $D$-functions we obtain 
\be
  \begin{aligned}
  \cA_{\s} = \int [\dd g] \, [\dh]^5 \dd^5 \mu \,   &
       \prod_{\alpha=1}^5 \prod_{a=1}^4 \prod_{i=1}^5 \, \mu_{\alpha}^2 \,
       e_{\star}(h_{\alpha} \,,\, x_{\alpha;a}(\lambda_{\alpha;i},\lambda_{\alpha;i+1})) \,
        D^{0,\mu_{\alpha}}_{0,0;0,0}\big(g(x_{\alpha;i}) g\mone(x_{\alpha+i;5-i})\big) \,,
        \label{new_SF_amplitude}
    \end{aligned}
\ee

and it is given in terms of the bivector coordinates $x$ (in the Euler angle parametrization), expressed as a combination of the coordinates of the edge vectors $\lambda$. Recall that $\lambda_{\alpha;i}$ stands for the coordinate of the $i^{th}$ edge of the tetrahedron $\alpha$, for $\alpha=1,\dots,5$ and $i=1,\dots,6$.  \\
In this picture, the timelike canonical $D$ matrices given in  \eqref{TimeLike_Kernel}, associated with  the timelike irreducible representations of the Lorentz group, can be computed explicitly as in  \eqref{CanonicalD_Expresion} and expressed 
\be
    D^{0 \, \mu}_{0,0;0,0}(g_1g_2\mone) = K_{\mu}(d_{\eta}(\mathbf{x}_{g_1},\,\mathbf{x}_{g_2})) =
    \frac{\sin(\mu d_{\eta})}{\mu \sinh d_{\eta}} \,,
    \label{DMatrices_BC}
\ee
where $d_{\eta}$ is the hyperbolic distance between two points $\mathbf{x}_{g_i}$ on the hyperboloid $Q_1$ associated with  the Lorentz transformations $g_i$. It is obvious now that, the amplitudes associated with a tetrahedron and to a 4-simplex can be expressed as a combination of the terms \eqref{DMatrices_BC}. This construction is realized by representing the terms \eqref{DMatrices_BC} as the graph dual to tetrahedra. In particular, we associate a set of coordinates $\mathbf{x} \in Q_1$ on the hyperboloid $Q_1$ to a given vertex of the graph, one parameter $\mu$ is attached to each edge of the graph, and the term \eqref{DMatrices_BC} is associated with each internal edge of the graph (we refer to Fig.(\ref{Fig:BarrettCraneVertex})).
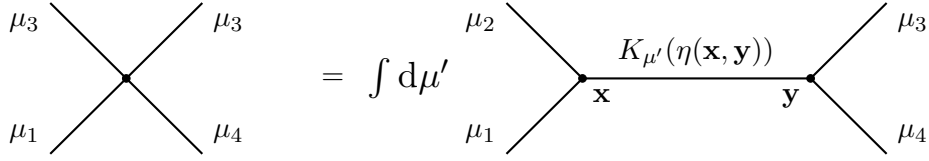
\begin{figure}
    \centering
    \input{BarretCraneVertexTerm}
    \caption{The tetrahedron amplitude of the Barrett-Crane model \cite{BarrettCrane:1999BarrettCrane1}}
    \label{Fig:BarrettCraneVertex}
\end{figure}
By merging together several of these graphs dual to tetrahedra, one can construct the \textit{relativistic spin networks}.
The relativistic spin networks were first introduced by Barrett and Crane in \cite{Barrett1999}, further studied in \cite{Reisenberger:1998SpinNetworkRelativistic,Barrett1998}, and are at the root of the Barrett-Crane spin foam model \cite{Perez:2003SpinFoam,Perez2000}.
\smallskip \\
The amplitude obtained in \eqref{new_SF_amplitude} presents itself as a nontrivial combination of such terms, i.e. BC amplitude-like terms, for the case of timelike bivectors (spacelike tetrahedra). 
We could also expect the full amplitude (\ref{NewSFAMplitude}) to reduce to a Barrett-Crane-like model with standard BC vertex and possibly complicated remaining measures, once all edge vectors and group elements (discrete connection) are integrated out. \\
Still, having formulated explicitly the model using the edge vectors (exploiting their relation to the translation group), and thus controlling in detail the full simplicial quantum geometries, we get a rather more sophisticated quantum amplitude for them. This further information, compared to the usual formulation of spin foam models (including the Barrett-Crane model), is encoded in the non-commutative plane waves appearing in \eqref{new_SF_amplitude}. \\
These further data ensure that no geometric information or condition is missing in the quantum amplitudes, from the discrete gravity point of view. This may not be necessarily apparent at the level of the quantum amplitudes \ref{NewSFAMplitude} for a closed triangulation, but it will become so for \lq transition amplitudes\rq  between non-trivial boundary data and when  evaluating quantum geometric observables. \\ 
The construction also confirms \cite{BaratinOriti:2011GFTBarrettCrane,JercherOritiPithis:2021GFTEmergentBarrettCrane} the viability of the Barrett-Crane quantization prescription, provided it is embedded in \lq a more sensitive state sum construction\rq \cite{CraneYetter}. \\
Finally, notice that the edge-vector construction, here provided explicitly for spacelike tetrahedra, can be straightforwardly extended to timelike ones, as for the Barrett-Crane model itself.

\section{Group field theory formulation}
\label{SubSec_GFTEdgeVectors}

Like any other spin foam amplitude for quantum geometry \cite{Oriti:2014GFTLQG}, the new spin foam amplitude \eqref{NewSFAMplitude} can be obtained as the Feynman amplitude of a group field theory \cite{GFT_thomas,Oriti:2014GFTLQG}. \\
The group field theory formulation provides a complete definition of the quantum dynamics of simplicial quantum geometry, i.e. a complete definition of the spin foam model. Indeed, it defines also a continuum limit of the model, via a sum over simplicial complexes (spin foams), thus removing the dependence on the initial complex. In the following, we give briefly some more details on this GFT formulation of the new model.\\

In this case, the fundamental GFT has the same general definition as the tetrahedron wave function given by \eqref{TetrahedronWaveFunction_Edges} expressed as a function on three copies of the translation group $F(\M^4)^{\times 3}$ with the proper closure conditions for the three triangles. More symmetrically, it can be given also as a function on six copies of the translation group constrained by the closure condition in \eqref{ClosureConstraint}.
As we encountered in sections \ref{Sec_QuantumTriangle}-\ref{Sec_QuantumTetrahedron}, each edge vector $e_i$ of the tetrahedron $\tau$ is described by a set of (four) coordinates $\lambda_i \in F(\bbR^4)$. Therefore, the fundamental GFT field, like the tetrahedral wavefunction, is expressed as a function:
\be
     \hat{\Phi}(\lambda_1,\lambda_2,\lambda_3,\lambda_4,\lambda_5,\lambda_6) \,\,\in 
     L^2[\lambda_1,\lambda_2,\lambda_3,\lambda_4,\lambda_5,\lambda_6] \cong F(\bbR^4)^{\times 6} \,,
    \label{Tetrahedron_WaveFunction}
 \ee
 subject to the closure conditions \eqref{ClosureConstraint}
 \be
     \hat{\Phi}_t = (\hat{\cC}_t \star \hat{\Phi}) \,,    
 \ee
 where the functions satisfy a nontrivial star product. 
As discussed above, a special case of tetrahedral states and wavefunctions (thus GFT fields) is given by functions of edge vectors depending only on the triangle bivectors constructed from them, that is by functions 
 \be
     (\hat{\cC}_t \, \hat{\Phi})(\lambda_1,\, \dots,\, \lambda_6) \equiv
    (\hat{\cC}_{\tau} \, \hat{\Phi})(x_1,\, x_2,\, x_3,\, x_4) = 
     \delta(x_1 + x_2 + x_3 + x_4) \star \hat{\Phi}(x_1,\, x_2,\, x_3,\, x_4) =\hat{\Phi}_b\,,
    \label{Tetrahedron_WaveFunction_BiVectors2}
 \ee
 where the field $\hat{\Phi}_b$ is a subclass of the generally defined GFT field characterized by the bivector construction. 
In turn, this class of GFT fields (and tetrahedral wavefunctions) can be mapped to functions of four copies of the Lorentz group (subject to diagonal gauge invariance) by non-commutative Fourier transform (see \cite{OritiRosati:2018GFTNonCommFourier}), and to functions of irreps of the Lorentz group by means of the Plancherel decomposition. 
 \be
    \Phi \in F(\SO(1,3)^{\times 4}) 
     \quad\overset{\star Fourier}{\longrightarrow}\quad
     \hat{\Phi}_t \in F(\bbR^6_{\star}) \cong F(\so^*(1,3))
    \quad\overset{eq. \eqref{Tetrahedron_WaveFunction_BiVectors}}{\longrightarrow}\quad
\quad\overset{eq. \eqref{Tetrahedron_WaveFunction_BiVectors}}{\underset{\cC: \, \eqref{Lorentz-Translations_Coefficients}}{\longrightarrow}}\quad
    \hat{\Phi}_b \in F(\bbR^4 \wedge \bbR^4) \,.
 \label{Chain_Field}
 \ee
 \medskip \\
Both these maps define a unitary equivalence and bring our edge vector formulation in the mathematical language more routinely used in the spin foam and group field theory context. Of course, the other expansion \eqref{BiVecWaveFunc_FourierDecomposition} is also available and defines one more unitary equivalence map 

\be
     \hat{\Phi}(\lambda_1,\, \dots,\, \lambda_6) =
     \int \dd \alpha^6 \dd {\alpha'}^6 \,
     T_{\alpha_{1;\nu}}^{\alpha_{1;\nu}'}(\lambda_2,\lambda_3) \,
     T_{\alpha_{2;\nu}}^{\alpha_{2;\nu}'}(\lambda_3,\lambda_5) \,
    T_{\alpha_{3;\nu}}^{\alpha_{3;\nu}'}(\lambda_6,\lambda_5) \,
     T_{\alpha_{4;\nu}}^{\alpha_{4;\nu}'}(\lambda_2,\lambda_6) \, \hat{\Phi}^{\alpha_{1;\nu},\,\alpha_{2;\nu},\,\alpha_{3;\nu},\,\alpha_{4;\nu}}_{\alpha_{1;\nu}'\,\alpha_{2;\nu}'\,\alpha_{3;\nu}'\,\alpha_{4;\nu}'} \,.
     \label{TetraWaveFunct_FourierDecompositionGFT}
 \ee
 where we used the compact notation $T_{\alpha_{\nu}}^{\alpha_{\nu}'}(\lambda_1,\lambda_2)$ for the plane waves $\la \alpha_{\nu} \,| \adag(\lambda_1) a(\lambda_2) - \adag(\lambda_2) a(\lambda_1) |\, \alpha_{\nu}' \ra$.
The GFT action for such fields, producing the desired amplitudes \eqref{NewSFAMplitude}, is given by an interaction term that encodes the combinatorics of a 4-simplex and its quantum geometry, and a propagator enforcing the required gluing conditions. The kernels for such interaction and the kinetic term producing the required propagator are exactly the 4-simplex amplitude  \eqref{4Simplex_Amplitude}, and the tetrahedral amplitude  \eqref{TetrahedronAmplitude}, the latter being in fact equal to the propagator itself. \\
The action of the GFT model (with the coupling constants set to 1 for simplicity), is\footnote{Of course, all the generalizations available for other GFT and spin foam models, in terms of combinatorics of the relevant cellular complexes (thus, interaction terms), field coloring (to control the resulting cellular topologies), and dynamical kernels are available for this model as well.}:
 \be
     \cS = 
     \int \dd \lambda^6 \, (\hat{\Phi} \star \hat{\Phi})(\{\lambda_i\}) +
 \int \dd \lambda^{12} \, (\hat{\cK} \star (\hat{\Phi} \cdot \hat{\Phi}))(\{\lambda_i;\lambda_j\})\,.
 \ee
The \textit{kinetic term} (identical to the \textit{propagator}) is taken thus to be
 \be
     \hat{\cK} = \prod_{i=1}^6 \, \delta(\lambda_i-\lambda_{i+6}) \,,
 \label{Propagator_Amplitude}
 \ee
 and enforces the identification of the edge vector data,
 while 
 \be
     \begin{aligned}
         \cS_{\cV} =
         \int \dd \lambda^{10} \,
        &
         (\hat{\cC} \, \hat{\Phi})(\lambda_{1},\lambda_{2},\lambda_{3},\lambda_{4},\lambda_{5},\lambda_{6}) \star
        (\hat{\cC} \, \hat{\Phi})(\lambda_{7},\lambda_{3},\lambda_{8},\lambda_{1},\lambda_{9},\lambda_{2}) \star
         (\hat{\cC} \, \hat{\Phi})(\lambda_{5},\lambda_{8},\lambda_{10},\lambda_{7},\lambda_{4},\lambda_{3}) \star
     \nonumber \\
         &
         (\hat{\cC} \, \hat{\Phi})(\lambda_{9},\lambda_{10},\lambda_{6},\lambda_{5},\lambda_{1},\lambda_{8}) \star
      (\hat{\cC} \, \hat{\Phi})(\lambda_{4},\lambda_{6},\lambda_{2},\lambda_{9},\lambda_{7},\lambda_{10}) 
        \nonumber \\
         =
      \int \dd \lambda^{30} \, 
       & \Big(\hat{\cV}(\{\lambda_{\alpha;i}\}) \star
        \big((\hat{\cC} \, \hat{\Phi})(\lambda_{1;1},\lambda_{1;2},\lambda_{1;3},\lambda_{1;4},\lambda_{1;5},\lambda_{1;6}) \,
        (\hat{\cC} \, \hat{\Phi})(\lambda_{2;1},\lambda_{2;2},\lambda_{2;3},\lambda_{2;4},\lambda_{2;5},\lambda_{2;6}) \,
     \nonumber \\
         & \,\,\,
         (\hat{\cC} \, \hat{\Phi})(\lambda_{3;1},\lambda_{3;2},\lambda_{3;3},\lambda_{3;4},\lambda_{3;5},\lambda_{3;6}) \,
         (\hat{\cC} \, \hat{\Phi})(\lambda_{4;1},\lambda_{4;2},\lambda_{4;3},\lambda_{4;4},\lambda_{4;5},\lambda_{4;6}) \,
     \nonumber \\
        & \,\,\,
        (\hat{\cC} \, \hat{\Phi})(\lambda_{5;1},\lambda_{5;2},\lambda_{5;3},\lambda_{5;4},\lambda_{5;5},\lambda_{5;6})\big)\Big) \,,
    \end{aligned}
 \ee
 is the \textit{interaction term} whose amplitude is given by
 \begin{align}
     \hat{\cV} 
     = \, &
     \big(\delta(\lambda_{1;1}\lambda_{2;4}\mone) \, \delta(\lambda_{1;1}\lambda_{4;5}\mone)\big) \,
 \big(\delta(\lambda_{1;2}\lambda_{2;6}\mone) \, \delta(\lambda_{1;2}\lambda_{5;3}\mone)\big) \,
     \big(\delta(\lambda_{1;3}\lambda_{2;2}\mone) \, \delta(\lambda_{1;3}\lambda_{3;6}\mone)\big) \,
     \big(\delta(\lambda_{1;4}\lambda_{3;5}\mone) \, \delta(\lambda_{1;4}\lambda_{5;1}\mone)\big) \,
     \nonumber \\
     &
     \big(\delta(\lambda_{1;5}\lambda_{3;1}\mone) \, \delta(\lambda_{1;5}\lambda_{4;4}\mone)\big) \,
     \big(\delta(\lambda_{1;6}\lambda_{4;3}\mone) \, \delta(\lambda_{1;6}\lambda_{5;2}\mone)\big) \,
    \big(\delta(\lambda_{2;1}\lambda_{3;4}\mone) \, \delta(\lambda_{2;1}\lambda_{5;5}\mone)\big) \,
     \big(\delta(\lambda_{2;3}\lambda_{3;2}\mone) \, \delta(\lambda_{2;3}\lambda_{4;6}\mone)\big) \,
    \nonumber \\
     &
     \big(\delta(\lambda_{2;5}\lambda_{4;1}\mone) \, \delta(\lambda_{2;5}\lambda_{5;4}\mone)\big) \,
     \big(\delta(\lambda_{3;3}\lambda_{4;2}\mone) \, \delta(\lambda_{3;3}\lambda_{5;6}\mone)\big) \,,
     \label{4Simplex_Amplitude_GFT}
 \end{align}
 that is associated with a combination of five tetrahedra (expressed in terms of their edge vectors) forming a 4-simplex\footnote{Note that here, for reasons of simplicity, we denoted the inverse of a group element with the usual exponent $^{-1}$.}.
 
 As we explained, the GFT fields, and therefore the action, can be equivalently expressed in Lie algebra, group and representation variables for the Lorentz group, in full for restricted GFT fields or partially for generic ones (since only their dependence on skew-symmetric combinations of edge vectors define bivectors). The same decomposition can be performed at the level of the Feynman amplitudes.
  The Feynman diagrams of the model are by constructions obtained as gluings of 4-simplices to form (up to further topological subtleties \cite{GurauRyan:2011TensorReview}, four-dimensional simplicial complexes. As is generically the case, the Feynman amplitudes of our GFT model are then constructed by convoluting (with respect to shared edge vector data) interaction kernels, and propagators. It is immediate to see that they are given by the amplitudes \eqref{NewSFAMplitude}. 
 
\section{Discussion and outlook}

In this paper, we have defined a new model for $4d$ Lorentzian quantum geometry, based on edge vectors ensuring therefore a complete encoding of (quantum) simplicial geometry. Several elements of the construction were suggested earlier in \cite{CraneYetter}. The model has been constructed first at the level of amplitudes assigned to a given simplicial complex and then completed by a sum over complexes (thus, by a (formal) definition of its continuum limit) via a Group Field Theory formulation.
The key technical tool for the construction was the encoding of edge vector data in representations of the translation group, and their subsequent connection to the representation theory of the Lorentz group. In turn, this connection was based on the theory of expansors and their reformulation as four-dimensional Lorentzian harmonic oscillators. The other technical tool was the non-commutative representation and Fourier transform for functions on Lie groups, developed for the Lorentz group in \cite{OritiRosati:2018GFTNonCommFourier}, based on the general techniques in \cite{GuedesOritiRaasakka:2013NonCommFourier}.\\
The amplitudes of the new model were then shown to relate directly to the Barrett-Crane quantization of simplicial geometry, upon expansion in terms of irreps of the Lorentz group. This suggests that they may reduce to a specific version of the Barrett-Crane model with the Barrett-Crane vertex amplitude completed by some possibly involved measure term encoding the relation between triangle bivectors (and areas) and edge vectors (and lengths). 
This relation indicates that the Barrett-Crane quantization of simplicial geometry has nothing wrong per se, as often claimed in the spin foam literature \cite{BaratinOriti:2011GFTBarrettCrane}, and the issue was rather to make sure that the geometric information not already contained in basis of eigenstates of area operators (balanced or "relativistic" spin networks) is also encoded correctly in the complete state sum/spin foam amplitude, as our construction allowed to do.\\

\noindent Our results provide a possible solution to what has been the main open issue in model building in the spin foam/GFT context (as far as models for Lorentzian Quantum Gravity are concerned) in the last decades. That is, to define amplitudes for simplicial complexes that correctly encode the complete set of geometric data (given by compatible edge vectors/lengths), and thus, when formulated in terms of bivectors, correctly encode the (simplicity) constraints reducing them to a compatible set of edge vectors.\\

\noindent Obviously, having a promising candidate solution does not imply that there are not many open issues to be tackled, or many avenues of further developments. We indicate some of them.\\

\noindent A first set concerns the quantum geometric features of the new model and its correspondence with discrete gravitational dynamics, as well as the analytic control over it.\\
We need to analyze further the properties of the non-commutative star product for functions of edge vectors and of bivectors and explore different choices and their mathematical consequences for the amplitudes. This is also needed to obtain more explicit and more easily computable expressions for its amplitudes, which have only been given here in rather implicit form (even though all the mathematical ingredients to compute them in principle have been provided). A comparison with other constructions from the non-commutative geometry literature (for functions on Minkowski space) \cite{Lizzi:2014pwa} and the import of tools developed there will certainly be beneficial, together with further work on the non-commutative Fourier transform on the Lorentz group proposed in \cite{OritiRosati:2018GFTNonCommFourier}. More generally, given how involved the new quantum amplitudes are from the analytic point of view, we need to develop new analytic techniques for their evaluation and complement them with numerical ones. In this respect, the task is the same one we face with other $4d$ spin foam models \cite{Dona:2022dxs}. Key tools may come from spinorial techniques, already applied successfully to $SU(2)$ spin networks \cite{Livine:2011gp} and whose associated star product has been studied also in a GFT context \cite{Dupuis:2011fx}, as well as from a Lorentzian generalization of earlier result on the formulation of spin networks in terms of harmonic oscillators \cite{Girelli:2005ii}. \\

While the encoding of simplicial geometry is manifest in our model, the connection to discrete gravity actions and thus discrete gravity dynamics is not, and it should be clarified. As for the other spin foam models \cite{FinocchiaroOriti:2018SFDufloMap}, the classical discrete gravity action appears when expressing the amplitudes in terms of both metric (edge vectors or bivectors) and connection (group) variables since the underlying gravity theory is a first-order one. For the new model the underlying action is expected to be a discretization of the Palatini action \cite{Caselle:1989hd} and it will be encoded in the explicit form of the non-commutative plane waves appearing when expanding in terms of simple bivectors (functions of edge vectors) and group elements, i.e. using \eqref{NCFourier}, before expanding further in Lorentz group irreps. The corresponding expression of the amplitudes should thus be analyzed in detail. Three aspects, in particular, should be investigated, in the resulting discrete gravity path integral. First, starting from this expression, a semiclassical approximation should be performed, as done in other spin foam models, to elucidate the reduction to simplicial geometry in terms of edge lengths, i.e. Regge geometries, in such approximation or, in other words, the dominance of solutions of the discrete Regge equations. Many results have been obtained for other quantum geometric models in this respect (see e.g \cite{Han_2013}), and will be useful for the new one as well. However, in other quantum geometric models typically one analyses the expression of the amplitudes in a spin network basis, i.e. in terms of eigenstates of triangle area operators, while the most straightforward strategy in our case will probably be along the lines of what is done in the $3d$ case in \cite{Oriti:2014aka}, where special attention has been placed on the non-commutative structure of the amplitudes, taking the semiclassical approximation directly in the expression showing them as non-commutative discrete gravity path integrals.\\
Second, the measure of the discrete gravity path integral defined by this expression (or after further Plancherel expansion) contains the constraints reducing the larger set of triangle areas (and bivectors) to the smaller set of edge lengths. The reduction of the first to the second, thus from so-called \lq discrete area metrics\rq to the usual piecewise-flat Regge geometries, has been another open issue in the last decades \cite{Dittrich:2023ava, Asante:2018wqy, Barrett:1997tx}. The analysis of the amplitudes of the new model should cast light on this issue as well.
Third, the role of degenerate simplicial geometries in the new state sum should also be investigated. Their suppression was indeed the main motivation for the construction of a model based on edge vectors in \cite{CraneYetter}.\\

\noindent At the mathematical level, the new amplitudes are constructed using a subtle interplay between the representation theory of the translation group and of the Lorentz group. We expect that, underlying the model, one can find a formulation in terms of the representation theory of the Poincar\'e  group, indeed often invoked as the relevant gauge symmetry of Palatini gravity \cite{Caselle:1989hd}. In fact, the mathematical analysis of the model may unravel exciting categorical structures, since our construction bears a number of points of contact with the categorical constructions in \cite{Girelli:2021zmt,Girelli:2022bdf,Asante:2019lki} based on the Poincar\'{e} 2-group, and the earlier ones in \cite{Baratin:2014era, Baratin:2009za}. A categorical formulation of this new state sum based on edge vectors was in fact suggested already in \cite{CraneYetter}. \\

\noindent Another set of further developments concerns the applications of a new model to the broader issues of quantum gravity.\\ 
For this new one, like for all other quantum geometric models, the main points to be established are the well-posedness of quantum theories, the continuum limit, and the extraction of effective gravitational physics from it (showing the connection to continuum GR). These tasks are most efficiently tackled in the Group Field Theory formulation, as it is being done for other $4d$ quantum geometric models as well as for simpler models. \\
The first two issues become those of perturbative renormalizability of the corresponding GFT model \cite{Carrozza:2016GFTRenGroup1} and of its non-perturbative RG flow and phase diagram \cite{Marchetti:2022nrf,Eichhorn:2021vid}. The third has many facets and requires to be tackled from several directions. It has been addressed partially, at least in the cosmological sector, in the context of GFT condensate cosmology \cite{Oriti:2016GFTCondensateEmergent,Oriti:2021GFTCondensateEmergent,Pithis:2019tvp}, via a mean field approximation (thus bypassing the spin foam formulation, corresponding to the perturbative expansion of the quantum dynamics) in which the universe is treated as a quantum condensate of GFT quanta, i.e. quantum tetrahedra. The new model of quantum geometry we have constructed should be tested also in this context.

\appendix

\section{Four dimensional Lorentzian harmonic oscillator}
\label{App_4dHarmonicOscillator}
We present here the solutions of the four dimensional Lorentzian harmonic oscillator in different basis formulations.
The isotropic four dimensional harmonic oscillator is described by the Hamiltonian
\be
    \cH = - \frac{1}{2} \Delta + \frac{1}{2} (t^2 - x^2 - y^2 - z^2) \,,
\ee
with associated Schr\"odinger equation $\cH \, \Psi = E \, \Psi$, where $\Delta$ is the four dimensional Lorentzian Laplacian operator $\Delta = \partial_t^2 - \partial_x^2 - \partial_y^2 - \partial_z^2$.
Following \cite{Genest2013,laplacian_hyperbolic_solution}, we solve such wave equation in the Minkowskian, cylindrical, spherical, and hyperbolic coordinates.
\begin{itemize}
    \item[(i)] 
    \textit{Minkowskian coordinates}.
    The Minkowskian basis is $|n_t,n_x,n_y,n_z\ra$, with $n_x,n_y,n_z \in \bbN$ and $n_t \in \bbR$. 
    It is associated with  the eigenvalue $E = (n_t + 1/2) - (n_x + n_y + n_z + 3/2)$ with eigenbasis
    \be
        \Psi_{n_t,n_x,n_y,n_z}(t,x,y,z) =
        \psi_{n_t}(t) \, 
        \psi_{n_x}(x) \, \psi_{n_y}(y) \, \psi_{n_z}(z) \,,
        \label{MinkowskiBasis}
    \ee
    where $\psi_n$ are the Hermite functions.
    
    \item[(ii)]
    \textit{Cylindrical coordinates}.
    The cylindrical basis is $|n_t,n_{\rho},m,n_z\ra$, with $n_t,n_{\rho},n_z \in \bbN$, $n_t \in \bbR$ and $m \in \bbZ$. 
    It is associated with  the eigenvalue $E = (n_t + 1/2) - (2n_{\rho} + |m| + n_z +3/2)$ with eigenbasis
    \be
        \Psi_{n_t,n_{\rho},m,n_z}(t,\rho,\phi,z) =
        \psi_{n_t}(t) \, 
        \frac{(-1)^{n_{\rho}}}{\sqrt{\pi}} \sqrt{\frac{n_{\rho}!}{\Gamma(n_{\rho} + |m| + 1)}} \, e^{-\frac{1}{2}{\rho}^2} {\rho}^{|m|} L_{n_{\rho}}^{(|m|)}({\rho}^2) \, e^{im\phi} \, \psi_{n_z}(z) \,,
    \ee
    where $L_n^{(\alpha)}$ are the Laguerre polynomials with the proper renormalization.
    
    \item[(iii)]
    \textit{Spherical coordinates}.
    The spherical basis is $|n_t,n_R,\ell,m\ra$, with $n_t,n_R,\ell \in \bbN$, $n_t \in \bbR$ and $m \in -\ell, \dots ,\ell$. 
    It is associated with  the eigenvalue $E = (n_t + 1/2) - (2n_R + \ell +3/2)$ with eigenbasis
    \be
        \Psi_{n_t,n_R,\ell,m}(t,R,\theta,\phi) =
        \psi_{n_t}(t) \, 
        (-1)^{n_R} \sqrt{\frac{n_R!}{\Gamma(n_R + \ell + 3/2)}} \, e^{-\frac{1}{2}R^2} R^{\ell} L_{n_R}^{(\ell + 1/2)}(R^2) \, Y^m_{\ell}(\theta,\phi) \,,
    \ee
    where $Y^m_{\ell}$ are the spherical harmonics.
    
    \item[(iv)]
    \textit{Hyperbolic coordinates}.
    The hyperbolic basis is $|n_r,\mu,\ell,m\ra$, with $n_r,\ell \in \bbN$, $m \in -\ell, \dots ,\ell$ and $n_t,\mu \in \bbR$. 
    It is associated with  the eigenvalue $E = 2n_r + i\mu +1$ with eigenbasis
    \be
        \Psi_{n_r,\mu,\ell,m}(r,\eta,\theta,\phi) =
        (-1)^{n_r} \sqrt{\frac{n_r!}{\Gamma(n_r + \mu + 1/2)}} \,
        r^{\mu-1} e^{-\frac{1}{2}r^2} \, L^{(\mu)}_{n_r}(r^2) \,
        \frac{1}{\sinh\eta} \, Q_{\ell}^{i\mu}(\coth\eta) \,
        Y^m_{\ell}(\theta,\phi) \,,
        \label{HyperbolicBasis}
    \ee
    where $Q^{\alpha}_{\lambda}$ are the Legendre function with the proper normalisation constant
\end{itemize}
Note that by our choice of hyperbolic coordinates, we assumed the energy to be positive. This amounts to requiring that the solutions $\Psi$ are timelike. \\ 
For completeness, we show the details of the derivation of the eigenbasis in the hyperbolic coordinates. 
Let us consider the change of coordinates:
\be
    \left\{\,
    \begin{aligned}
        & t = r\cosh\eta \\
        & x = r\sinh\eta \sin\theta \cos\phi \\
        & y = r\sinh\eta \sin\theta \sin\phi \\
        & z = r\sinh\eta \cos\theta 
    \end{aligned}
    \right.
    \qquad \text{ with } \quad
    \left[\,
    \begin{aligned}
        & \eta \in \bbR \\
        & \theta \in [0,\pi] \\
        & \phi \in [0,2\pi] 
    \end{aligned}
    \right.
    \label{MinkowskianHyperbolic_ChangeCoordinates}
\ee
In these coordinates the Laplacian writes
\be
    \begin{aligned}
        \Delta 
        & =
        \partial_t^2 - \partial_x^2 - \partial_y^2 - \partial_z^2 \\
        & =
        \frac{1}{r^3} \partial_r(r^3 \, \partial_r \,\,\,)
        - \frac{1}{r^2 \sinh^2\eta} \partial_{\eta}(\sinh^2\eta \, \partial_{\eta} \,\,\,)
        - \frac{1}{r^2 \sinh^2\eta \sin\theta} \partial_{\theta}(\sin\theta \, \partial_{\theta} \,\,\,)
        - \frac{1}{r^2 \sinh^2\eta\sin^2\theta} \partial_{\phi}^2 \,.
    \end{aligned}
    \label{Laplacian}
\ee
We are looking for a separable solution for the above differential equation, namely of the type $\Psi(r,\eta,\theta,\phi) = u_r(r) \, u_{\eta}(\eta) \, u_{\theta}(\theta) \, u_{\phi}(\phi)$. 
First, we note that the non-radial part of the differential equation is exactly the Laplacian equation for the Casimir $C_1$ \eqref{CasimirC1_Laplace}. The timelike solutions are the eigenfunctions with eigenvalue $-1-\mu^2$:
\be
    \big(C_1 - (1+\mu^2)\big) \, (u_{\eta}(\eta) \, u_{\theta}(\theta) \, u_{\phi}(\phi)) = 0 \,.
\ee
Therefore, the rotational part gives the usual spherical harmonics $u_{\theta}(\theta) \, u_{\phi}(\phi) = Y^m_{\ell}(\theta,\phi)$.
While the radial and the hyperbolic contribution instead satisfy the equations
\be
    \begin{aligned}
        &
        \partial_{\eta}^2 \, u_{\eta} + 
        2\coth\eta \, \partial_{\eta} \, u_{\eta} -
        \frac{\ell(\ell+1)}{\sinh^2\eta} \, u_{\eta} = -
        (1+\mu^2) \, u_{\eta} \,,\\
        &
        \partial_r^2 \, u_r + \frac{3}{r} \partial_r \, u_r +
        \frac{1}{r^2} (1+\mu^2) \, u_r - r^2 \, u_r = -2E \, u_r \,.
    \end{aligned}
\ee
Two solutions for the hyperbolic part is given by the Legendre functions of first or second kind. According to \cite{Genest2013}, we take the solution of the second kind, given by:
\be
    u_{\eta}(\eta) = \frac{1}{\sinh\eta} \, Q_{\ell}^{i\mu}(\coth\eta) \,. 
\ee
A solution for the radial part is instead given in terms of the generalized Laguerre polynomials:
\be
    u_r(r) = r^{i\mu-1} e^{-\frac{1}{2}r^2} \, L^{(i\mu)}_{n_r}(r^2) \,,
\ee
with $E = 2n_r+i\mu+1$.
The general (non-normalized) eigenfunction of the four dimensional isotropic Lorentzian harmonic oscillator in the hyperbolic basis is thus given by 
\be
    \Psi_{n_r,\mu,\ell,m}(r,\eta,\theta,\phi) =
    r^{i\mu-1} e^{-\frac{1}{2}r^2} \, L^{(i\mu)}_{n_r}(r^2) \,
    \frac{1}{\sinh\eta} \, Q_{\ell}^{i\mu}(\coth\eta) \,
    Y^m_{\ell}(\theta,\phi) \,.
\ee
The eigenfunction in the spherical coordinates in terms of the eigenfunction in the Minkowskian coordinates \cite{hyperbolic_laplacian} writes
\be
    \begin{aligned}
        \Psi_{n_t,n_R,\ell,m}(t,R,\theta,\phi) =
        \sum_{n_{\rho},n_x,n_y,n_z} \, &
        \frac{i^{m + |m|} (-1)^{\tn_x+n_{\xi}} (\sigma_m i)^{n_y}}{2^{(1-\delta_{m,0})/2}} \,
        \\
        &
        \cC^{\frac{1+|m|}{2}, \frac{1}{4} + \frac{q_z}{2}, \frac{\ell}{2} + \frac{3}{4}}_{n_{\rho}, \tn_z, n_R} \,
        \cC^{\frac{1}{4} + \frac{q_x}{2}, \frac{1}{4} + \frac{q_y}{2}, \frac{1+|m|}{2}}_{\tn_x, \tn_y, n_{\rho}} \,
        \Psi_{n_t,n_x,n_y,n_z}(t,x,y,z) \,.
    \end{aligned}
\ee
Here $\sigma_m=sign(m)$ and $\cC^{\nu_1, \nu_2, \nu_{12}}_{n_1, n_2, n_{12}}$ is the $\su(1,1)$ Clebsh-Gordan coefficient defined as
\be
    \begin{aligned}
        \cC^{\nu_1, \nu_2, \nu_{12}}_{n_1, n_2, n_{12}} 
        = \,\, &
        \Big(\frac{(2\nu_1)_{n_1} (2\nu_2)_{n_2} (2\nu_1)_{n_1+n_2-n_{12}}}{n_1! n_2! n_{12}! (n_1+n_2-n_{12})! (2\nu_2)_{n_1+n_2-n_{12}} (2\nu_{12})_{n_{12}} (\nu_1+\nu_2+\nu_{12}-1)_{n_1+n_2-n_{12}}}\Big)^{1/2} 
        \\
        &
        (n_1+n_2)! \, 
        _{3}F_{2}
        \left(
        \begin{aligned}
            -n_1,\, -n_1-n_2 & +n_{12},\, \nu_1+\nu_2+\nu_{12} -1 \\ 
            &
            2\nu_1,\, -n_1-n_2
        \end{aligned}
        ; 1
        \right)
    \end{aligned}
    \label{ClebshGordan_su(1,1)}
\ee
where $_{m}F_n$ is the generalized hypergeometric function.

\paragraph{Realization of $\su(1,1)$ and $\so(1,3)$ algebras.}
Let us first remind the $\su(1,1)$ structure. Let $J_0,J_{\pm}$ be the $\su(1,1)$ generators obeying the commutation relations:
\be
    [J_0 \,,\, J_{\pm}] = \pm J_{\pm}
    \,\,,\quad
    [J_+ \,,\, J_-] = -2 J_0 \,,
    \label{su(1,1)Commutators}
\ee
with casimir operator $Q = J_0^2 - J_+J_- - J_0$. The $\su(1,1)$ irreducible representations are labelled by an integer $n \in \bbN$ and a real number $\nu \in \bbR$, where the action of the generators on it yields
\be
    \begin{aligned}
        &
        J_0 \, |\nu;n\ra = (\nu + n) \, |\nu;n\ra \,,\\
        &
        J_+ \, |\nu;n\ra = \sqrt{(n + 1)(n + \nu)} \, |\nu;n+1\ra \,,\\
        &
        J_- \, |\nu;n\ra = \sqrt{n(n + \nu-1)} \, |\nu;n-1\ra \,,\\
        &
        Q \, |\nu;n\ra = \nu(\nu-1) \, |\nu;n\ra \,. 
    \end{aligned}
    \label{su(1,1)Action}
\ee
In \cite{Genest2013} it was pointed that each one dimensional harmonic oscillator can be mapped into an $\su(1,1)$ sub-algebra. For the spacelike oscillators in the $\{x,y,z\}$ directions, the map between each of them and the $\su(1,1)$ basis is provided by $|n_a\ra \cong |1/4+q_a/2 \,;\, \tn_a\ra$, with $a = x,y,z$ and $n_a = 2\tn_a + q_a$. We show below how to derive such map for the harmonic oscillator in the $t$ direction. 
To this scope we use Dirac's work \cite{Dirac}; from this, we know that infinite dimensional representations of the Lorentz group are realized on homogeneous polynomials on Minkowski space with coordinates $\xi_{\mu}$, where the Lorentzian signature is implemented as a negative power of the time coordinate. Such polynomials can be recast as a four dimensional Lorentzian harmonic oscillator with constant energy\footnote{The request of constant energy for the harmonic oscillator is equivalent to the homogeneity condition of the polynomials.}. 
From Dirac's construction we have the relation \eqref{Map_OperatorCoordinates} between the harmonic oscillator ladder operators, the coordinates on Minkowski space $\xi_{\mu}$ and the harmonic oscillator coordinates $x_{\mu}$.
In order to derive the map between the quantum oscillator and the $\su(1,1)$ irreducible representations, we first consider the realization of the $\su(1,1)$ generators in the harmonic oscillator basis
\be
    J_0 = \frac{1}{2}(\adag_0 a_0 + 1/2) 
    \,\,,\quad
    J_+ = \frac{1}{2} (\adag_0)^2
    \,\,,\quad
    J_- = \frac{1}{2} (a_0)^2 \,.
    \label{su(1,1)-HarmonicOscillator}
\ee
The $\su(1,1)$ irreducible representation basis is diagonal with respect the operators $J_0,Q$, whose action is
\be
    \begin{aligned}
        J_0 \, |\nu;n\ra = (\nu + n) \, |\nu;n\ra 
        & \,\, := \,\,
        \frac{1}{2}(\adag_0 a_0 + 1/2) \, |n_t\ra =
        \frac{1}{2}(n_t + 1/2) \, |n_t\ra \,\\
        Q \, |\nu;n\ra = \nu(\nu-1) \, |\nu;n\ra
        & \,\, := \,\, 
        \Big(\frac{1}{4}(\adag_0 a_0 + 1/2)^2 -
        \frac{1}{4} (\adag_0)^2 (a_0)^2 -
        \frac{1}{2}(\adag_0 a_0 + 1/2) \Big) \, |n_t\ra =
        -\frac{3}{16} \, |n_t\ra \,. 
    \end{aligned}
\ee
Solving the system of equations $n+\nu = n_t/2+1/4$ and $\nu(\nu-1) = -3/16$, one derives the map
\be
    |\nu;n\ra \cong |1/4+q_t/2 \,;\, \tn_t\ra \,.
\ee
Notice that this map relates the $\su(1,1)$ sub-algebra characterized by $\nu = \{1/4,3/4\}$ to the time oriented harmonic oscillator whose real contribution $\on_t$ is restricted to the set of even numbers.

\paragraph{Hyperbolic vs Minkowskian basis.}
We would like now to obtain the basis of the four dimensional harmonic oscillator in the hyperbolic coordinates, as the tensor product between the three dimensional harmonic oscillator (in the spherical basis $|n_R,\ell,m\ra \cong |\ell/2+3/4 \,;\, n_R\ra$) and the time oriented harmonic oscillator $|n_t\ra \cong |1/4+q_t/2\,;\, \tn_t\ra$. Both can be mapped to an $\su(1,1)$ basis, thus, following again \cite{Genest2013}, we use the $\su(1,1)$ Clebsh-Gordan coefficients to derive their tensor product.
Concretely, we aim to obtain the relation
\be
    |n_r,\mu;\ell,m\ra = \sum_{n_t,n_R} \, 
    \la n_t,n_R,\ell,m|n_r,\mu;\ell,m\ra \, |n_t,n_R,\ell,m\ra \,,
    \label{ClebshGordanRelation}
\ee
for states with the same energy given by:
\be
    2n_r + i\mu + 1 = (n_t + 1/2) - (2n_R + \ell + 3/2) \,,
    \label{ConstantEnergyCondition}
\ee
along with the $\ell,m$ parameters which label the three dimensional rotational part that remains unchanged from the spherical to the hyperbolic basis.
We thus have to express the basis in the hyperbolic coordinates (written as an $\su(1,1)$ basis $|\nu;n\ra$) as a combination of the tensor product 
\be
    |n_t\ra \ot |n_R,\ell,m\ra \cong 
    |1/4+q_t/2\,;\, \tn_t\ra \ot |\ell/2+3/4 \,;\, n_R\ra \,.
\ee
To this aim, we write the generators of the tensor product basis as the difference\footnote{The minus sign of the linear combination reflects the Lorentzian signature.} of the generators of the two initial basis
\be
    J_0 = J_0^{t} \ot 1 - 1 \ot J_0^{xyz}
    \,\,,\quad
    J_{\pm} = J_{\pm}^{t} \ot 1 - 1 \ot J_{\mp}^{xyz} \,.
\ee
It is straightforward to check that each set of generators $\{J_0,J_{\pm}\}$, $\{J_0^t,J_{\pm}^t\}$ and $\{J_0^{xyz},J_{\pm}^{xyz}\}$ satisfy the commutation relations \eqref{su(1,1)Commutators} with respective Casimir operators $Q,Q^t,Q^{xyz}$ and actions \eqref{su(1,1)Action} on the respective basis. 
The tensor product basis $|n;\nu\ra = |n_r,\mu;\ell,m\ra$ is completely determined by requiring it to be diagonal with respect to the tensor product operators $J_0$ and $Q$. Using \eqref{ConstantEnergyCondition}, the first gives the condition
\be
    n + \nu = 
    \frac{1}{2}(n_t+1/2) - \frac{1}{2}(2n_R + \ell + 3/2) =
    n_r + \frac{1}{2}(i\mu + 1) \,.
    \label{J_0Condition}
\ee
Moreover, we demand the solutions $|n_r,\mu;\ell,m\ra$ to be a basis for the principal series of the irreducible representation of the Lorentz group.
We recall that the Lorentz algebra with generators $L_a,N_a$ has two Casimir operators \cite{BarrettCrane:1997BarrettCrane}:
\be
    C_1 = J^2 - N^2
    \,\,,\quad
    C_2 = 2 J \cdot N \,,
\ee
with respective eigenvalues of the principal series of the irreducible representation $C_1 = j^2 - \mu^2 - 1$ and $C_2 = 2j\mu$. 
Using \eqref{Lorentz-HarmonicOscillator}, one can directly check that the Casimir $C_2$ automatically vanishes. According to \cite{Baez_1998}, the vanishing of $C_2 = 2j\mu$ means that $|n_r,\mu;\ell,m\ra$ is a basis for the balanced representations of the Lorentz group. In particular, since we assumed the energy of the harmonic oscillator to be positive, the Casimir $C_1$ has to be negative, and thus $k=0$ which implies $C_1 = -1-\mu^2$. Hence, the basis $|n_r,\mu;\ell,m\ra$ provides a quantization for the spacelike bivectors.
Using \eqref{Lorentz-HarmonicOscillator}, from a direct computation one can check that
\be
    Q = \frac{1}{4} C_1 \,.
\ee
By enforcing this equation on the basis $|n_r,\mu;\ell,m\ra$ together with   \eqref{J_0Condition}, we end up with the system of equations
\be
    \left\{\,
    \begin{aligned}
        & n + \nu = n_r + \frac{1}{2}(i\mu+1) \,,\\
        & \nu(\nu-1) = -\frac{1}{4}(1+\mu^2) \,,
    \end{aligned}
    \right.
    \quad \Rightarrow \quad
    n = n_r
    \,\,,\quad
    \nu = \frac{1}{2}(1+i\mu) \,.
\ee
This gives the Clebsh-Gordan coefficients in \eqref{ClebshGordanRelation} for the harmonic oscillator from the spherical coordinates to the hyperbolic coordinates (the timelike principal series of the irreducible representations of the Lorentz group):
\be
    \la n_t,n_R,\ell,m \,|\, n_r,\mu;\ell,m \ra = e^{i\varphi} \, \cC^{1/4+q_t/2,\, \ell/2+3/4,\, (1+i\mu)/2}_{\tn_t+\on_t/2,\, n_R,\, n_r} \,,
\ee
where $e^{i\varphi}$ is a phase factor that can be fixed by requiring that the expansion \eqref{ClebshGordanRelation} holds also for the solutions \eqref{HyperbolicBasis}. The overall map of the four dimensional oscillator, between the Minkowskian basis and the hyperbolic basis is thus
\be
    \Psi_{n_r,\mu,\ell,m}(r,\eta,\theta,\phi) =
    \sum_{n_t,n_x,n_y,n_z} \, 
    \cC^{n_t,\, n_x,\, n_y,\, n_z}_{n_r,\, \mu,\, \ell,\, m} \, \Psi_{n_t,n_x,n_y,n_z}(t,x,y,z) \,,
\ee
with
\begin{align}
    \cC^{n_t,\, n_x,\, n_y,\, n_z}_{n_r,\, \mu,\, \ell,\, m} 
    & :=
    \la n_t,n_x,n_y,n_z | n_r,\mu;\ell,m \ra 
    \nonumber \\
    & =
    \sum_{n_R, n_{\rho}} \,
    \frac{i^{m + |m|} (-1)^{\tn_x+n_{\xi}} (\sigma_m i)^{n_y}}{2^{(1-\delta_{m,0})/2}} \,
    e^{i\varphi} \,\,
    \cC^{\frac{1+|m|}{2}, \frac{1}{4} + \frac{q_z}{2}, \frac{\ell}{2} + \frac{3}{4}}_{n_{\rho}, \tn_z, n_R} \,
    \cC^{\frac{1}{4} + \frac{q_x}{2}, \frac{1}{4} + \frac{q_y}{2}, \frac{1+|m|}{2}}_{\tn_x, \tn_y, n_{\rho}} \,
    \cC^{1/4+q_t/2,\, \ell/2+3/4,\, (1+i\mu)/2}_{\tn_t+\on_t/2,\, n_R,\, n_r} \,.
\end{align}

\newpage
\bibliography{biblio.bib}
\bibliographystyle{jk.bst}
\end{document}

%% file: QuantumTriangle.tex
\begin{tikzpicture}[scale=1 , rotate around x=0]
    \coordinate (v1) at (0,0,3);
    \coordinate (v2) at (0,0,-3);
    \coordinate (v3) at (-5,0,0);
    
    \coordinate (c) at ($ 1/3*(v1) + 1/3*(v2) + 1/3*(v3) $);
    
    \draw[fill]
    (v1) circle [radius=0.03]
    (v2) circle [radius=0.03]
    (v3) circle [radius=0.03];
    
    \draw[-> , thick]
    ($ 0.95*(v3) + 0.05*(v1) $) -- node[below left] {$e_1$} ($ 0.05*(v3) + 0.95*(v1) $); 
    \draw[-> , thick] 
    ($ 0.95*(v1) + 0.05*(v2) $) -- node[below right] {$e_2$} ($ 0.05*(v1) + 0.95*(v2) $); 
    \draw[-> , thick , dashed] 
    ($ 0.95*(v3) + 0.05*(v2) $) -- node[above] {$e_3$} ($ 0.05*(v3) + 0.95*(v2) $); 
    
    \path[nearly transparent , fill=darklavander] 
    (v1) -- ($ 0.5*(v3) + 0.5*(v1) $) to [out=35 , in=170] ($ 0.5*(v2) + 0.5*(v1) $) -- (v1);
    
    \node[below] 
    at (c) {$b = e_1 \wedge e_2$};
    \node[right] 
    at (1.5,0,0) {$e_1+e_2+e_3=0$};
\end{tikzpicture}

%% file: Tetrahedron.tex
\begin{tikzpicture}
    \coordinate (x) at (0.1,0);
    \coordinate (y) at (0,0.1);
    
    \coordinate (w1) at (12,0);
    \coordinate (w2) at (9,0);
    \coordinate (w3) at (10.5,2.5);
    \coordinate (w4) at (13.5,2.5);
    \coordinate (w5) at (7.5,2.5);
    \coordinate (w6) at (10.5,-2.5);

    \node at ($ 1/3*(w1) + 1/3*(w2) + 1/3*(w3) $) {$t_1$};
    \node at ($ 1/3*(w1) + 1/3*(w3) + 1/3*(w4) $) {$t_2$};
    \node at ($ 1/3*(w1) + 1/3*(w2) + 1/3*(w6) $) {$t_3$};
    \node at ($ 1/3*(w2) + 1/3*(w3) + 1/3*(w5) $) {$t_4$};
    
    \draw[fill] (w1) circle  [radius=0.03];
    \draw[fill] (w2) circle  [radius=0.03];
    \draw[fill] (w3) circle  [radius=0.03];
    \draw[fill] (w4) circle  [radius=0.03];
    \draw[fill] (w5) circle  [radius=0.03];
    \draw[fill] (w6) circle  [radius=0.03];
    
    \draw[-> , thick] ($ (w1) - (x) $) -- node[above, scale=0.9] {$e_{1}$} ($ (w2) + (x) $);
    \draw[-> , thick] ($ (w2) + (x) + (y) $) -- node[left, scale=0.9] {$e_{2}$} ($ (w3) - (x) - (y) $);
    \draw[-> , thick] ($ (w3) + (x) - (y) $) -- node[right, scale=0.9] {$e_{3}$} ($ (w1) - (x) + (y) $);
    \draw[-> , thick] ($ (w3) + (x) $) -- node[above, scale=0.9] {$e_{4}$} ($ (w4) - (x) $);
    \draw[-> , thick] ($ (w4) - (x) - (y) $) -- node[right, scale=0.9] {$e_{5}$} ($ (w1) + (x) + (y) $);
    \draw[<- , thick] ($ (w1) - (x) - (y) $) -- node[right, scale=0.9] {$e_{5}$} ($ (w6) + (x) + (y) $);
    \draw[-> , thick] ($ (w6) - (x) + (y) $) -- node[left, scale=0.9] {$e_{6}$} ($ (w2) + (x) - (y) $);
    \draw[<- , thick] ($ (w2) - (x) + (y) $) -- node[left, scale=0.9] {$e_{6}$} ($ (w5) + (x) - (y) $);
    \draw[<- , thick] ($ (w5) + (x) $) -- node[above, scale=0.9] {$e_{4}$} ($ (w3) - (x) $);
    
    \draw[<-> , dotted , thick] ($ (w5) - 7.5*(y) $) to [out=-100 , in=160] ($ (w6) - 10*(x) + 7.5*(y) $);
    \draw[<-> , dotted , thick] ($ (w4) - 7.5*(y) $) to [out=-80 , in=20] ($ (w6) + 10*(x) + 7.5*(y) $);
    \draw[<-> , dotted , thick] ($ (w5) + 7.5*(x) + 4.5*(y) $) to [out=35 , in=145] ($ (w4) - 7.5*(x) + 4.5*(y) $);
\end{tikzpicture}

%% file: QuantumTetrahedron.tex
\begin{tikzpicture}[scale=0.7 , rotate around y=5]
    \coordinate (v1) at (0,0,4);
    \coordinate (v2) at (7,0,0);
    \coordinate (v3) at (0,0,-4);
    \coordinate (v4) at ($ 1/3*(v1) + 1/3*(v2) + 1/3*(v3) + (0,6.5,0) $);
    
    \coordinate (c) at ($ 1/4*(v1) + 1/4*(v2) + 1/4*(v3) + 1/4*(v4)$);
    
    \draw[fill] 
    (v1) circle [radius=0.03]
    (v2) circle [radius=0.03]
    (v3) circle [radius=0.03]
    (v4) circle [radius=0.03];
    
    \draw[-> , thick , dashed] 
    ($ 0.95*(v2) + 0.05*(v3) $) -- node[above] {$e_2$} ($ 0.95*(v3) + 0.05*(v2) $);
    \draw[-> , thick , dashed] 
    ($ 0.95*(v3) + 0.05*(v4) $) -- node[below right] {$e_4$} ($ 0.95*(v4) + 0.05*(v3) $);
    \draw[-> , thick , dashed] 
    ($ 0.95*(v4) + 0.05*(v2) $) -- node[right] {$e_6$} ($ 0.95*(v2) + 0.05*(v4) $);
    
    \draw[-> , thick]
    ($ 0.95*(v1) + 0.05*(v2) $) -- node[below right] {$e_1$} ($ 0.95*(v2) + 0.05*(v1) $);
    \draw[<- , thick]
    ($ 0.95*(v1) + 0.05*(v3) $) -- node[right] {$e_3$} ($ 0.95*(v3) + 0.05*(v1) $);
    \draw[<- , thick]
    ($ 0.95*(v1) + 0.05*(v4) $) -- node[above left] {$e_5$} ($ 0.95*(v4) + 0.05*(v1) $);
    
    \node[right] at (7.5,3,0) {$
    \begin{aligned}
        & e_1 + e_2 + e_3 = 0 \\
        & e_1 + e_5 + e_6=0 \\
        & e_4 + e_5 + e_3 =0
    \end{aligned}$};
\end{tikzpicture}

%% file: 4Simplex.tex
\begin{tikzpicture}[scale=0.75]
    \coordinate (v11) at (0,-5);
    \coordinate (v12) at (2,-5);
    \coordinate (v13) at (1,-6.73);
    \coordinate (v142) at (1,-3.27);
    \coordinate (v143) at (-1,-6.73);
    \coordinate (v144) at (3,-6.73);
    
    \coordinate (c11) at ($1/3*(v11) + 1/3*(v12) + 1/3*(v13)$);
    \coordinate (c12) at ($1/3*(v11) + 1/3*(v12) + 1/3*(v142)$);
    \coordinate (c13) at ($1/3*(v11) + 1/3*(v13) + 1/3*(v143)$);
    \coordinate (c14) at ($1/3*(v12) + 1/3*(v13) + 1/3*(v144)$);
    
    \node[above left , scale=1.2] at (v11) {$\tau_1$};
    
    \draw[- , thick] (v11) -- (v12) -- (v13) -- cycle;
    \draw[- , thick] (v12) -- (v144) -- (v13);
    \draw[- , thick] (v13) -- (v143) -- (v11);
    \draw[- , thick] (v11) -- (v142) -- (v12);
    
    \node[violet] at (c11) {$t_{1;4}$};
    \node[blue] at (c12) {$t_{1;2}$};
    \node[green] at (c13) {$t_{1;3}$};
    \node[red] at (c14) {$t_{1;1}$};
    
    \draw[decoration={markings,mark=at position 0.6 with {\arrow[scale=1.5,thick,>=stealth]{>}}},postaction={decorate}] (v12) -- (v11);
    \draw[decoration={markings,mark=at position 0.6 with {\arrow[scale=1.5,thick,>=stealth]{>}}},postaction={decorate}] (v12) -- (v13);
    \draw[decoration={markings,mark=at position 0.6 with {\arrow[scale=1.5,thick,>=stealth]{>}}},postaction={decorate}] (v12) -- (v142);
    \draw[decoration={markings,mark=at position 0.6 with {\arrow[scale=1.5,thick,>=stealth]{>}}},postaction={decorate}] (v12) -- (v144);
    \draw[decoration={markings,mark=at position 0.6 with {\arrow[scale=1.5,thick,>=stealth]{>}}},postaction={decorate}] (v13) -- (v144);
    \draw[decoration={markings,mark=at position 0.6 with {\arrow[scale=1.5,thick,>=stealth]{>}}},postaction={decorate}] (v11) -- (v142);
    \draw[decoration={markings,mark=at position 0.6 with {\arrow[scale=1.5,thick,>=stealth]{>}}},postaction={decorate}] (v13) -- (v143);
    \draw[decoration={markings,mark=at position 0.6 with {\arrow[scale=1.5,thick,>=stealth]{>}}},postaction={decorate}] (v11) -- (v143);
    \draw[decoration={markings,mark=at position 0.6 with {\arrow[scale=1.5,thick,>=stealth]{>}}},postaction={decorate}] (v11) -- (v13);

    \coordinate (v21) at (5,-9);
    \coordinate (v22) at (7,-9);
    \coordinate (v23) at (6,-7.23);
    \coordinate (v242) at (8,-7.23);
    \coordinate (v243) at (6,-10.73);
    \coordinate (v244) at (4,-7.23);
    
    \coordinate (c21) at ($1/3*(v21) + 1/3*(v22) + 1/3*(v23)$);
    \coordinate (c22) at ($1/3*(v22) + 1/3*(v23) + 1/3*(v242)$);
    \coordinate (c23) at ($1/3*(v21) + 1/3*(v22) + 1/3*(v243)$);
    \coordinate (c24) at ($1/3*(v21) + 1/3*(v23) + 1/3*(v244)$);
    
    \node[below right , scale=1.2] at (v22) {$\tau_2$};
    
    \draw[- , thick] (v21) -- (v22) -- (v23) -- cycle;
    \draw[- , thick] (v22) -- (v242) -- (v23);
    \draw[- , thick] (v22) -- (v243) -- (v21);
    \draw[- , thick] (v21) -- (v244) -- (v23);
    
    \node[teal] at (c21) {$t_{2;3}$};
    \node[magenta] at (c22) {$t_{2;1}$};
    \node[brown] at (c23) {$t_{2;2}$};
    \node[red] at (c24) {$t_{2;4}$};
    
    \draw[decoration={markings,mark=at position 0.6 with {\arrow[scale=1.5,thick,>=stealth]{>}}},postaction={decorate}] (v23) -- (v244);
    \draw[decoration={markings,mark=at position 0.6 with {\arrow[scale=1.5,thick,>=stealth]{>}}},postaction={decorate}] (v23) -- (v242);
    \draw[decoration={markings,mark=at position 0.6 with {\arrow[scale=1.5,thick,>=stealth]{>}}},postaction={decorate}] (v23) -- (v21);
    \draw[decoration={markings,mark=at position 0.6 with {\arrow[scale=1.5,thick,>=stealth]{>}}},postaction={decorate}] (v23) -- (v22);
    \draw[decoration={markings,mark=at position 0.6 with {\arrow[scale=1.5,thick,>=stealth]{>}}},postaction={decorate}] (v21) -- (v22);
    \draw[decoration={markings,mark=at position 0.6 with {\arrow[scale=1.5,thick,>=stealth]{>}}},postaction={decorate}] (v21) -- (v244);
    \draw[decoration={markings,mark=at position 0.6 with {\arrow[scale=1.5,thick,>=stealth]{>}}},postaction={decorate}] (v22) -- (v242);
    \draw[decoration={markings,mark=at position 0.6 with {\arrow[scale=1.5,thick,>=stealth]{>}}},postaction={decorate}] (v21) -- (v243);
    \draw[decoration={markings,mark=at position 0.6 with {\arrow[scale=1.5,thick,>=stealth]{>}}},postaction={decorate}] (v22) -- (v243);

    \coordinate (v31) at (0,0);
    \coordinate (v32) at (2,0);
    \coordinate (v33) at (1,1.73);
    \coordinate (v342) at (3,1.73);
    \coordinate (v343) at (1,-1.73);
    \coordinate (v344) at (-1,1.73);
    
    \coordinate (c31) at ($1/3*(v31) + 1/3*(v32) + 1/3*(v33)$);
    \coordinate (c32) at ($1/3*(v32) + 1/3*(v33) + 1/3*(v342)$);
    \coordinate (c33) at ($1/3*(v31) + 1/3*(v32) + 1/3*(v343)$);
    \coordinate (c34) at ($1/3*(v31) + 1/3*(v33) + 1/3*(v344)$);
    
    \node[below left , scale=1.2] at (v31) {$\tau_3$};
    
    \draw[- , thick] (v31) -- (v32) -- (v33) -- cycle;
    \draw[- , thick] (v32) -- (v342) -- (v33);
    \draw[- , thick] (v32) -- (v343) -- (v31);
    \draw[- , thick] (v31) -- (v344) -- (v33);
    
    \node[olive] at (c31) {$t_{3;2}$};
    \node[magenta] at (c32) {$t_{3;4}$};
    \node[blue] at (c33) {$t_{3;3}$};
    \node[orange] at (c34) {$t_{3;1}$};
    
    \draw[decoration={markings,mark=at position 0.6 with {\arrow[scale=1.5,thick,>=stealth]{>}}},postaction={decorate}] (v32) -- (v342);
    \draw[decoration={markings,mark=at position 0.6 with {\arrow[scale=1.5,thick,>=stealth]{>}}},postaction={decorate}] (v32) -- (v343);
    \draw[decoration={markings,mark=at position 0.6 with {\arrow[scale=1.5,thick,>=stealth]{>}}},postaction={decorate}] (v32) -- (v31);
    \draw[decoration={markings,mark=at position 0.6 with {\arrow[scale=1.5,thick,>=stealth]{>}}},postaction={decorate}] (v32) -- (v33);
    \draw[decoration={markings,mark=at position 0.6 with {\arrow[scale=1.5,thick,>=stealth]{>}}},postaction={decorate}] (v33) -- (v342);
    \draw[decoration={markings,mark=at position 0.6 with {\arrow[scale=1.5,thick,>=stealth]{>}}},postaction={decorate}] (v33) -- (v344);
    \draw[decoration={markings,mark=at position 0.6 with {\arrow[scale=1.5,thick,>=stealth]{>}}},postaction={decorate}] (v31) -- (v344);
    \draw[decoration={markings,mark=at position 0.6 with {\arrow[scale=1.5,thick,>=stealth]{>}}},postaction={decorate}] (v31) -- (v343);
    \draw[decoration={markings,mark=at position 0.6 with {\arrow[scale=1.5,thick,>=stealth]{>}}},postaction={decorate}] (v31) -- (v33);

    \coordinate (v41) at (-5,-9);
    \coordinate (v42) at (-3,-9);
    \coordinate (v43) at (-4,-7.23);
    \coordinate (v442) at (-2,-7.23);
    \coordinate (v443) at (-4,-10.73);
    \coordinate (v444) at (-6,-7.23);
    
    \coordinate (c41) at ($1/3*(v41) + 1/3*(v42) + 1/3*(v43)$);
    \coordinate (c42) at ($1/3*(v42) + 1/3*(v43) + 1/3*(v442)$);
    \coordinate (c43) at ($1/3*(v41) + 1/3*(v42) + 1/3*(v443)$);
    \coordinate (c44) at ($1/3*(v41) + 1/3*(v43) + 1/3*(v444)$);
    
    \node[below left , scale=1.2] at (v41) {$\tau_4$};
    
    \draw[- , thick] (v41) -- (v42) -- (v43) -- cycle;
    \draw[- , thick] (v42) -- (v442) -- (v43);
    \draw[- , thick] (v42) -- (v443) -- (v41);
    \draw[- , thick] (v41) -- (v444) -- (v43);
    
    \node[cyan] at (c41) {$t_{4;1}$};
    \node[green] at (c42) {$t_{4;2}$};
    \node[brown] at (c43) {$t_{4;3}$};
    \node[orange] at (c44) {$t_{4;4}$};
    
    \draw[decoration={markings,mark=at position 0.6 with {\arrow[scale=1.5,thick,>=stealth]{>}}},postaction={decorate}] (v43) -- (v442);
    \draw[decoration={markings,mark=at position 0.6 with {\arrow[scale=1.5,thick,>=stealth]{>}}},postaction={decorate}] (v42) -- (v442);
    \draw[decoration={markings,mark=at position 0.6 with {\arrow[scale=1.5,thick,>=stealth]{>}}},postaction={decorate}] (v43) -- (v42);
    \draw[decoration={markings,mark=at position 0.6 with {\arrow[scale=1.5,thick,>=stealth]{>}}},postaction={decorate}] (v43) -- (v41);
    \draw[decoration={markings,mark=at position 0.6 with {\arrow[scale=1.5,thick,>=stealth]{>}}},postaction={decorate}] (v42) -- (v41);
    \draw[decoration={markings,mark=at position 0.6 with {\arrow[scale=1.5,thick,>=stealth]{>}}},postaction={decorate}] (v43) -- (v444);
    \draw[decoration={markings,mark=at position 0.6 with {\arrow[scale=1.5,thick,>=stealth]{>}}},postaction={decorate}] (v41) -- (v444);
    \draw[decoration={markings,mark=at position 0.6 with {\arrow[scale=1.5,thick,>=stealth]{>}}},postaction={decorate}] (v41) -- (v443);
    \draw[decoration={markings,mark=at position 0.6 with {\arrow[scale=1.5,thick,>=stealth]{>}}},postaction={decorate}] (v42) -- (v443);

    \coordinate (v51) at (0,-11);
    \coordinate (v52) at (2,-11);
    \coordinate (v53) at (1,-12.73);
    \coordinate (v542) at (1,-9.27);
    \coordinate (v543) at (-1,-12.73);
    \coordinate (v544) at (3,-12.73);
    
    \coordinate (c51) at ($1/3*(v51) + 1/3*(v52) + 1/3*(v53)$);
    \coordinate (c52) at ($1/3*(v51) + 1/3*(v52) + 1/3*(v542)$);
    \coordinate (c53) at ($1/3*(v51) + 1/3*(v53) + 1/3*(v543)$);
    \coordinate (c54) at ($1/3*(v52) + 1/3*(v53) + 1/3*(v544)$);
    
    \node[above left , scale=1.2] at (v51) {$\tau_5$};
    
    \draw[- , thick] (v51) -- (v52) -- (v53) -- cycle;
    \draw[- , thick] (v52) -- (v544) -- (v53);
    \draw[- , thick] (v53) -- (v543) -- (v51);
    \draw[- , thick] (v51) -- (v542) -- (v52);
    
    \node[olive] at (c51) {$t_{5;3}$};
    \node[violet] at (c52) {$t_{5;1}$};
    \node[cyan] at (c53) {$t_{5;4}$};
    \node[teal] at (c54) {$t_{5;2}$};
    
    \draw[decoration={markings,mark=at position 0.6 with {\arrow[scale=1.5,thick,>=stealth]{>}}},postaction={decorate}] (v52) -- (v542);
    \draw[decoration={markings,mark=at position 0.6 with {\arrow[scale=1.5,thick,>=stealth]{>}}},postaction={decorate}] (v52) -- (v51);
    \draw[decoration={markings,mark=at position 0.6 with {\arrow[scale=1.5,thick,>=stealth]{>}}},postaction={decorate}] (v52) -- (v53);
    \draw[decoration={markings,mark=at position 0.6 with {\arrow[scale=1.5,thick,>=stealth]{>}}},postaction={decorate}] (v52) -- (v544);
    \draw[decoration={markings,mark=at position 0.6 with {\arrow[scale=1.5,thick,>=stealth]{>}}},postaction={decorate}] (v51) -- (v542);
    \draw[decoration={markings,mark=at position 0.6 with {\arrow[scale=1.5,thick,>=stealth]{>}}},postaction={decorate}] (v51) -- (v53);
    \draw[decoration={markings,mark=at position 0.6 with {\arrow[scale=1.5,thick,>=stealth]{>}}},postaction={decorate}] (v544) -- (v53);
    \draw[decoration={markings,mark=at position 0.6 with {\arrow[scale=1.5,thick,>=stealth]{>}}},postaction={decorate}] (v543) -- (v53);
    \draw[decoration={markings,mark=at position 0.6 with {\arrow[scale=1.5,thick,>=stealth]{>}}},postaction={decorate}] (v51) -- (v543);

    \draw[- , double distance=2pt , thick , dotted , blue] (v142) -- (v343);
    \draw[- , double distance=2pt , thick , dotted , green] (v143) -- (v442);
    \draw[- , double distance=2pt , thick , dotted , red] (v144) -- (v244);
    \draw[- , double distance=2pt , thick , dotted , magenta] (v242) to [out=30 , in=30] (v342);
    \draw[- , double distance=2pt , thick , dotted , orange] (v444) to [out=150 , in=150] (v344);
    \draw[- , double distance=2pt , thick , dotted , brown] (v443) to [out=-90 , in=180] (1,-15) to [out=0 , in=-90] (v243);
    \draw[- , double distance=2pt , thick , dotted , violet] (v542) -- (v13);
    \draw[- , double distance=2pt , thick , dotted , teal] (v544) to [out=-30 , in=-150] (v21);
    \draw[- , double distance=2pt , thick , dotted , cyan] (v543) to [out=-150 , in=-30] (v42);
    
    \draw[- , double distance=2pt , thick , dotted , olive] (v33) to [out=90 , in=90] ($ (v33) + (9,0) $) -- (10,-15) to [out=-90 , in=-90] (1,-15.5);
    \draw[- , double distance=2pt , thick , dotted , olive] (1,-14.5) -- (v53);

    \draw[fill] (v12) circle [radius=0.1];
    \draw[thick , fill=white] (v11) circle [radius=0.1];
    \draw[thick , fill=gray] (v13) circle [radius=0.1];
    \draw[thick , fill=lime] (v142) circle [radius=0.1];
    \draw[thick , fill=lime] (v143) circle [radius=0.1];
    \draw[thick , fill=lime] (v144) circle [radius=0.1];
    
    \draw[fill] (v23) circle [radius=0.1];
    \draw[thick , fill=gray] (v21) circle [radius=0.1];
    \draw[thick , fill=yellow] (v22) circle [radius=0.1];
    \draw[thick , fill=lime] (v242) circle [radius=0.1];
    \draw[thick , fill=lime] (v243) circle [radius=0.1];
    \draw[thick , fill=lime] (v244) circle [radius=0.1];
    
    \draw[fill] (v32) circle [radius=0.1];
    \draw[thick , fill=white] (v31) circle [radius=0.1];
    \draw[thick , fill=yellow] (v33) circle [radius=0.1];
    \draw[thick , fill=lime] (v342) circle [radius=0.1];
    \draw[thick , fill=lime] (v343) circle [radius=0.1];
    \draw[thick , fill=lime] (v344) circle [radius=0.1];
    
    \draw[thick , fill=white] (v43) circle [radius=0.1];
    \draw[thick , fill=yellow] (v41) circle [radius=0.1];
    \draw[thick , fill=gray] (v42) circle [radius=0.1];
    \draw[thick , fill=lime] (v442) circle [radius=0.1];
    \draw[thick , fill=lime] (v443) circle [radius=0.1];
    \draw[thick , fill=lime] (v444) circle [radius=0.1];
    
    \draw[fill] (v52) circle [radius=0.1];
    \draw[thick , fill=white] (v51) circle [radius=0.1];
    \draw[thick , fill=yellow] (v53) circle [radius=0.1];
    \draw[thick , fill=gray] (v542) circle [radius=0.1];
    \draw[thick , fill=gray] (v543) circle [radius=0.1];
    \draw[thick , fill=gray] (v544) circle [radius=0.1];

    \draw[fill] ($(v343) + (-4,1)$) circle [radius=0.1];
    \node[right] at ($(v343) + (-4,1)$) {$:v$};
    
    \draw[fill=white] ($(v343) + (-4,0)$) circle [radius=0.1];
    \node[right] at ($(v343) + (-4,0)$) {$:v'$};
    
    \draw[fill=gray] ($(v343) + (-4,-1)$) circle [radius=0.1];
    \node[right] at ($(v343) + (-4,-1)$) {$:v''$};
    
    \draw[thick , fill=yellow] ($(v343) + (-4,-2)$) circle [radius=0.1];
    \node[right] at ($(v343) + (-4,-2)$) {$:v'''$};
    
    \draw[thick , fill=lime] ($(v343) + (-4,-3)$) circle [radius=0.1];
    \node[right] at ($(v343) + (-4,-3)$) {$:v''''$};
\end{tikzpicture}

%% file: BarretCraneVertexTerm.tex
\begin{tikzpicture}
    \coordinate (a) at (-1,-1);
    \coordinate (b) at (-1,1);
    \coordinate (c) at (1,1);
    \coordinate (d) at (1,-1);
    
    \coordinate (i) at (6,0);
    
    \coordinate (a') at ($ (a) + (i) $);
    \coordinate (b') at ($ (b) + (i) $);
    \coordinate (c') at ($ (c) + 1.5*(i) $);
    \coordinate (d') at ($ (d) + 1.5*(i) $);
    
    \coordinate (x) at ($ 0.5*(a') + 0.5*(b') + (1,0) $);
    \coordinate (y) at ($ 0.5*(c') + 0.5*(d') - (1,0) $);
    
    \node[left] at ($ 0.25*(a') + 0.25*(b') + 0.25*(c) + 0.25*(d) $) {$=$};
    \node[right , scale=1.25] at ($ 0.25*(a') + 0.25*(b') + 0.25*(c) + 0.25*(d) $) {$\int \dd \mu'$};
    
    \draw[fill] 
    ($ 0.5*(a) + 0.5*(b) + 0.5*(c) + 0.5*(d) $) circle [radius=0.05]
    (x) circle [radius=0.05] node[below right] {$\mathbf{x}$}
    (y) circle [radius=0.05] node[below left] {$\mathbf{y}$};
    
    \draw[- , thick] 
    (a) node[above left] {$\mu_1$} -- (c) node[below right] {$\mu_3$}
    (b) node[below left] {$\mu_3$} -- (d) node[above right] {$\mu_4$}
    (a') node[above left] {$\mu_1$} -- (x)
    (b') node[below left] {$\mu_2$} -- (x)
    (c') node[below right] {$\mu_3$} -- (y)
    (d') node[above right] {$\mu_4$} -- (y)
    (x) -- node[above , scale=1] {$K_{\mu'}(\eta(\mathbf{x},\mathbf{y}))$}
    (y);
\end{tikzpicture}